%% file: main.tex
\documentclass[letterpaper,11pt]{article}
\usepackage{color}
\usepackage[colorlinks=true, citecolor=blue]{hyperref}
\usepackage{url}
\usepackage[utf8]{inputenc}
\usepackage{microtype}
\usepackage{xspace}
\usepackage{tkz-euclide}
\usepackage{tikz}
\usetikzlibrary{shapes,backgrounds,patterns,3d,calc}
\usepackage{xinttools}
\usepackage[margin=1in]{geometry}
\usepackage{enumerate}

\usepackage[noend]{algpseudocode}
\usepackage{algorithm}
\usepackage{fullpage}

\algnewcommand\Input{\item[\textbf{Input:}]}
\algnewcommand\Output{\item[\textbf{Output:}]}

\usepackage{geometry}
\usepackage{subcaption}
\let\proof\relax\let\endproof\relax
\usepackage{amsthm}
\usepackage{amsmath}
\usepackage{amsthm}
\usepackage{amssymb}
\usepackage{esint}
\usepackage[capitalise]{cleveref}

\newenvironment{innerproof}
 {\proof}
 {\endproof}

\newtheorem{thm}{Theorem}[section]
\newtheorem{lemma}[thm]{Lemma}

\newtheorem{definition}[thm]{Definition}
\newtheorem{corollary}[thm]{Corollary}
\newtheorem{question}[thm]{Question}
\newtheorem{claim}[thm]{Claim}
\newtheorem{observation}[thm]{Observation}

\newcommand{\defparproblem}[4]{
	\vspace{1mm}
	\noindent\fbox{
		\begin{minipage}{0.96\textwidth}
			\begin{tabular*}{\textwidth}{@{\extracolsep{\fill}}lr} \textsc{#1} & {\bf{Parameter:}} #3 \\ \end{tabular*}
			{\bf{Input:}} #2 \\
			{\bf{Task:}} #4
		\end{minipage}
	}
	\vspace{1mm}
}

\newcommand{\defproblem}[3]{
	\vspace{1mm}
	\noindent\fbox{
		\begin{minipage}{0.96\textwidth}
			\begin{tabular*}{\textwidth}{@{\extracolsep{\fill}}lr} \textsc{#1} &  \\ \end{tabular*}
			{\bf{Input:}} #2 \\
			{\bf{Task:}} #3
		\end{minipage}
	}
	\vspace{1mm}
}

\newcommand{\tw}{\mathbf{tw}}
\newcommand{\td}{\mathbf{td}}
\newcommand{\ed}{\mathbf{ed}}

\newcommand{\depth}{\mathrm{depth}}
\newcommand{\height}{\mathrm{height}}

\newcommand{\lp}{\mathrm{lp}}
\newcommand{\rp}{\mathrm{rp}}

\newcommand{\br}[1]{\left(#1\right)}

\newcommand{\Oh}{\mathcal{O}}
\newcommand{\hh}{\ensuremath{\mathcal{H}}}
\newcommand{\dhh}{\ensuremath{\mathcal{D(H)}}}
\newcommand{\rr}{\ensuremath{\mathcal{R}}}

\newcommand{\hhdn}[1][\hh]{\mathbf{dn}_{#1}}
\newcommand{\hhtw}[1][\hh]{\tw_{#1}}
\newcommand{\hhdepth}[1][\hh]{\ed_{#1}}
\newcommand{\hhtwfull}[1][\hh]{{#1}-treewidth}
\newcommand{\hhdepthfull}[1][\hh]{{#1}-elimination distance}
\newcommand{\hhtwdecomp}[1][\hh]{tree {#1}-decomposition}
\newcommand{\hhdepthdecomp}[1][\hh]{{#1}-elimination forest}

\newcommand{\customtw}[1]{\ensuremath{\tw_\mathsf{#1}}}
\newcommand{\customdepth}[1]{\ensuremath{\ed_\mathsf{#1}}}

\newcommand{\biptw}{\tw_\mathsf{bip}}
\newcommand{\bipdepth}{\ed_\mathsf{bip}}

\newcommand{\hsep}[1]{$(\hh,#1)$-separation\xspace}
\newcommand{\hsepk}{$(\hh,k)$-separation\xspace}
\newcommand{\gsep}{\ensuremath{G\, /\, T}\xspace}
\newcommand{\hhsepfind}{\textsc{$(\hh,h)$-separation finding}\xspace}
\newcommand{\hhsepfindres}{\textsc{Restricted $(\hh,h)$-separation finding}\xspace}

\ifdefined\DEBUG{}
\newcommand{\mic}[1]{{\color{blue}{#1}}}

\def\rem#1{{\marginpar{\raggedright\scriptsize #1}}}
\newcommand{\micr}[1]{\rem{\textcolor{blue}{\(\bullet \) #1}}}

\newcommand{\bmp}[1]{{\color{purple}{#1}}}
\newcommand{\bmpr}[1]{\rem{\textcolor{purple}{\(\bullet \) #1}}}

\newcommand{\jjh}[1]{{\color{orange}{#1}}}
\newcommand{\jjhr}[1]{\rem{\textcolor{orange}{\(\bullet \) #1}}}

\else
\newcommand{\mic}[1]{#1}
\newcommand{\bmp}[1]{#1}
\newcommand{\jjh}[1]{#1}
\newcommand{\micr}[1]{ }
\newcommand{\bmpr}[1]{ }
\newcommand{\jjhr}[1]{ }
\fi

\title{Vertex Deletion Parameterized by Elimination Distance \\ and Even Less\footnote{This project has received funding from the European Research Council (ERC) under the European Union's Horizon 2020 research and innovation programme (grant agreement No 803421, ReduceSearch). An extended abstract~\cite{JansenKW21a} of this work appeared in the proceedings of the 53rd Symposium on Theory of Computing, STOC 2021.}}

\author{Bart M. P. Jansen\footnote{Address: \texttt{b.m.p.jansen@tue.nl}}  \\ Eindhoven University of~Technology
\and Jari J. H. de Kroon\footnote{Address: \texttt{j.j.h.d.kroon@tue.nl}} \\ Eindhoven University of~Technology
\and Micha{\l} W{\l}odarczyk\footnote{Address: \texttt{m.wlodarczyk@tue.nl}} \\ Eindhoven University of~Technology}
\date{}

\begin{document}


\maketitle{}

\begin{abstract}
We study the parameterized complexity of various classic vertex-deletion problems such as \textsc{Odd cycle transversal}, \textsc{Vertex planarization}, and \textsc{Chordal vertex deletion} under {hybrid} {parameterizations}.
Existing FPT algorithms for these problems either focus on the parameterization by solution size, detecting solutions of size~$k$ in time~$f(k) \cdot n^{\Oh(1)}$, or width parameterizations, finding arbitrarily large optimal solutions in time~$f(w) \cdot n^{\Oh(1)}$ for some width measure~$w$ like treewidth. We unify these lines of research by presenting FPT algorithms for parameterizations that can simultaneously be arbitrarily much smaller than the solution size and the treewidth.

The first class of parameterizations is based on the {notion} of \emph{elimination distance} of the input graph to the target graph class~$\hh$, which intuitively measures the number of rounds needed to obtain a graph in~$\hh$ by removing one vertex from each connected component in each round. The~second class of parameterizations consists of a relaxation of the notion of treewidth, allowing arbitrarily large bags that induce subgraphs belonging to the target class of the deletion problem as long as these subgraphs have small neighborhoods.
{Both kinds of parameterizations have been introduced recently and have already spawned several independent results.}

Our contribution is twofold. First, we present a~framework
for computing approximately optimal decompositions {related to these graph measures.}
Namely, if the cost of {an} optimal decomposition is $k$, we show how to find a~decomposition {of} cost $k^{\Oh(1)}$ {in time~$f(k) \cdot n^{\Oh(1)}$}. 
\mic{This is applicable to any class $\hh$ for which the corresponding vertex-deletion problem admits~a constant-factor approximation algorithm or an FPT algorithm paramaterized by the solution size.}
Secondly, we exploit the constructed decompositions for solving vertex-deletion problems by {extending} ideas from
algorithms using iterative compression {and the} {finite state property}.
For the three mentioned vertex-deletion problems, and all problems which can be formulated as hitting a finite set of \bmp{connected forbidden (a)~minors or (b)~(induced) subgraphs}, we obtain FPT algorithms with respect to {both} studied parameterizations. For example, we present an algorithm running in~{time} $n^{\Oh(1)} + 2^{k^{\Oh(1)}}\cdot(n+m)$ and polynomial space for \textsc{Odd cycle transversal} parameterized by the elimination distance~$k$ to the class of bipartite graphs.
\end{abstract}
\clearpage


{
  \hypersetup{linkcolor=black}
  \tableofcontents
}

\clearpage

\section{Introduction}

The field of parameterized algorithmics~\cite{CyganFKLMPPS15,DowneyF13} develops fixed-parameter tractable (FPT) algorithms to solve NP-hard problems exactly, which are provably efficient on inputs whose parameter value is small. \bmp{The purpose of this work is to unify two lines of research in parameterized algorithms for \bmp{vertex-deletion} problems that were previously mostly disjoint.} On the one hand, there are algorithms that work on a structural decomposition of the graph, whose running time scales exponentially with a graph-complexity measure but polynomially with the size of the graph. Examples of such algorithms include dynamic programming over a tree decomposition~\cite{BodlaenderCKN15,BodlaenderK08} (which forms a recursive decomposition by small separators), dynamic-programming algorithms based on cliquewidth~\cite{CourcelleMR00}, rankwidth~\cite{HlinenyOSG08,Oum17,OumS06}, and Booleanwidth~\cite{Bui-XuanTV11} (which are recursive decompositions of a graph by simply structured although not necessarily small separations). The second line of algorithms are those that work with the ``solution size'' as the parameter, whose running time scales exponentially with the solution size. Such algorithms take advantage of the properties of inputs that admit small solutions. Examples of the latter include the celebrated \emph{iterative compression} algorithm to find a minimum odd cycle transversal~\cite{ReedSV04} and algorithms to compute a minimum vertex-deletion set whose removal makes the graph chordal~\cite{CaoM16,Marx10}, interval~\cite{CaoM15}, or  planar~\cite{JansenLS14,Kawarabayashi09,MarxS12}. 

In this work we combine the best of both these lines of research, culminating in fixed-parameter tractable algorithms for parameterizations which can be simultaneously smaller than natural parameterizations by solution size and width measures like treewidth. To achieve this, we (1) \bmp{employ recently introduced} graph decompositions tailored to the optimization problem which will be solved using the decomposition, (2) develop fixed-parameter tractable algorithms to compute approximately optimal decompositions, and (3) show how to exploit the decompositions to obtain the desired hybrid FPT algorithms. We apply these ideas to well-studied graph problems which can be formulated in terms of {vertex deletion}: find a smallest vertex set whose removal ensures {that} the resulting graph belongs to a certain graph class. Vertex-deletion problems are among the most prominent problems studied in parameterized algorithmics~\cite{ChenLLOR08,FominLPSZ20,GuptaLL20,LiN20,LokshtanovR0Z20}. We develop new algorithms for \textsc{Odd cycle transversal}, \textsc{Vertex planarization}, \textsc{Chordal vertex deletion}, and (induced) \textsc{$H$-free deletion} for fixed {connected}~$H$, among others. 

To be able to state our results, we give a high-level description of the graph decompositions we employ;  formal definitions are postponed to Section~\ref{sec:preliminaries}. Each decomposition is targeted at a~specific graph {class}~$\hh$. We use two types of decompositions, corresponding to relaxed versions of treedepth~\cite{NesetrilM06} and treewidth~\cite{Bodlaender98,RobertsonS86}, respectively.

The first type of graph decompositions we employ are \emph{$\hh$-elimination forests}, which decompose graphs of bounded $\hh$-elimination distance. The $\hh$-elimination distance~$\hhdepth(G)$ of an undirected graph~$G$ is a graph parameter that was recently introduced~\cite{BulianD16,BulianD17,HolsKP20}, which admits a recursive definition similar to treedepth. If~$G$ is connected and belongs to~$\hh$, then~$\hhdepth(G) = 0$. If~$G$ is connected but does not belong to~$\hh$, then~$\hhdepth(G) = 1 + \min_{v \in V(G)} \hhdepth(G-v)$. If~$G$ has multiple connected components~$G_1, \ldots, G_p$, then~$\hhdepth(G) = \max_{i=1}^p \hhdepth(G_i)$. The process of eliminating vertices in the second case of the definition explains the name \emph{elimination distance}. The elimination process can be represented by a tree structure called an $\hh$-elimination forest, whose depth corresponds to~$\hhdepth(G)$. These decompositions can be used to obtain polynomial-space algorithms, similarly as for treedepth~\cite{BelbasiF19,FurerY17,PilipczukW18}.

\bmp{The second type of decompositions we employ are called \emph{tree $\hh$-decompositions}, which decompose graphs of bounded \hhtwfull{}. These decompositions are obtained by relaxing the definition of treewidth and were recently introduced by Eiben et al.~\cite{EibenGHK19}, building on similar hybrid parameterizations used in the context of solving SAT~\cite{GanianRS17a} and CSPs~\cite{GanianRS17d}.} A tree $\hh$-decomposition of a graph~$G$ is a~tree decomposition of~$G$, together with a set~$L \subseteq V(G)$ of base vertices. Base vertices are not allowed to occur in more than one bag, and the base vertices in a bag must induce a subgraph belonging to~$\hh$. \bmp{The connected components of~$G[L]$ are called the \emph{base components} of the decomposition.} The {width} of such a decomposition is defined as the maximum number of non-base vertices in any bag, minus one. A tree $\hh$-decomposition therefore represents a decomposition of a graph by small separators, into subgraphs which are either small or belong to~$\hh$. The minimum width of such a decomposition for~$G$ is the \hhtwfull{} of~$G$, denoted~$\hhtw(G)$.  We have~$\hhtw(G) \leq \hhdepth(G)$ for all graphs~$G$, hence the former is a potentially smaller parameter. However, working with this parameter will require exponential-space algorithms.
{We remark that both considered parameterizations are \bmp{interesting} mainly for the case where the class \hh{} has unbounded treewidth, e.g., $\hh \in \{\mathsf{bipartite},\mathsf{chordal},\mathsf{planar}\}$,
as otherwise $\hhtw(G)$ is comparable with $\tw(G)$ (see Lemma~\ref{lem:treewidth-of-hh}).
Therefore we do not study classes such as trees or outerplanar graphs.
}

We illustrate the key ideas of the graph decomposition for the case of \textsc{Odd cycle transversal} (OCT), which is the vertex-deletion problem which aims to reach a bipartite graph. Hence OCT corresponds to instantiations of the decompositions defined above for~$\hh$ {being} the class~$\mathsf{bip}$ of bipartite graphs. Observe that if~$G$ has an odd cycle transversal of~$k$ vertices, then~$\bipdepth(G) \leq k$. In the other direction, the value~$\bipdepth(G)$ may be arbitrarily much smaller than the size of a minimum OCT; see Figure~\ref{fig:octparameters}. At the same time, the value of~$\bipdepth(G)$ may be arbitrarily much smaller than the rankwidth (and hence treewidth, cliquewidth, and treedepth) of~$G$, since the $n \times n$ grid graph is bipartite but has rankwidth~$n-1$~\cite{Jelinek10}. Hence a time bound of~$f(\bipdepth(G)) \cdot |G|^{\Oh(1)}$ can be arbitrarily much better than bounds which are fixed-parameter tractable with respect to the size of an optimal odd cycle transversal or with respect to pure width measures of~$G$. Working with~$\bipdepth$ as the parameter for \textsc{Odd cycle transversal} therefore facilitates algorithms which simultaneously improve on the solution-size~\cite{ReedSV04} and width-based~\cite{LokshtanovMS18} algorithms for OCT. As a~sample of our results, we will prove that \textsc{Odd cycle transversal} can be solved in FPT-time and polynomial space parameterized by $\bipdepth$, and in FPT-time-and-space parameterized by $\biptw$. Using this terminology, we now proceed to a detailed description of our contributions.

\begin{figure}
    \centering
    \includegraphics[scale=0.9]{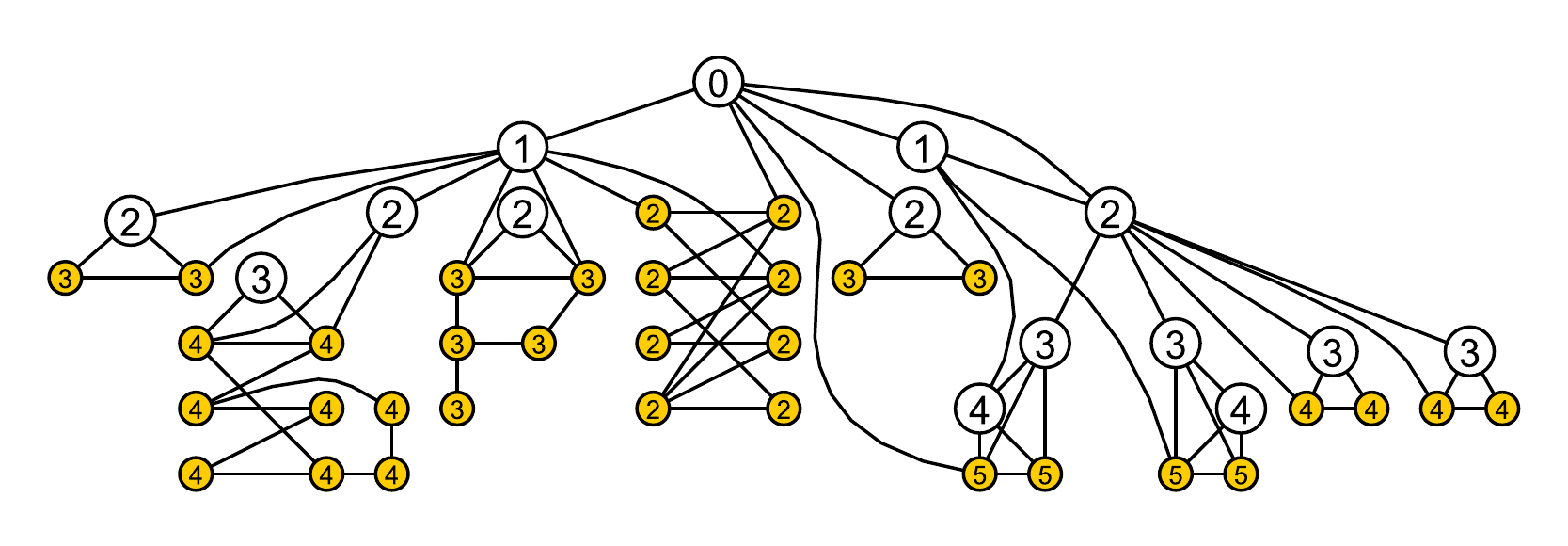}
    \caption{The vertex labels correspond to depth values in a  $\mathsf{bip}$-elimination forest of depth~5. By attaching triangles to vertices of depth at most three, the minimum size of an odd cycle transversal can increase boundlessly without increasing~$\hhdepth[\mathsf{bip}]$. The vertices in yellow are the base vertices of the decomposition and induce a bipartite graph.}
    \label{fig:octparameters}
\end{figure}

\paragraph{Results}
{Our work advances the budding theory of hybrid parameterizations. The resulting framework generalizes and improves upon several isolated results in the literature.} 
We present three types of {theorems}: FPT algorithms to compute approximate decompositions, FPT algorithms employing these decompositions, and hardness results providing limits of tractability in {this paradigm}.
While several results of the first and second kind have been known,
we obtain the first framework which allows to handle miscellaneous graph classes with a~unified set of techniques and to deliver algorithms  
with \bmp{a} running time of the form $2^{k^{\Oh(1)}} \cdot n^{\Oh(1)}$, where $k$ is either $\hhdepth(G)$ or $\hhtw(G)$, even if no decomposition is given in the input.

The following theorem\footnote{This decomposition theorem strengthens the statements from the extended abstract~\cite{JansenKW21a}, to be applicable whenever a suitable algorithm for the parameterization by solution size exists.} \mic{gives a simple overview of} our FPT-approximation algorithms for $\hhdepth$ and $\hhtw$.
\mic{The more general version, as well as detailed results for concrete graph classes,
can be found in \cref{sec:summary}.}

\begin{thm} \label{thm:decomposition:general}
Let $\hh$ be a hereditary union-closed class of graphs.
Suppose that $\hh$\textsc{-deletion} either admits a polynomial-time constant-factor approximation algorithm or an exact FPT algorithm parameterized by the solution size $s$, running in time $2^{s^{\Oh(1)}} \cdot n^{\Oh(1)}$.
There is an algorithm that, given an $n$-vertex graph $G$, computes in $2^{\hhtw(G)^{\Oh(1)}} \cdot n^{\Oh(1)}$ time a tree $\hh$-decomposition of $G$ of width~$\Oh(\hhtw(G)^5)$. Analogously, there is an algorithm that computes in~$2^{\hhdepth(G)^{\Oh(1)}} \cdot n^{\Oh(1)}$ time an $\hh$-elimination forest of $G$ of depth~$\Oh(\hhdepth(G)^3)$.
\end{thm}

The prerequisites of the theorem are satisfied for classes of, e.g., bipartite graphs, chordal graphs, (proper) interval graphs, and graphs excluding a finite family of connected (a) minors or (b) induced subgraphs (see \cref{table:decomposition} in \cref{sec:summary} for a longer list).
In fact, the theorem is applicable also when provided with an FPT algorithm for $\hh$\textsc{-deletion} with an arbitrary dependency on the parameter.
This is the case for classes of graphs excluding a finite family of connected topological minors, where we only know that the obtained dependency on $\hhtw(G) \,/\, \hhdepth(G)$ is a computable function.
What is more, for some graph classes we are able to deliver algorithms with better approximation guarantees or running in polynomial time.
Finally, we also show how to construct $\hh$-elimination forests using only polynomial space at the expense of slightly worse approximation, which is later leveraged in the applications.

\bmp{While FPT algorithms to compute $\hh$-elimination distance for minor-closed graph classes $\hh$ were known before~\cite{BulianD17} via an excluded-minor characterization (even to compute the exact distance), to the best of our knowledge the general approach of Theorem~\ref{thm:decomposition:general} is the first to be able to \mic{(approximately)} compute $\hh$-elimination distance for potentially \emph{dense} target classes $\hh$ such as chordal graphs. Concerning $\hh$-treewidth, we are only aware of a single prior result by Eiben et al.~\cite{EibenGHK19} which deals with classes~$\hh_c$ consisting of the graphs of rankwidth at most~$c$. To the best of our knowledge, our results for $\hhtw$ are the first to handle classes~$\hh$ of unbounded rankwidth\footnote{Developments following the publication of the extended abstract of this article are given in ``Related work''.}.}

\mic{When considering graph classes defined by forbidden (topological) minors or induced subgraphs, Theorem~\ref{thm:decomposition:general} works only when all the obstructions are given by connected graphs due to the requirement that $\hh$ must be union-closed.
We show however that this connectivity assumption can be dropped and the same result holds for any graph class defined by a finite set of forbidden (topological) minors or induced subgraphs.
In particular, we provide FPT-approximation algorithms for the class of split graphs (since they are characterized by a finite set of forbidden induced subgraphs) as well as for the class of interval graphs (since the corresponding vertex-deletion problem admits a constant-factor approximation algorithm). This answers a challenge raised by Eiben, Ganian, Hamm, and Kwon~\cite{EibenGHK19}.
}


For the specific case of eliminating vertices to obtain a graph of maximum degree at most~$d \in \Oh(1)$ (i.e.,~a graph which does not contain~$K_{1,d+1}$ or any of its spanning supergraphs as an induced subgraph), corresponding to the family~$\hh_{\leq d}$ of graphs of degree at most~$d$, our algorithm runs in time~$2^{(\hhdepth[\hh_{\leq d}](G))^{\Oh(1)}} \cdot n^{\Oh(1)}$ and outputs a degree-$d$-elimination forest of depth $\Oh((\hhdepth[\hh_{\leq d}](G))^2)$. This improves on earlier work by Bulian and Dawar~\cite{BulianD16}, who gave an FPT-algorithm with an unspecified but computable parameter dependence to compute a decomposition of depth~$k^{2^{\Omega(k)}}$ for~$k = \hhdepth[\hh_{\leq d}](G)$.

We complement Theorem~\ref{thm:decomposition:general} by showing that, assuming FPT~$\neq$~W[1], no FPT-approximation algorithms are possible for $\hhdepth[\mathsf{perfect}]$ or $\hhtw[\mathsf{perfect}]$, corresponding to perfect graphs.

By applying problem-specific insights to the problem-specific graph decompositions, we obtain FPT algorithms for vertex-deletion to~$\hh$ parameterized by~$\hhdepth$ and~$\hhtw$.

\begin{thm} \label{thm:solving:general}
Let $\hh$ be a hereditary class of graphs that is defined by a finite number of forbidden {connected} {(a)~minors, or (b) induced subgraphs, or (c)~$\hh \in \{\mathsf{bipartite},\mathsf{chordal}, \mic{\mathsf{interval}}\}$.} There is an algorithm that, given an $n$-vertex graph $G$, computes a minimum vertex set~$X$ such that~$G - X \in \hh$ in time~$f(\hhtw(G)) \cdot n^{\Oh(1)}$.
\end{thm}

{As a consequence of case (a), for the class defined by the set of forbidden minors $\{K_5, K_{3,3}\}$ 
we obtain an~algorithm for \textsc{Vertex planarization} parameterized by $\hhtw[\mathsf{planar}](G)$.
In \mic{all} cases the running time is of the form~$2^{k^{\Oh(1)}} \cdot n^{\Oh(1)}$, where $k = \hhtw(G)$.
This is obtained by combining the $k^{\Oh(1)}$-approximation from Theorem~\ref{thm:decomposition:general} with efficient algorithms working on given decompositions.

For example, for $\hh$ being the class of bipartite graphs, we present an algorithm for \textsc{Odd cycle transversal} with running time
$2^{\Oh(k^3)} \cdot n^{\Oh(1)}$
parameterized by $k = \hhtw(G)$.
\mic{For $\hh = \mathsf{chordal}$, the running time for \textsc{Chordal deletion} is
$2^{\Oh(k^{10})} \cdot n^{\Oh(1)}$.}
If $\hh$ is defined by a family of connected forbidden (induced) subgraphs on at most $c$ vertices, we can solve $\hh$\textsc{-deletion} in time $2^{\Oh(k^{6c})} \cdot n^{\Oh(1)}$.
Analogous statements with slightly better running times hold for parameterizations by $\hhdepth(G)$, where in some cases we are additionally able to deliver polynomial-space algorithms.
}

\bmp{All our algorithms are \emph{uniform}: a single algorithm works for all values of the parameter. Using an approach based on finite integer index introduced by Eiben et al.~\cite[Thm. 4]{EibenGHK19}, it is not difficult to leverage the decompositions of Theorem~\ref{thm:decomposition:general} into \emph{non-uniform} fixed-parameter tractable algorithms for the corresponding vertex-deletion problems with an unknown parameter dependence in the running time. Our contribution in Theorem~\ref{thm:solving:general} therefore lies in the development of \emph{uniform} algorithms with concrete bounds on the running times.

For example, for the problem of deleting vertices to obtain a graph of maximum degree at most~$d$ our running time is~$2^{\Oh(\hhtw[\hh_{\leq d}](G))^{6(d+1)}} \cdot n^{\Oh(1)}$. This improves upon a recent result by Ganian, Flute, and Ordyniak~\cite[Thm.~11]{GanianKO20}, who gave an algorithm parameterized by the \emph{core fracture number}~$k$ of the graph whose parameter dependence is~$\Omega(2^{2^{{(2k)}^k}})$. A graph with core fracture number~$k$ admits an $\hh_{\leq d}$-elimination forest of depth at most~$2k$, so that the parameterizations by~$\hhtw[\hh_{\leq d}]$ and~$\hhdepth[\hh_{\leq d}]$ can be arbitrarily much smaller and at most twice as large as the core fracture number.
}

Intuitively, the use of decompositions in our algorithms is a strong generalization of the ubiquitous idea of splitting a computation on a graph into independent computations on its connected components. Even if the components are not fully independent but admit limited interaction between them, via small vertex sets, Theorem~\ref{thm:solving:general} exploits this algorithmically. We consider the theorem \bmp{an important step}  in the quest to identify the smallest problem parameters that can explain the tractability of inputs to NP-hard problems~\cite{FellowsJR13,GuoHN04,Niedermeier10} (cf.~\cite[\S 6]{Niedermeier06}). The theorem explains for the first time how, for example, classes of inputs for \textsc{Odd cycle transversal} whose solution size and \bmp{rankwidth} are both large, can nevertheless be solved efficiently and exactly.

Not all classes~$\hh$ for which Theorem~\ref{thm:decomposition:general} yields decompositions, are covered by Theorem~\ref{thm:solving:general}. The problem-specific adaptations needed to exploit the decompositions require significant technical work. In this article, we chose to focus on a number of key applications.
We purposely restrict to the case where the forbidden induced subgraphs or minors are connected, {as otherwise
the problem does not exhibit the behavior with ``limited  interaction between components'' anymore.}
We formalize this in Section~\ref{subsec:not:closed} with a proof that the respective $\hh$-\textsc{deletion} problem parameterized by either $\hhtw$ or $\hhdepth$ is para-NP-hard.
Regardless of that, we provide our decomposition theorems in full generality, as they may be of independent interest. 

For some problems, we can leverage $\hh$-elimination forests to obtain polynomial-space analogues of Theorem~\ref{thm:solving:general}. For example, for the \textsc{Odd cycle transversal} problem we obtain an algorithm running in time~$n^{\Oh(1)} + 2^{\hhdepth[\mathsf{bip}](G)^{\Oh(1)}}\cdot(n+m)$ and polynomial space. The additive behavior of the running time comes from the fact that in this case, an approximate $\mathsf{bip}$-elimination forest can be computed in \emph{polynomial time}, while the algorithm working on the decomposition is linear in~$n + m$.

{A~result of a different kind is an~FPT algorithm for \textsc{Vertex cover} parameterized by $\hhtw \, / \,\hhdepth$ for $\hh$ being any class for which Theorem~\ref{thm:decomposition:general} works and on which \textsc{Vertex cover} is solvable in polynomial time.
This means that \textsc{Vertex cover} is tractable on graphs with small \hhtwfull{} for $\hh \in \{\mathsf{bipartite, chordal, claw\text{-}free}\}$. \bmp{While this tractability was known when the input is given \emph{along} with a decomposition~\cite{EibenGHK19}, the fact that suitable decompositions can be computed efficiently to establish tractability when given only a graph as input, is novel.}
}

\bmp{
\paragraph{Related work}
Our results fall in a recent line of work of using \emph{hybrid parameterizations}~\cite{BulianD16,BulianD17,EibenGS15,EibenGS18,GanianKO20,GanianOR17,GanianOS19,GanianRS17a,GanianRS17d,GanianRS17b,HecherTW20} which simultaneously capture the connectivity structure of the input instance and properties of its optimal solutions. Several of these works give decomposition algorithms for parameterizations similar to ours; we discuss those which are most relevant.


Eiben, Ganian, Hamm, and Kwon~\cite{EibenGHK19} introduced the notion of $\hh$-treewidth and developed an~algorithm to compute $\hhtw[\mathsf{rw} \leq c]$ corresponding to $\hh$-treewidth where~$\hh$ is the class of graphs of rankwidth at most~$c$, for any fixed~$c$. As the rankwidth~$\mathsf{rw}(G)$ of any graph~$G$ satisfies~$\mathsf{rw}(G) \leq c + \hhtw[\mathsf{rw} \leq c](G)$, their parameterization effectively bounds the rankwidth of the input graph, allowing the use of model-checking theorems for Monadic Second Order Logic to compute a decomposition. In comparison, our Theorem~\ref{thm:decomposition:general} captures several classes of unbounded rankwidth such as bipartite graphs and chordal graphs.

Eiben, Ganian, and Szeider~\cite{EibenGS18} introduced the hybrid parameter \emph{$\hh$-well-structure number} of a~graph~$G$, which is a hybrid of the vertex-deletion distance to~$\hh$ and the rankwidth of a graph. They give FPT algorithms with unspecified parameter dependence~$f(k)$ to compute the corresponding graph decomposition for {classes}~$\hh$ defined by a~finite set of forbidden induced subgraphs, forests, and chordal graphs, and use this to solve \textsc{Vertex cover} parameterized by the $\hh$-well-structure number. Their parameterization is orthogonal to ours; for {a~class}~$\hh$ of unbounded rankwidth, graphs of $\hh$-elimination distance~$1$ can have arbitrarily large $\hh$-well-structure number, as the latter measure is the sum rather than maximum over connected components that do not belong to~$\hh$.

Ganian, Ramanujan, and Szeider~\cite{GanianRS17d} introduced a related hybrid parameterization in the context of constraint satisfaction problems, the \emph{backdoor treewidth}. They show that if a certain set of variables~$X$ which allows a CSP to be decided quickly, called a strong backdoor, exists for which a suitable variant of treewidth is bounded for~$X$, then such a set~$X$ can be computed in FPT time. They use the technique of \emph{recursive understanding} (cf.~\cite{ChitnisCHPP16}) for this decomposition algorithm. This approach leads to a running time with a very quickly growing parameter dependence, at least doubly-exponential (cf.~\cite[Lemma 8]{GanianRS17dArxiv}).
\mic{After the conference version of this article was announced, several results employing recursive understanding to compute $\hhdepth(G)$ for various classes $\hh$ have been obtained \cite{AgrawalKLPRSZ21, AgrawalKFR21, FominGT21, jansen2021fpt}.}
\micr{todo: update the arxiv citations when the papers get published, potentially plus the one from soda} 


While the technique of recursive understanding and our decomposition framework both revolve around repeatedly finding separations of a certain kind, the approaches are fundamentally different. The separations employed in recursive understanding are independent of the target class~$\hh$ and always separate the graph into two large sides by a small separator. In comparison, the properties of the separations that our decomposition framework works with crucially vary with~$\hh$. As a final point of comparison, we note that our decomposition framework is sometimes able to deliver an approximately optimal decomposition in \emph{polynomial time} (for example, for the class~$\hh$ of bipartite graphs), which is impossible using recursive understanding as merely finding a single reducible separation requires exponential time in the size of the separator.






}

\paragraph{Organization}
We start in Section~\ref{sec:outline} by providing an informal overview of our algorithmic contributions. In Section~\ref{sec:preliminaries} we give formal preliminaries on graphs and complexity. There we also define our graph decompositions. Section~\ref{sec:decomposition} shows how to compute approximate tree $\hh$-decompositions and $\hh$-elimination forests for various graph classes~$\hh$, building up to a proof of Theorem~\ref{thm:decomposition:general}. The algorithms working on these decompositions are presented in Section~\ref{sec:solving}, leading to a proof of Theorem~\ref{thm:solving:general}. In Section~\ref{sec:hardness} we present two hardness results marking the boundaries of tractability. We conclude with a range of open problems in Section~\ref{sec:conclusions}.

\section{Outline} \label{sec:outline}
\subsection{Constructing decompositions}
We give a high-level overview of the techniques behind Theorem~\ref{thm:decomposition:general}, which are generic and can be applied in more settings than just those mentioned in the theorem. Theorem~\ref{thm:solving:general} requires {solutions tailored} for each distinct~$\hh$, and will be discussed in Section~\ref{sec:outline:deletion}. 

The starting point for all algorithms to compute $\hh$-elimination forests or tree $\hh$-decompositions, is a subroutine to compute \bmp{a novel kind of graph separation called}~$(\hh, k)$-\emph{separation}, where~$k \in \mathbb{N}$. Such a separation in a graph~$G$ consists of a pair of disjoint vertex subsets~$(C,S)$ such that~$C$ induces a subgraph belonging to~$\hh$ and consists of one or more connected components of~$G-S$ for some separator~$S$ of size \jjh{at most}~$k$. If there is an FPT-time (approximation) algorithm, parameterized by~$k$, for the problem of computing an $(\hh,k)$-separation $(C,S)$ such that~$Z \subseteq C \cup S$ for some input set~$Z$, then we show how to obtain algorithms for computing $\hh$-elimination forests and tree $\hh$-decompositions in a black-box manner. We use a four-step approach to approximate $\hhtw(G)$ and~$\hhdepth(G)$:
\begin{itemize}
    \item \bmp{Roughly speaking, given an initial vertex~$v$ chosen from the still-to-be decomposed part of the graph, we iteratively apply the separation algorithm to find an $(\hh,k)$-separation~$(C \ni v,S)$ until reaching a separation whose set~$C$ cannot be extended anymore.} We then compute a preliminary decomposition by repeatedly extracting such extremal $(\hh,k)$-separations~$(C_i, S_i)$ from the graph.
    \item We use the extracted separations~$(C_i, S_i)$ to define a partition of~$V(G)$ into connected sets~$V_i \subseteq C_i \cup S_i$. Then we define a contraction~$G'$ of~$G$, by contracting each connected set~$V_i$ to a single vertex.
    \bmp{The extremal property of the separations~$(C_i,S_i)$ ensures that, apart from corner cases, the sets~$V_i$ cannot completely live in base components of an optimal tree $\hh$-decomposition or $\hh$-elimination forest.}
    These properties of the preliminary decomposition will enforce that the treewidth (respectively treedepth) of~$G'$ is not much larger than~$\hhtw(G)$ (respectively~$\hhdepth(G)$).
    \item We invoke an exact~\cite{Bodlaender96,ReidlRVS14} or approximate~\cite{BodlaenderDDFMP16,CzerwinskiNP19,RobertsonS95b} algorithm to compute the treewidth (respectively treedepth) of~$G'$.
    \item Finally, we transform the tree decomposition (elimination forest) of~$G'$ into a tree $\hh$-decomposition ($\hh$-elimination forest) of~$G$ without increasing the width (depth) too much.
\end{itemize}

The last step is more challenging for the case of constructing a tree $\hh$-decomposition.
In order to make it work, we need to additionally assume that the $C$-sides of the computed separations do not intersect too much.
We introduce the notion of restricted separation that formalizes this requirement and work with a~restricted version of the preliminary decomposition.

With the recipe above for turning $(\hh,k)$-separations into the desired graph decompositions, the challenge remains to find such separations. We use several algorithmic approaches for various~$\hh$. \mic{Table~\ref{table:decomposition} on page~\pageref{table:decomposition} collects} the results on obtaining \hhdepthdecomp{s} and \hhtwdecomp{s} of moderate depth/width for various graph classes
(formalized as theorems in Section~\ref{sec:summary}).

For bipartite graphs, there is an elementary polynomial-time algorithm that given a graph~$G$ and connected vertex set~$Z$, computes a~$(\hh, \leq 2k)$-separation~$(C \supseteq Z,S)$ if there exists a~$(\hh, k)$-separation~$(C' \supseteq Z,S')$. The algorithm is based on computing a minimum $s-t$ vertex separator, using the fact that computing a vertex set~$S \not \ni v$ such that the component of~$G - S$ containing~$v$ is bipartite, can be phrased as separating an ``even parity'' copy of~$v$ from an ``odd parity'' copy of~$v$ in an auxiliary graph. \bmp{Hence this variant of the~$(\mathsf{bip},k)$-separation problem can be approximated in \emph{polynomial} time.}

For other considered graph classes we present branching algorithms {to compute approximate separations in FPT time}.
For classes~$\hh$ {defined by a finite set of forbidden induced subgraphs~$\mathcal{F}$}, we can find a vertex set $F \subseteq V(G)$ such that $G[F] \in \mathcal{F}$ in polynomial time, {if one exists}.
If there exists an \hsepk $(C,S)$ covering the given set $Z$, i.e., $Z \subseteq C \in \hh$ and $|S| \le k$, then $F$ cannot be fully contained in $C$.
\bmp{To find a separation satisfying $F \cap S \neq \emptyset$, we can} {guess a vertex in the intersection and recurse on} a subproblem with $k' = k-1$.
\bmp{For separations with~$F \cap S = \emptyset$,} some vertex $v \in F$ lies in a different component of $G - S$ than $Z$.
We prove that in this case we can assume that $S$ contains an important $(v, Z)$-separator of size $\le k$ and we can branch on all such important separators.

Our most general approach to finding $(\hh, k)$-separations relies on the packing-covering duality of obstructions to $\hh$. 
It is known that graphs with bounded treewidth enjoy such a packing-covering duality
for connected obstructions for various graph classes \cite{RaymondT17}.
We extend this observation to graphs with bounded \hhtwfull{} and obstructions to $\hh$.
Namely, we show that when $\hhtw(G) \le k$ then $G$ contains either a vertex set $X$ at most $k(k+1)$ such that $G-X \in \hh$ or $k+1$ vertex-disjoint connected subgraphs which do not belong to $\hh$. 

In order to exploit this existential observation algorithmically, we again take advantage of important separators.
Suppose that there exists an \hsepk{} $(C,S)$ such that the given vertex set $Z$ is contained in $C$.
In the first scenario, we rely on the existing algorithms for $\hh$-deletion to find an $\hh$-deletion set $X$.
Even if the known algorithm is only a constant-factor approximation, it suffices to find $X$ of size $\Oh(k^2)$.
Then the pair $(V(G) \setminus X, X)$ forms an $(\hh, \Oh(k^2))$-separation {as desired}.

\mic{
In the second scenario, by a counting argument there exists at least one obstruction to $\hh$ which is disjoint from $S$.
This obstruction cannot be contained in $C$, so by connectivity it is disjoint from $C \cup S$.
Then the set $S$ forms a small separator between this obstruction and $Z$.}
By considering all vertices $v \in V(G) \setminus Z$ and all important $(v,Z)$-separators of size at most $k$, we can detect such an obstruction -- let us refer to it as $F \subseteq V(G)$ -- with a small boundary, i.e., $|N(F)| \le k$.
There are two cases depending on whether $F \cap S = \emptyset$.
If so, we proceed similarly as for graph classes defined by forbidden induced subgraphs.
Namely, we show that there exists an important $(F,Z)$-separator of size $\le k$ which is contained in $S$ (for some feasible solution $(C,S)$) and we can perform branching.
If $F \cap S \ne \emptyset$, then we take $N(F)$ into the solution and neglect $F$.
Observe that in the remaining part of the graph there exists an $(\hh, k-1)$-separation $(C^*,S^*)$ with $Z \subseteq C^*$.
We have thus decreased the value of the parameter by one at the expense of paying the size of $N(F)$, which is at most $k$.
This gives another branching rule.
Combining these two rules yields a recursive algorithm which outputs an $(\hh, \Oh(k^2))$-separation.

\subsection{Solving vertex-deletion problems}
\label{sec:outline:deletion}

{Here we provide an~overview of the algorithms that work on the problem-specific decompositions.
The results described above {allow} us to construct such decompositions
of depth/width being a polynomial function of the optimal value,
so in further arguments we can assume that a~respective \hhdepthdecomp{} or \hhtwdecomp{} is given.
For {most} applications of our framework, we build atop existing algorithms that process (standard) elimination forests and tree decompositions.
In order to make them work with the more general {types} of graph decomposition, we need to specially handle the base components.
To do this, we generalize arguments from the known algorithms {parameterized by the solution size}. \bmp{An overview of the resulting running times for solving \textsc{$\hh$-deletion} is given in Table~\ref{table:algorithms}.}

We follow the ideas of gluing graphs and finite state property dating back to the results of Fellows and Langston~\cite{Fellows89} (cf. \cite{Arnborg91, Bodlaender96reduction}).
We will present a meta-theorem which gives a recipe to solve \hh\textsc{-deletion}
parameterized by \hhtwfull{} or, as a~corollary, by \hhdepthfull{}.
\mic{It works for any class which is hereditary, union-closed, and satisfies two technical conditions.
}

The first condition concerns the operation $\oplus$ of gluing graphs.
Given two graphs with specified boundaries and an isomorphism between the boundaries, we can glue them along the boundaries by identifying the boundary vertices.
Technical details aside, two boundaried graphs $G_1, G_2$ are equivalent with respect to $\hh$\textsc{-membership} if for any other boundaried graph $H$, we have  $G_1 \oplus H \in \hh \Leftrightarrow G_2 \oplus H \in \hh$.
We say that the $\hh$\textsc{-membership} problem is finite state if the number of such equivalence classes is finite for each \bmp{boundary size}~$k$. 
We are interested in an upper bound $r_\hh(k)$, so that for every graph with boundary of size $k$ one can find an~equivalent graph on $r_\hh(k)$ vertices.
In our applications, we are able to provide polynomial bounds on $r_\hh(k)$, which could be
significantly harder for the approach based on \emph{finite integer index}~\cite{BodlaenderF01,EibenGHK19}.
\bmp{Before describing how to bound~$r_\hh(k)$, we first explain how such a bound can lead to an algorithm parameterized by treewidth.}

Each bag of a tree  decomposition forms a small separator.
Consider a bag $X_t \subseteq V(G)$ of size $k$ and set $A_t \subseteq V(G)$ of vertices introduced in the subtree of node $t$.
Then the subgraph $G_t$ induced by vertices $A_t \cup X_t$ has a natural small boundary $X_t$.
Suppose that for two subsets $S_1, S_2 \subseteq A_t\cup X_t$,
we have $S_1 \cap X_t = S_2 \cap X_t$ and the boundaried graphs $G_t - S_1, G_t - S_2$ are equivalent \mic{with respect to $\hh$\textsc{-membership}}.
Then $S_1$ are $S_2$ are equally good for the further choices of the algorithm: if some set $S' \subseteq V(G) \setminus (A_t \cup X_t)$ extends $S_1$ to a valid solution, the same holds for $S_2$.
If we can enumerate all equivalence classes for $\hh$\textsc{-membership}, we could store at each node $t$ of the tree decomposition and each equivalence class $\mathcal{C}$, the minimum size of a deletion set $S$ within $G_t$ so that $G_t - S \in \mathcal{C}$; this provides \bmp{sufficient} information for a dynamic-programming routine.
Such an approach \bmp{has} been employed to design optimal algorithms solving \hh\textsc{-deletion} parameterized by treewidth for minor-closed classes~\cite{baste20complexity}.

\bmp{We modify this idea to additionally handle the base components, which are arbitrarily large subgraphs that belong to~$\hh$ stored in the leaves of a tree $\hh$-decomposition or $\hh$-elimination forest, and which are separated from the rest of the graph by a vertex set~$X_t$ whose size is bounded by the cost~$k$ of the decomposition. This separation property ensures that any optimal solution~$S$ to~\textsc{$\hh$-deletion} contains at most~$k$ vertices from a base component~$A_t$, as otherwise we could replace $A_t \cap S$ by $X_t$ to obtain a smaller solution. This means that in principle, we can afford to use an algorithm for the parameterization by the solution size and run it with the cost value~$k$ of the decomposition. However, such an algorithm does not take into account the connections between the base component and the rest of the graph. } 
If we wanted to \bmp{take this into account by computing} a minimum-size deletion set $S_t$ in a base component $A_t$ \bmp{for which} $G[A_t \cup X_t] - S_t$ belongs to \bmp{a} given equivalence class $\mathcal{C}$, we would need a far-reaching generalization of the algorithm solving \hh\textsc{-deletion} parameterized by the solution size. 
\bmp{Working with a variant of the deletion problem that supports \emph{undeletable vertices} allows us to alleviate this issue.} We enumerate the minimal representatives of all the equivalence classes. 
Then, given a bag $X_t$ and the base component $A_t$,
we consider all subsets $X_s \subseteq X_t$ and perform gluing $G' = G[A_t \cup (X_t \setminus X_s)]$ with each representative $R$ along $X_t  \setminus X_s$.
One of the representatives $R$ is equivalent to the graph $G - A_t - S$, with the boundary $X_t \setminus S$, where $S$ is the optimal solution.
Therefore, the set $S \cap A_t$ is a solution to $\hh$\textsc{-deletion} for the graph $G' \oplus R$
and any subset \mic{$S^R_t \subseteq A_t$} with this property can be extended to a solution in $G$ using the vertices from $S \setminus (A_t \cup X_t)$.
Since we can assume that $|S \cap A_t|$ is at most the width of the decomposition, we can find its replacement $S^R_t \subseteq A_t$ of minimum size as long as we can solve $\hh$\textsc{-deletion} with some vertices marked as undeletable, parameterized by the solution size.
This constitutes the second condition for $\hh$.
We check that for all studied problems, the known algorithms can be adapted to work in this setting.

The generic dynamic programming routine works as follows.
First, we generate the minimal representatives of the equivalence classes with respect to $\hh$\textsc{-membership}.
The size of this family is governed by the bound $r_\hh(k)$, which differs for \mic{various} classes $\hh$.
For each base component $A_t$ and for each representative $R$,
we perform the gluing operation, compute a minimum-size subset $S^R_t \subseteq A_t$ that solves $\hh$\textsc{-deletion} on the obtained graph, \jjh{and add it} to a~family $\mathcal{S}_t$.
{Then for any optimal solution $S$, there exists $S_t \in \mathcal{S}_t$ such that $(S \setminus A_t) \cup S_t$ is also an optimal solution.
Such a family of partial solutions for $A_t \subseteq V(G)$ is called $A_t$-\emph{exhaustive}.
We proceed \bmp{to} compute exhaustive families bottom-up in a decomposition tree}, combining exhaustive families for children by brute force to get a new exhaustive family for their parent, and then trim the size of that family so it never grows too large. The following theorem summarizes our meta-approach.

\begin{thm} \label{thm:metathm:informal}
Suppose that the class $\hh$ is hereditary \mic{and union-closed}.
Furthermore, assume that \textsc{$\hh$-deletion} with undeletable vertices admits an algorithm with running time $f(k)\cdot n^{\Oh(1)}$, where $k$ is the solution size.
Then \textsc{$\hh$-deletion} \bmp{on an $n$-vertex graph~$G$} can be solved in time~$f(k) \cdot 2^{\Oh(r_\hh(k)^2 + k)} \cdot n^{\Oh(1)}$ when given a~tree $\hh$-decomposition \bmp{of~$G$} of width~$k$ consisting of~$n^{\Oh(1)}$ nodes.
\end{thm}

If the class $\hh$ is defined by a finite set of forbidden connected minors, we can take advantage of the theorem by
Baste,  Sau, and Thilikos~\cite{baste20complexity} which implies that $r_\hh(k) = \Oh_\hh(k)$, and a~construction by Sau, Stamoulis, and Thilikos~\cite{sau20apices} to solve \hh\textsc{-deletion} with undeletable vertices.
Combining these results with our framework gives an algorithm \mic{running in time $2^{k^{\Oh(1)}}\cdot n^{\Oh(1)}$ for parameter $k = \hhtw(G)$.}
For other classes we \mic{first} need to develop some theory to make them amenable \bmp{to} the meta-theorem.

\paragraph{Chordal deletion}
We explain \bmp{briefly} how the presented framework allows us to solve \textsc{Chordal} \textsc{deletion} when given a tree $\mathsf{chordal}$-decomposition.
The upper bound $r_\mathsf{chordal}(k)$ on the sizes of representatives
comes from a new characterization of chordal graphs.
\mic{Consider a boundaried chordal graph $G$ with a boundary $X$ of size $k$.
We define the \emph{condensing} operation, which contracts edges with both endpoints in $V(G) \setminus X$ and removes vertices from $V(G) \setminus X$ {which are simplicial} (a vertex is simplicial if its neighborhood is a clique).
Let $\widehat{G}$ be obtained by condensing $G$.
We prove that $G$ and $\widehat{G}$ are equivalent with respect to $\mathsf{chordal-}$\textsc{membership}, that is, for any other boundaried graph $H$, the result of gluing $H \oplus G$ gives a chordal graph if and only $H \oplus \widehat{G}$ is chordal.
Furthermore, we show that after condensing there can be at most $\Oh(k)$ vertices in $\widehat{G}$, which implies $r_\mathsf{chordal}(k) = \Oh(k)$.}
As a result, there can be at most $2^{\Oh(r_\mathsf{chordal}(k)^2)} = 2^{\Oh(k^2)}$ equivalence classes of graphs with boundary at most~$k$. 

The second required ingredient is an algorithm for $\textsc{Chordal deletion}$ with undeletable vertices, parameterized by the solution size $k$.
We provide a simple reduction to the standard $\textsc{Chordal deletion}$ problem, which admits an algorithm with running time~$2^{\Oh(k \log k)} \cdot n^{\Oh(1)}$~\cite{CaoM16}.
\mic{Our construction} directly implies that \textsc{Chordal deletion} on graphs of (standard) treewidth~$k$ can be solved in time~$2^{\Oh(k^2)} \cdot n^{\Oh(1)}$. To the best of our knowledge, this is the first explicit treewidth-DP for \textsc{Chordal deletion}; previous algorithms all relied on Courcelle's theorem.

Together with the FPT-approximation algorithm for computing a tree $\mathsf{chordal}$-decomposition, we obtain (Corollary~\ref{thm:final-chordal-tw}) an algorithm solving \textsc{Chordal deletion} in time~$2^{\Oh(k^{10} )} \cdot n^{\Oh(1)}$ when parameterized by $k = \mathbf{tw}_\mathsf{chordal}(G)$.

\paragraph{Interval deletion}
In order to solve \textsc{Interval deletion} with a given tree $\mathsf{interval}$-decomposition, we need to bound the function $r_\mathsf{interval}(k)$.
A graph $G$ is interval if and only if $G$ is chordal and $G$ does not contain an asteroidal triple (AT), that is, a triple of vertices so that
for any two of them there is a path between them avoiding the closed neighborhood of the third.
As we have already developed a theory to understand  $\mathsf{chordal-}$\textsc{membership}, the main task remains to identify graph modifications which preserve the structure of  asteroidal triples.
We show that given a boundaried interval graph $G$ with boundary $X$ of size $k$, we can mark a set $Q \subseteq V(G)$ of size $k^{\Oh(1)}$ so that for any \mic{result of gluing} $F = H \oplus G$, if $F$ contains some AT, then $F$ contains an AT $(w_1,w_2,w_3)$ so that $\{w_1,w_2,w_3\} \cap V(G) \subseteq X \cup Q$.
Afterwards, the vertices from $V(G) \setminus (X \cup Q)$ can be either removed or contracted together, so that in the end we obtain a boundaried graph $G'$ of size $k^{\Oh(1)}$ which is equivalent to $G$ with respect to $\mathsf{interval-}$\textsc{membership}.

As the most illustrative case \bmp{of the marking scheme}, consider vertices $v_1,v_2,v_3 \in V(G)$, such that there exists a $(v_1,v_2)$-path $P$ in the graph $F - N_F[v_3]$.
Recall that $F = H \oplus G$, where $H$ is an unknown boundaried graph.
Intuitively, the marking scheme tries to enumerate all ways in which $P$
crosses the boundary $X$ of $G$.
The naive way would be to enumerate each ordered subset of $X$ and consider paths which visit these vertices of $X$ in the given order.
There are $\Omega(k!)$ such combinations though.
Instead, we fix an interval model of $G$
and show that only $\Oh(1)$ subpaths of $P$ on the $G$-side of $X$ are relevant.
These are the subpaths closest to the vertices $v_1, v_2, v_3$ in the interval model.
This allows us to define a \emph{signature} of path $P$ with the following properties: (1) if there exists another triple $w_1,w_2,w_3 \in V(G)$ which matches the signature, then there exists a $(w_1,w_2)$-path $P'$ in the graph $F - N_F[w_3]$ and (2) there are only $k^{\Oh(1)}$ different signatures.
For any triple of signatures, representing respectively $(v_1,v_2)$-path, $(v_2,v_3)$-path, and $(v_3,v_1)$-path, we check if there exists a triple of vertices that matches all of them (modulo ordering).
If yes, we mark the vertices in any such triple.
This procedure marks a set $Q \subseteq V(G)$ of $k^{\Oh(1)}$ vertices, as intended.

Similarly as for \textsc{Chordal deletion}, we adapt the algorithm for \textsc{Interval deletion} parameterized by the solution size~\cite{CaoM2015} to work with undeletable vertices.
Together, this makes \textsc{Interval deletion} amenable to our framework (Corollary~\ref{thm:meta-interval:final}).

\paragraph{Hitting \jjh{connected} forbidden induced subgraphs}

We use the same approach
to solve \textsc{$\hh$-deletion} when~$\hh$ is defined by a finite set~$\mathcal{F}$ of connected forbidden induced subgraphs on at most~$c$ vertices each. \bmp{The standard technique of bounded-depth branching provides an FPT algorithm for the parameterization by solution size in the presence of undeletable vertices.}
We prove an upper bound $r_\hh(k) = \Oh(k^c)$
and obtain
 running time~$2^{\Oh(k^{6c})} \cdot n^{\Oh(1)}$, where~$k = \hhtw(G)$ (Corollary~\ref{thm:final-induced-tw}). 

In the special case when~$\hh$ is defined by a single forbidden induced subgraph that is a clique~$K_t$, we can additionally obtain (Corollary~\ref{thm:klfree-elimcombine}) a polynomial-space algorithm for the parameterization by~$k=\hhdepth$, which runs in time~$2^{\Oh(k^3 \log k)} \cdot n^{\Oh(1)}$. Here the key insight is that a forbidden clique is represented on a single root-to-leaf path of an $\hh$-elimination forest, allowing for a polynomial-space branching step that avoids the memory-intensive dynamic-programming technique.

When the forbidden clique~$K_t$ has size two, then we obtain the \textsc{Vertex cover} problem. The family~$\hh$ of~$K_2$-free graphs is simply the class of edge-less graphs, and the elimination distance to an edge-less graph is not a smaller parameter than treedepth or treewidth. But for \textsc{Vertex cover} we can work with even more relaxed parameterizations. For~$\hh$ defined by a finite set of forbidden induced subgraphs such that \textsc{Vertex cover} is polynomial-time solvable on graphs from~$\hh$ (for example, claw-free graphs), or $\hh \in \{\mathsf{chordal}, \mathsf{bipartite}\}$ (for which \textsc{Vertex cover} is also polynomial-time {solvable}), we can find a minimum vertex cover in FPT time parameterized by~$\hhdepth$ (Corollary~\ref{thm:vcelimcombine}) and~$\hhtw$ (Corollary~\ref{thm:vcdpcombine}). In the former case, the algorithm even works in polynomial space. More generally, \textsc{Vertex cover} is FPT parameterized by~$\hhdepth$ and~$\hhtw$ for any hereditary class~$\hh$ on which the problem is polynomial, when a decomposition is given in the input.

\paragraph{Odd cycle transversal}
\bmp{While this problem can be shown to fit into the framework of Theorem~\ref{thm:metathm:informal}, we provide specialized algorithms with improved \mic{guarantees}. Given a tree $\mathsf{bip}$-decomposition of width~$k$, we can solve the problem in time~$2^{\Oh(k)} \cdot n^{\Oh(1)}$ (Theorem~\ref{thm:oct-dp}) by utilizing a subroutine developed for iterative compression~\cite{ReedSV04}.
\mic{What is more, 
we obtain the same time bound within polynomial space when given a $\mathsf{bip}$-elimination forest of depth~$k$.}
Below, we describe how the iterative-compression subroutine is used in the latter algorithm.}

Suppose we are given a vertex set $S \subseteq V(G)$ that separates $G$ into connected components $C_1, \dots, C_p$.
The optimal solution may remove some subset $S_X \subseteq S$ and the vertices in $S \setminus S_X$ can {then be} divided into 2 groups $S_1, S_2$ reflecting the 2-coloring of the resulting bipartite graph.
If $|S| \le d$, then there {are} $3^d$ choices to split $S$ into $(S_1, S_2, S_X)$.
We can consider all of them and solve the problem recursively on $C_1 \cup S, \dots, C_p \cup S$, restricting to the solutions coherent with the partition of $(S_1, S_2, S_X)$.
We call such a subproblem with restrictions an~\emph{annotated problem}.

A \bmp{standard} elimination forest provides us with a~convenient mechanism of separators, given by the node-to-root paths, \bmp{to be used in the scheme above}.
The length of each such path is bounded by the depth $d$ of the elimination forest, so we can solve the problem recursively, starting from the separation given by the root node, in time $\Oh(3^d \cdot d^{\Oh(1)} \cdot n)$ \bmp{when given a depth-$d$ elimination forest}.
Moreover, such a computation needs only to keep track of the annotated vertices in each recursive call, so it can be implemented to run in polynomial space.

If we replace the standard elimination forest with a 
$\mathsf{bip}$-elimination forest, the idea is analogous but we need to additionally take care of the base components.
In each such subproblem we are given a bipartite component $C$, a partition $(S_1, S_2, S_X)$ of $N(C)$, and want to find a minimal set $C_X \subseteq C$
so that $C - C_X$ is coherent with $(S_1, S_2, S_X)$:
that is, there is no even path from $u \in S_1$ to $v \in S_2$ and no odd path between vertices from each $S_i$.
It turns out that the same subproblem occurs in the algorithm for \textsc{Odd cycle transversal} parameterized by the solution size in a single step in {iterative compression}. 
\bmp{This} problem, called
\textsc{Annotated bipartite coloring}, can be reduced to finding {a} minimum cut and is solvable in polynomial-time.
Furthermore, we can assume that $C_X$ is at most as large as the depth of the given $\mathsf{bip}$-elimination forest, because otherwise we could remove the set $N(C)$ instead of $C_X$ and obtain a smaller solution.
This observation allows us to improve the running time
to be linear in $n+m$ (Corollary~\ref{thm:octelimcombine}), so that given a graph~$G$ of $\mathsf{bip}$-elimination distance~$k$ we can compute a minimum odd cycle transversal in~$n^{\Oh(1)} + 2^{\Oh(k^{3}\log k)} \cdot (n+m)$ time and polynomial space \bmp{by using a polynomial-space algorithm to \mic{construct} a $\mathsf{bip}$-elimination forest}.


\bmp{
\begin{table}[bt]
\centering
\caption{\small Running times for solving \textsc{$\hh$-Deletion} parameterized by~$\hhdepth$ and~$\hhtw$. The algorithms only need the graph~$G$ as input and construct a decomposition within the same time bounds.} \label{table:algorithms}
\[\begin{array}{|c|r|r|}\hline
\text{class } \hh & k = \hhtw & k =\hhdepth \\ \hline
\mathsf{bipartite} & 2^{\Oh(k^3)} \cdot n^{\Oh(1)} & n^{\Oh(1)} + 2^{\Oh(k^{3}\log k)} \cdot (n+m) \\
\mathsf{chordal} & 2^{\Oh(k^{10})} \cdot n^{\Oh(1)} & 2^{\Oh(k^{6})} \cdot n^{\Oh(1)}  \\
\mic{\mathsf{interval}} & 2^{k^{\Oh(1)}} \cdot n^{\Oh(1)} & 2^{k^{\Oh(1)}} \cdot n^{\Oh(1)} \\
\mic{\mathsf{planar}} & 2^{\Oh(k^{5}\log k)} \cdot n^{\Oh(1)} & 2^{\Oh(k^{3} \log k)} \cdot n^{\Oh(1)}  \\
\mic{  \text{forbidden connected minors}} & 2^{k^{\Oh(1)}} \cdot n^{\Oh(1)} & 2^{k^{\Oh(1)}} \cdot n^{\Oh(1)}  \\ 
 \text{forbidden connected induced subgraphs } & & \\
  \text{on~$\leq c$ vertices} & 2^{\Oh(k^{6c})} \cdot n^{\Oh(1)} & 2^{\Oh(k^{4c})} \cdot n^{\Oh(1)} \\
 \text{forbidden clique $K_t$ } & 2^{\Oh(k^{6t})} \cdot n^{\Oh(1)} & 2^{\Oh(k^{3}\log k)} \cdot n^{\Oh(1)} \\
\hline
\end{array}\]
\end{table}
}

\section{Preliminaries} \label{sec:preliminaries}
The set $\{1,\ldots,p\}$ is denoted by $[p]$. For a finite set~$S$, we denote by~$2^S$ the powerset of~\jjh{$S$} consisting of all its subsets. For~$n \in \mathbb{N}$ we use~$S^n$ to denote the set of $n$-tuples over~$S$. 
We consider simple undirected graphs without self-loops. A graph $G$ has vertex set $V(G)$ and edge set $E(G)$. We use shorthand $n = |V(G)|$ and $m = |E(G)|$. For disjoint $A,B \subseteq V(G)$, we define $E_G(A,B) = \{uv \mid u \in A, v \in B, uv \in E(G)\}${, where we omit subscript $G$ if it is clear from context}.
For $A \subseteq V(G)$, the graph induced by $A$ is denoted by $G[A]$ and {we say that the vertex set $A$ is connected if the graph $G[A]$ is connected.}
We use shorthand $G-A$ for the graph $G[V(G) \setminus A]$. For $v \in V(G)$, we write $G-v$ instead of $G-\{v\}$. The {open} neighborhood of $v \in V(G)$ is $N_G(v) := \{u \mid uv \in E(G)\}$, where we omit the subscript $G$ if it is clear from context. {For a vertex set~$S \subseteq V(G)$ the open neighborhood of~$S$, denoted~$N_G(S)$, is defined as~$\bigcup _{v \in S} N_G(v) \setminus S$. The closed neighborhood of a single vertex~$v$ is~$N_G[v] := N_G(v) \cup \{v\}$, and the closed neighborhood of a vertex set~$S$ is~$N_G[S] := N_G(S) \cup S$.} 
For two disjoint sets $X,Y \subseteq V(G)$, we say that $S \subseteq V(G) \setminus (X \cup Y)$ is an $(X,Y)$-separator if the graph $G-S$ does not contain any~path from any $u\in X$ to any $v\in Y$.

{For graphs~$G$ and~$H$, we write~$G \subseteq H$ to denote that~$G$ is a subgraph of~$H$.}
{A tree is a connected graph that is acyclic. A~forest is a disjoint union of trees. In tree $T$ with root $r$, we say that $t \in V(T)$ is an ancestor of $t' \in V(T)$ (equivalently $t'$ is a descendant of $t$)} {if $t$ lies on the (unique) path from $r$ to $t'$. } A graph $G$ admits a proper $q$-coloring, if there exists a function $c \colon V(G) \to [q]$ such that $c(u) \neq c(v)$ for all $uv \in E(G)$. A graph is {bipartite} if and only if it admits a proper 2-coloring. {A graph is {chordal} if it does not contain any induced cycle of length at least four. A graph is an {interval} graph if it is the intersection graph of a set of intervals on the real line. It is well-known that all interval graphs are chordal (cf.~\cite{BrandstadtLS99}).}

\mic{Let $R_S^G(X)$ be the set of vertices reachable from $X \setminus S$ in $G - S$, where superscript $G$ is omitted if it is clear from context.}

\begin{definition}[{\cite{CyganFKLMPPS15}}] \label{def:imp:sep}
Let $G$ be a graph and let $X,Y \subseteq V(G)$ be two {disjoint} sets of vertices. Let $S \subseteq V(G) \setminus (X \cup Y)$ be an $(X,Y)$-separator.
\mic{We say that $S$ is an \emph{important $(X,Y)$-separator} if it is inclusion-wise minimal and there is no $(X,Y)$-separator $S' \subseteq V(G) \setminus (X \cup Y)$ such that $|S'| \leq |S|$ and $R^G_S(X) \subsetneq R^G_{S'}(X)$.}
\end{definition}

\begin{lemma}[{\cite[Proposition 8.50]{CyganFKLMPPS15}}] \label{lem:sep-to-importantsep}
Let $G$ be a graph and $X,Y \subseteq V(G)$ \mic{be} two disjoint sets of vertices. Let $S \subseteq V(G) \setminus (X \cup Y)$ be an $(X,Y)$-separator.
\mic{Then there is an important $(X,Y)$-separator $S' = N_G(R^G_{S'}(X))$ such that $|S'| \leq |S|$ and $R^G_S(X) \subseteq R^G_{S'}(X)$.}
\end{lemma}

\paragraph{Contractions and minors}
A contraction of $uv \in E(G)$ introduces a new vertex adjacent to all of {$N(\{u,v\})$}, after which $u$ and $v$ are deleted. The result of contracting $uv \in E(G)$ is denoted $G / uv$. For $A \subseteq V(G)$ such that $G[A]$ is connected, we say we contract $A$ if we simultaneously contract all edges in $G[A]$ and introduce a single new vertex.

We say that $H$ is a \emph{contraction} of $G$, if we can turn $G$ into $H$ by a series of edge contractions.
Furthermore, $H$ is a \emph{minor} of $G$, if we can turn $G$ into $H$ by a series of edge contractions, edge deletions, and vertex deletions.
We can represent the result of such \jjh{a} process with a mapping $\phi \colon V(H) \to 2^{V(G)}$, such that subgraphs $(G[\phi(h)])_{h\in V(H)}$ are connected and vertex-disjoint{, with an edge of~$G$ between a vertex in~$\phi(u)$ and a vertex in~$\phi(v)$ for all~$uv \in E(H)$.}
The sets $\phi(h)$ are called branch sets and the family $(\phi(h))_{h\in V(H)}$ is called a minor-model of $H$~in~$G$. A family $(\phi(h))_{h\in V(H)}$ of branch sets is called a contraction-model of~$H$ in~$G$ if the sets~$(\phi(h))_{h\in V(H)}$ partition~$V(G)$ and for each pair of distinct vertices~$u,v \in V(H)$ we have~$uv \in E(H)$ if and only if there is an edge in~$G$ between a vertex in~$\phi(u)$ and a vertex in~$\phi(v)$.

{A subdivision of an edge $uv$ is an operation that replaces the edge $uv$ with a~vertex $w$ connected to both $u$ and $v$.
We say that graph $G$ is a~subdivision of $H$ if $H$ can be transformed into $G$ by a~series of edge subdivisions.
A graph $H$ is called a~\emph{topological minor} of $G$ if there is a~subgraph of $G$ being a~subdivision of $H$.
Note that this implies that $H$ is also a~minor of $G$, but the implication in the opposite direction does not hold.}

\paragraph{Graph classes and decompositions}
We always assume that $\hh$ is a hereditary class of graphs, that is, closed under taking induced subgraphs.
A set $X \subseteq V(G)$ is called an~$\hh$-deletion set if $G - X \in \hh$.
The task of finding a smallest $\hh$-deletion set is called the $\hh$-\textsc{deletion} problem (also referred to as $\hh$-\textsc{vertex deletion}, but we abbreviate it since we do not consider edge deletion problems). \jjh{When parameterized by the solution size $k$, the task for the $\hh$-\textsc{deletion} problem is to either find a minimum-size $\hh$-deletion set or report that no such set of size at most $k$ exists.}
\mic{We say that class $\hh$ is union-closed if a disjoint union of two graphs from $\hh$ also belongs to $\hh$.}

\begin{definition}
For a graph class~$\hh$, an $\hh$-elimination forest of graph $G$ is pair~$(T, \chi)$ where~$T$ is a rooted forest and~$\chi \colon V(T) \to 2^{V(G)}$, such that:
\begin{enumerate}
    \item For each internal node~$t$ of~$T$ we have~$|\chi(t)| = 1$.
    \item The sets~$(\chi(t))_{t \in V(T)}$ form a partition of~$V(G)$.
    \item For each edge~$uv \in E(G)$, if~$u \in \chi(t_1)$ and~$v \in \chi(t_2)$ then~$t_1, t_2$ are in ancestor-descendant relation in~$T$.
    \item For each leaf $t$ of~$T$, the graph~$G[\chi(t)]$, {called a base component}, belongs to~$\hh$.
\end{enumerate}
The \emph{depth} of~$T$ is the maximum number of edges on a root-to-leaf path.
{We refer to the union of base components as the set of base vertices.} 
The $\hh$-elimination distance of~$G$, denoted $\hhdepth(G)$, is the minimum depth of an $\hh$-elimination forest for~$G$.
{A pair $(T, \chi)$ is a (standard) elimination forest if \hh{} is the class of empty graphs, i.e., the base components are empty. The treedepth of~$G$, denoted $\td(G)$, is the minimum depth of a standard elimination forest.}
\end{definition}

{It is straight-forward to verify that for any~$G$ and~$\hh$, the minimum depth of an $\hh$-elimination forest of~$G$ is equal to the $\hh$-elimination distance as defined recursively in the introduction. (This is the reason we have defined the depth of an $\hh$-elimination forest in terms of the number of edges, while the traditional definition of treedepth counts vertices on root-to-leaf paths.)}

{The following definition captures our relaxed notion of tree decomposition.}

\begin{definition} \label{def:tree:h:decomp}
For a graph class $\hh$, a tree $\hh$-decomposition of graph $G$ is a triple $(T, \chi, L)$ where~$L \subseteq V(G)$,~$T$ is a rooted tree, and~$\chi \colon V(T) \to 2^{V(G)}$, such that:
\begin{enumerate}
    \item For each~$v \in V(G)$ the nodes~$\{t \mid v \in \chi(t)\}$ form a {non-empty} connected subtree of~$T$. \label{item:tree:h:decomp:connected}
    \item For each edge~$uv \in E(G)$ there is a node~$t \in V(G)$ with~$\{u,v\} \subseteq \chi(t)$.
    \item For each vertex~$v \in L$, there is a unique~$t \in V(T)$ for which~$v \in \chi(t)$,  with~$t$ being a leaf of~$T$. \label{item:tree:h:decomp:unique}
    \item For each node~$t \in V(T)$, the graph~$G[\chi(t) \cap L]$ belongs to~$\hh$. \label{item:tree:h:decomp:base}
\end{enumerate}
{The \emph{width} of a tree $\hh$-decomposition is defined as~$\max(0, \max_{t \in V(T)} |\chi(t) \setminus L| - 1)$.} The $\hh$-treewidth of a graph~$G$, denoted~$\hhtw(G)$, is the minimum width of a tree $\hh$-decomposition of~$G$.
The connected components of $G[L]$ are called base components
{and the vertices in $L$ are called base vertices.}

{A pair~$(T, \chi)$ is a (standard) \emph{tree decomposition} if~$(T, \chi, \emptyset)$ satisfies all conditions of an $\hh$-decomposition; the choice of~$\hh$ is irrelevant.}
\end{definition}

{In the definition of width, we subtract one from the size of a largest bag to mimic treewidth. The maximum with zero is taken to prevent graphs~$G \in \hh$ from having ~$\hhtw(G) = -1$.}

{The following lemma shows that the relation between the standard notions of treewidth and treedepth translates into a relation between~$\hhtw$ and~$\hhdepth$.}

\begin{lemma}\label{lem:treedepth-treewidth}
For any class $\hh$ and graph $G$, we have $\hhtw(G) \le \hhdepth(G)$. {Furthermore, given an $\hh$-elimination forest of depth $d$ we can construct a tree $\hh$-decomposition of width} {$d$ in polynomial time.}
\begin{proof}
The argument is analogous as in the relation between treedepth and treewidth.
{If $G \in \hh$, then $\hhtw(G) = \hhdepth(G) = 0$.}
Otherwise, consider an \hhdepthdecomp{} $(T, \chi)$ of depth $d>0$ with the set of base vertices $L$ and a sequence of leaves $t_1, \dots, t_m$ of $T$ given by in-order traversal.
For each $t_i$ let $S_i$ denote the set of vertices in the nodes on the path from $t_i$ to the root of its tree, with $t_i$ excluded.
By definition, $|S_i| \le d$.
Consider a new tree $T_1$, where $t_1, \dots, t_m$ are connected as a path, rooted at $t_1$.
For each $i$ we create a new node $t'_i$ connected only to $t_i$.
We set $\chi_1(t_i) = S_i$ and $\chi_1(t'_i) = S_i \cup \chi(t_i)$.
Then $(T_1, \chi_1, L)$ forms a \hhtwdecomp{} of width $d-1$.
\end{proof}
\end{lemma}

\begin{lemma}\label{lem:treewidth-of-hh}
Suppose $(T, \chi, L)$ is a tree $\hh$-decomposition of $G$ of width $k$ and the maximal treewidth in $\hh$ is $d$. Then the treewidth of $G$ is at most $d+k+1$.
Moreover, if the corresponding decompositions are given, then the requested tree decomposition of $G$ can be constructed in polynomial time.
\end{lemma}
\begin{proof}
For a node~$t \in V(T)$, the graph~$G[\chi(t) \cap L]$ belongs to~$\hh$, so it admits a tree decomposition $(T_t, \chi_t)$ of width $d$.
Consider a tree $T_1$ given as a disjoint union of $T$ and $\bigcup_{t \in V(T)} T_t$ with additional edges between each $t$ and any node from $V(T_t)$.
We define $\chi_1(t) = \chi(t) \setminus L$ for $t \in V(T)$ and $\chi_1(x) = \chi_t(x) \cup (\chi(t) \setminus L)$ for $x \in V(T_t)$. The maximum size of a bag in $T_1$ is at most $\max_{t \in V(T)} |\chi(t) \setminus L| + \max_{t \in V(T),\, x \in V(T_t)} |\chi_t(x)| \le d + k + 2$.

Let us check that $(T_1, \chi_1)$ is a tree decomposition of $G$,
starting from condition (1).
If $v \in L$ then it belongs to exactly one set $\chi(t) \cap L$ and $\chi_1^{-1}(v) = \chi_t^{-1}(v)$.
If $v \not\in L$, then $\chi_1^{-1}(v) = \chi^{-1}(v) \cup \bigcup_{t \in \chi{-1}(v)} V(T_t)$.
In both cases these are connected subtrees of $T_1$.

Now we check condition (2) for $uv \in E(G)$.
If $u,v \in L$, then both $u,v$ belong to a single set $\chi(t) \cap L$ and there is a bag $\chi_t(x)$ containing $u,v$.
If $u,v \not\in L$, then both $u,v$ appear in some bag of $T$ and also in its counterpart in $T_1$. 
If $u \in L,\, v\not\in L$, then for some $t \in V(T)$ we have $u \in \chi(t) \cap L$, $v \in \chi(t) \setminus L$ and $u \in \chi_t(x)$ for some $x \in V(T_t)$.
Hence, $u,v \in \chi_1(x)$.
The conditions (3,4) are not applicable to a~standard tree decomposition.
\end{proof}

{When working with tree $\hh$-decompositions, we will often exploit the following structure of base components that follows straight-forwardly from the definition.}

\begin{observation} \label{obs:basecomponent:neighborhoods} 
Let~$(T,\chi,L)$ be a tree $\hh$-decomposition of a graph~$G$, for an arbitrary class~$\hh$.
\begin{itemize}
    \item For each set~$L^* \subseteq L$ for which~$G[L^*]$ is connected there is a unique node~$t \in V(T)$ such that~$L^*\subseteq \chi(t)$ while no vertex of~$L^*$ occurs in~$\chi(t')$ for~$t' \neq t$, and such that~$N_G(L^*) \setminus L \subseteq \chi(t) \setminus L$. 
    \item For each node~$t \in V(T)$ we have~$N_G(\chi(t) \cap L) \subseteq \chi(t) \setminus L$.
\end{itemize}
\end{observation}

\section{Constructing decompositions} \label{sec:decomposition}
\subsection{From separation to decomposition} \label{sec:decomp-abstract}
We begin with a basic notion of separation which describes a single step of cutting off a piece of the graph that belongs to~$\hh$; \bmp{recall that~$\hh$ is assumed throughout to be hereditary.}
Any base component together with its neighborhood forms a~separation.

\begin{definition}
For disjoint $C, S \subseteq V(G)$, the pair $(C,S)$ is called an $(\hh,k)$-\textbf{separation} in~$G$~if:
\begin{enumerate}
    \item $G[C] \in \hh$, 
    \item $|S| \le k$,
    \mic{ \item $N_G(C) \subseteq S$. }
\end{enumerate}
\end{definition}



\begin{definition}
For an \hsepk $(C, S)$ and set $Z \subseteq V(G)$,
we say that $(C,S)$ \textbf{covers} $Z$ if $Z \subseteq C$, or \textbf{weakly covers} $Z$ if $Z \subseteq C \cup S$.
Set $Z \subseteq V(G)$ is called $(\hh, k)$-\textbf{separable} if
there exists an \hsepk that covers $Z$.
Otherwise $Z$ is $(\hh, k)$-\textbf{inseparable}.
\end{definition}

We would like to decompose a graph into a~family of well-separated subgraphs from $\hh$, that could later be transformed into base components.
Imagine we are able to find an \hsepk covering $Z$ algorithmically, or conclude that $Z$ is $(\hh, k)$-inseparable.
We can start a~process at an arbitrary vertex $z$, set $Z = \{z\}$, and find a separation $(C,S)$ (weakly) covering $z$.
However, the entire set $C \cup S$ might be still a part of a larger base component.
To check that, we can now ask whether $Z = C \cup S$ is $(\hh, k)$-{separable} or not, and iterate this process until we reach a maximal separation, that is, one that cannot be further extended.

After this step is terminated, we consider the connected components of $G - S$.
The ones that are contained in $C$ are guaranteed to belong to $\hh$, so we do not care about them anymore.
In the remaining components, we try to repeat this process recursively, however some care is need to specify the invariants.
Let us formalize what kind of structure we expect at the end of this process.

\begin{definition}\label{def:separation-decomposition}
An $(\hh,k_1,k_2)$-\textbf{separation decomposition} of \bmp{a connected graph} $G$ is a rooted tree ${T}$, where each node $t \in V({T})$ is associated with a triple $(V_t, C_t, S_t)$ of subsets of $V(G)$, such that:
\begin{enumerate}
    \item the subsets $V_t$ are vertex disjoint, induce non-empty connected subgraphs, and sum up to $V(G)$,
    \item for each $t \in V({T})$ the pair $(C_t, S_t)$ is an \hsep{k_2} {in~$G$} and $V_t \subseteq C_t \cup S_t$, 
    \label{item:separation:k2}
    \item each edge $e \in E(G)$ is either contained inside some $G[V_t]$ or there exists $t_1, t_2 \in V({T})$, such that $t_1$ is an ancestor of $t_2$ and $e \in E(V_{t_1} \cap S_{t_1}, V_{t_2})$, \label{item:separation:ancestor}
    \item if $t$ is not a leaf in $T$, then $V_t$ is $(\hh,k_1)$-inseparable. \label{item:separation:inseparable}
\end{enumerate}
We obtain an $(\hh,k_1,k_2)$-\textbf{separation quotient graph} $G\, /\, T$ by contracting each subgraph $(V_t)_{t \in {V(T)}}$ into {a corresponding} vertex $t$.
\end{definition}

The idea behind this definition is to enforce that, when working with a graph with an (unknown) $\hh$-elimination forest $(\mathcal{T}_\hh, \chi_\hh)$\footnote{{We use the notation $\mathcal{T}_\hh$ to  distinguish it from the separation decomposition tree $T$.}} of depth $k_1$, each non-leaf node $t \in V(T)$ of an~$(\hh,k_1,k_2)$-{separation decomposition}~$T$ has some non-base vertex of  $(\mathcal{T}_\hh, \chi_\hh)$ in the set $V_t$.
This will allow us to make a connection between the parameters $(k_1, k_2)$ of the separation decomposition $T$ with the treedepth of the quotient graph \gsep.
The separations $(C_t, S_t)$ will {be} useful for going in the opposite direction, i.e., turning a standard elimination forest of \gsep into an $\hh$-elimination forest of~$G$.

\bmp{We remark that \emph{any} \mic{connected} graph~$G$ admits a $(\hh,k_1,k_2)$-separation decomposition, for \mic{any parameters $k_1 < k_2$}; hence one cannot deduce properties of~$G$ from the existence of a separation decomposition. If a graph~$G$ has bounded $\hhtw$ or $\hhdepth$, this will be reflected by the separation decomposition due to the quotient graph~$G\, /\, T$ having bounded treewidth or treedepth.}

{
We record the following observation for later use, which is implied by the fact that~$N_G(C_t) \subseteq S_t$ whenever~$(C_t, S_t)$ is an~$(\hh, k_1)$-separation.

\begin{observation} \label{obs:sepdec:connected:meets:s}
Let ${T}$ be an $(\hh,k_1,k_2)$-separation decomposition of a \mic{connected} graph $G$ and let~$t \in V(T)$. If~$Z \subseteq V(G)$ is a connected vertex set such that~$C_t \cap Z \neq \emptyset$ and~\jjh{$Z \not \subseteq C_t$}, then~$S_t \cap Z \neq \emptyset$.
\end{observation}
}

In order to construct a separation decomposition we need an algorithm for finding {an} $\hh$-separation with a moderate separator size.
In a typical case, we do not know how to find an \hsepk efficiently, but we will provide fast algorithms with approximate guarantees.
In the definition below, we relax not only the separator size but we also allow the algorithm to return a weak coverage.

\defparproblem{$(\hh,h)$-separation finding}{Integers $k \le t$, a graph $G$ of $\hh$-treewidth bounded by $t$, and a connected {non-empty} subset $Z \subseteq V(G)$.}{$t$}
{Either return an \hsep{h(t)} $(C,S)$ that {weakly covers} $Z$ or conclude that $Z$ is $(\hh,k)$-inseparable.}

\mic{In the application within this section the input integers $k$ and $t$ coincide, but we distinguish them in the problem definition \bmp{to facilitate a recursive approach to solve the problem} in later sections.}
We formulate an algorithmic connection between this problem and the task of constructing {a}~separation decomposition.
The proof is postponed to  Lemma~\ref{lem:restricted-decomposition}, which is a slight generalization of the claim below.

\begin{lemma}\label{lem:crown-decomposition}
Suppose there exists an algorithm $\mathcal{A}$ for \textsc{$(\hh,h)$-separation finding} running in time $f(n,t)$.
Then there is an algorithm that, given {an} integer $k$ and \mic{connected} graph $G$ of $\hh$-treewidth at most $k$,
runs in time $f(n,k) \cdot n^{\Oh(1)}$ and returns an {$(\hh,k, h(k)+1)$-separation decomposition} of $G$.
If $\mathcal{A}$ runs in polynomial space, then the latter algorithm does as well.
\end{lemma}

\mic{The connections between separation finding, separation decompositions, and constructing the final decompositions are sketched on \cref{fig:roadmap}.} \bmp{We remark that the problem of \textsc{$(\hh,h)$-separation finding} is related to the problem of finding a large \emph{$k$-secluded} connected induced subgraph that belongs to~$\hh$, i.e., a connected induced $\hh$-subgraph whose open neighborhood has size at most~$k$. When~$\hh$ is defined by a finite set of forbidden induced subgraphs, this problem is known to be fixed-parameter tractable in~$k$~\cite{GolovachHLM20} via the technique of recursive understanding, which yields a triple-exponential parameter dependence.}

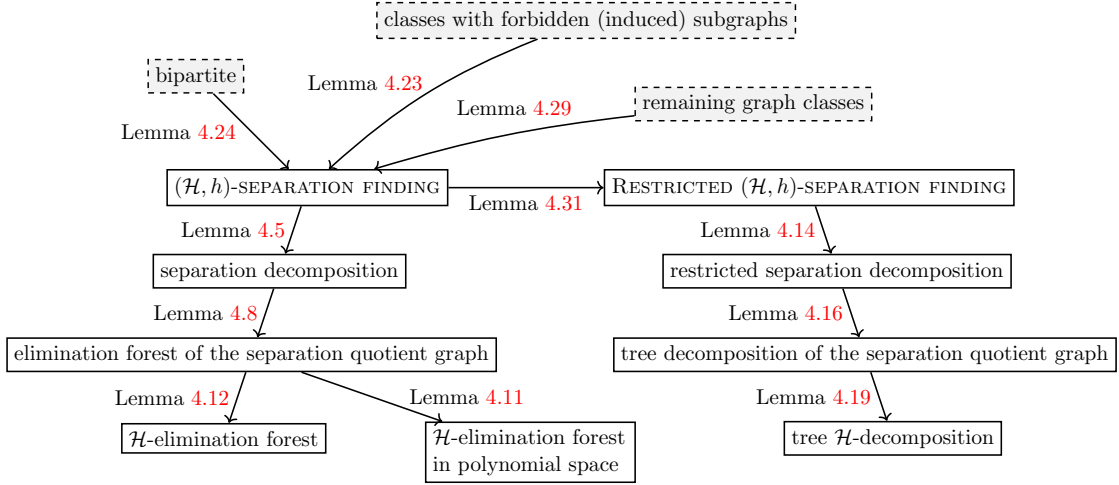
\begin{figure}
\input{roadmap}
  \caption{Roadmap for obtaining various types of decompositions. 
  {The arrows represent conceptual steps in solving the subroutines and constructing intermediate structures. 
    The function $h$ captures the approximation guarantee
    in the $(\hh,h)$-\textsc{separation finding} problem and governs the approximation guarantee in the process of building a~decomposition.
    For classes given by forbidden (induced) subgraphs and for $\hh = \mathsf{bipartite}$ we provide algorithms with $h(x) = \Oh(x)$, whereas in the general case we have $h(x) = \Oh(x^2)$.
    }}
    \label{fig:roadmap}
\end{figure}

\subsubsection{Constructing an \texorpdfstring{$\hh$}{H}-elimination forest}

Now we explain how to exploit a separation decomposition for constructing an \hhdepthdecomp{}.
The first step is to bound the treedepth of the separation quotient graph.
To do this, we observe that the vertices given by contractions of base components (or their subsets) always form an independent set.

\begin{lemma}\label{lem:independent}
Consider any $\hh$-elimination forest (resp. tree $\hh$-decomposition) of \mic{a connected graph} $G$ of depth $k_1$ (resp. width $k_1 - 1$)  with the set of base vertices $L$ and distinct nodes $t_1, t_2$ from {an} $(\hh,k_1,k_2)$-separation decomposition ${T}$ of $G$.
If $V_{t_1}, V_{t_2} \subseteq L$, then $t_1t_2 \not\in E(\gsep)$.
\end{lemma}
\begin{proof}
In both types of decomposition each $V_{t_i}$ {for $i \in \{1,2\}$} must be located in {a} single base component because it is connected, by Observation~\ref{obs:basecomponent:neighborhoods}.
The neighborhood of this base component has size at most $k_1$
so $V_{t_1}, V_{t_2}$ are $(\hh,k_1)$-separable. 
By property (\ref{item:separation:inseparable}) of Definition~\ref{def:separation-decomposition}, $t_1$ and $t_2$ must be leaves in $T$ and by property (\ref{item:separation:ancestor}) we have that $E(V_{t_1}, V_{t_2}) = \emptyset$.
\end{proof}

{Since the edges in an \hhdepthdecomp{} can only connect vertices in ancestor-descendant relation, we obtain the following observation that allows us to modify the decomposition accordingly to subgraph contraction. }

\begin{observation}\label{obs:elimination-forest-connected-set}
Suppose $(\mathcal{T}, \chi)$ {is} an \hhdepthdecomp{} of \mic{a connected graph} $G$
and $A \subseteq V(G)$ induces a~connected subgraph of $G$.
Then the set $Y_A = \{x \in V(\mathcal{T}) \mid \chi(x) \cap A \ne \emptyset\}$ contains a~node being an ancestor of all the nodes in $Y_A$.
\end{observation}

\begin{lemma}\label{lem:quotient-depth}
Let ${T}$ be an $(\hh,k_1,k_2)$-separation decomposition of a \mic{connected} graph~$G$ with $\hhdepth(G) \le k_1$.
Then the treedepth of \gsep is at most $k_1+1$.
\end{lemma}
\begin{proof}
Consider an {optimal} \hhdepthdecomp{} $(\mathcal{T}_\hh, \chi_\hh)$ of $G$.
For a node $t \in V(T)$ of the separation decomposition $T$
we define {a} set $Y_t \subseteq V(\mathcal{T}_\hh)$ as $Y_t = \{x \in V(\mathcal{T}_\hh) \mid \chi_\hh(x) \cap V_t \ne \emptyset\}$. 
By Observation~\ref{obs:elimination-forest-connected-set},
there must be a node $y_t \in Y_t$ being an ancestor for all the nodes in $Y_t$.

Consider a new function $\chi_1 \colon V(\mathcal{T}_\hh) \to 2^{V(T)}$ defined as $\chi_1(x) = \{t \in V(T) \mid y_t = x\}$, {so that the bag~$\chi_1(x)$ consists of precisely those~$t \in V(T)$ for which~$x$ is the highest node of~$\mathcal{T}_\hh$ whose bag intersects~$V_t$}. 
We want {to} use it for constructing an elimination forest of \gsep by identifying the vertices of \gsep with $V(T)$.
{If $E(V_{t_1}, V_{t_2}) \ne \emptyset$, then Observation~\ref{obs:elimination-forest-connected-set} applied to $V_{t_1} \cup V_{t_2}$ implies that $y_{t_1}, y_{t_2}$ are in ancestor-descendant relation in  $(\mathcal{T}_\hh, \chi_\hh)$, and therefore also in $(\mathcal{T}_\hh, \chi_1)$.}
Since the sets {$(V_t)_{t \in V(T)}$} are vertex-disjoint {and~$|\chi_\hh(x)| = 1$ for non-leaf nodes~$x$ of~$\mathcal{T}_\hh$}, the relation $y_{t_1} = y_{t_2} = x$ is only possible if $x$ is a~leaf node in $\mathcal{T}_\hh$.
{The structure~$(\mathcal{T}_\hh, \chi_1)$ is almost an elimination forest of~$G/T$,} except for the fact that the leaf nodes have non-singleton bags.
However, if $x=y_t$ is a leaf node, then the set $V_t$ consists of base vertices {of $(\mathcal{T}_\hh, \chi_\hh)$}. 
By Lemma~\ref{lem:independent}, there are no edges in \gsep between such vertices.
We create a~new leaf node $x_t$ for each such $t \in \chi_1(x)$, remove $x$, replace it with the singleton leaves $(x_t)$, and set $\chi_1(x_t) = t$.

If the constructed elimination forest has any node with $\chi_1(x) = \emptyset$, we can remove it and connect its children directly to its parent.
{As in a~standard elimination forest we require that the base components are empty, we just need to add auxiliary leaves with empty bags.}
Hence, we have constructed an elimination forest of \gsep of depth at most one larger than the depth of $(\mathcal{T}_\hh, \chi_\hh)$.
\end{proof}

We have bounded the treedepth of \gsep so now we could employ known algorithms to construct an elimination forest of \gsep of moderate depth.
After that, we want to go in the opposite direction and construct an \hhdepthdecomp{} of $G$ relying on the given elimination forest of \gsep.

\begin{lemma}\label{lem:treedepth-to-h-depth}
Suppose we are given a \mic{connected} graph $G$, an $(\hh,k_1,k_2)$-separation decomposition ${T}$ of $G$, and an~elimination forest of \gsep of depth $d$.
Then we can construct an $\hh$-elimination forest of $G$ of depth $d\cdot k_2$ in polynomial time.
\end{lemma}
\mic{\begin{proof}
We give a proof by induction, replacing each node in the~elimination forest of \gsep with a corresponding separator from $T$.
However, when removing vertices from $G$ we might break some invariants of an $(\hh,k_1,k_2)$-separation decomposition. 
Therefore, we formulate a weaker invariant,
which is more easily maintained. 

\begin{claim}\label{lem:treedepth-to-h-depth:induction}
Suppose that a (not necessarily connected) graph $H$ satisfies the following:
\begin{enumerate}
    \item $V(H)$ admits a partition into a family of connected non-empty sets $(V_i)_{i=1}^\ell$,
    \item for each $i \in [\ell]$ there is an $(\hh,k)$-separation $(C_i, S_i)$ weakly covering $V_i$ in $H$, and
    \item the quotient graph $H'$ obtained from $H$ by contracting each set $V_i$ into a single vertex, has treedepth at most $d$.
\end{enumerate}
Then $\hhdepth(H) \le d \cdot k$.
Furthermore, \bmp{an $\hh$-elimination forest of depth at most~$d \cdot k$} can be computed in polynomial time when given the sets $(V_i, C_i, S_i)_{i=1}^\ell$ and the elimination forest of $H'$.
\end{claim}
\begin{innerproof}
We will prove the claim by induction on $d$.
The case $d = 0$ is trivial as the graph $H$ is empty.

Suppose $d \ge 1$.
The sets $(V_i)_{i=1}^\ell$ are connected so we get a partition for each connected component of $H$.
It suffices to process each connected component of $H$, so we can assume that \bmp{$H$ is connected. Hence} the given elimination forest of $H'$ has a single root $j$. We refer to the vertices of $H'$ by the indices from $[\ell]$.
Consider the graph $H_j = H - V_j$: we are going to show that it satisfies the inductive hypothesis for $d - 1$. 
Clearly, the sets $\{V_i \mid i \in [\ell],\, i \ne j\}$ 
form a partition of $V(H_j)$; they are connected and non-empty.
For $i \ne j$, we obtain an \hsepk $(C_i \setminus V_j, S_i \setminus V_j)$ weakly covering $V_i$ in $H_j$. \bmp{(Recall that~$\hh$ is hereditary.)}
Finally, the quotient graph, obtained from $H_j$ by contracting the sets in the partition, is $H' - j$.
The elimination forest of $H' - j$ obtained by removing the root from the elimination forest of $H'$ has depth at most $d-1$.

By the inductive hypothesis, we obtain an $\hh$-elimination forest of $H_j$ of depth at most $(d-1) \cdot k$.
The graph $H - (C_j \cup S_j)$ is a subgraph of $H_j$, so we can easily turn the $\hh$-elimination forest of $H_j$ into an $\hh$-elimination forest $(T_j,\chi_j)$ of $H - (C_j \cup S_j)$.
We start constructing an $\hh$-elimination forest of $H$ with a rooted path $P_j$ consisting of vertices of $S_j$, in arbitrary order.
We make the roots of $(T_j,\chi_j)$ children to the lowest vertex on $P_j$.
It remains to handle the connected components of $H-S_j$ lying inside $C_j$.
They all belong to $\hh$ and $N_H(C_j) \subseteq S_j$ so we can turn them into leaves, also attached to the lowest vertex on $P_j$.
The depth of such a decomposition is bounded by $k-1$ (the number of edges in $P_j$) plus $1 + (d-1)\cdot k$, which gives $d \cdot k$.
\end{innerproof}

Given an $(\hh,k_1,k_2)$-separation decomposition of $G$, we can directly 
apply
\cref{lem:treedepth-to-h-depth:induction} with $k = k_2$.
This concludes the proof.
\end{proof} }

\begin{lemma}\label{lem:decomp-ed-polyspace}
Suppose there exists an algorithm $\mathcal{A}$ for \textsc{$(\hh,h)$-separation finding} running in time $f(n,t)$.
Then there is an algorithm that, given graph $G$ with $\hh$-elimination distance $k$,
runs in time $f(n,k) \cdot n^{\Oh(1)}$, and returns an $\hh$-elimination forest of $G$ of depth $\Oh(h(k)\cdot k^2\log^{3 / 2} k)$.
If $\mathcal{A}$ runs in polynomial space, then the latter algorithm does as well.
\end{lemma}
\begin{proof}
\mic{It suffices to process each connected component of $G$ independently, so we can assume that $G$ is connected.}
By Lemma~\ref{lem:treedepth-treewidth}, we know that $\hhtw(G) \le k$ so we can apply 
Lemma~\ref{lem:crown-decomposition} to find an $(\hh,k,h( k)+1)$-separation decomposition ${T}$ in time $f(n,k) \cdot n^{\Oh(1)}$.
If $\mathcal{A}$ runs in polynomial space, then the construction in Lemma~\ref{lem:crown-decomposition} preserves this property.
The graph \gsep is guaranteed by Lemma~\ref{lem:quotient-depth} to have treedepth bounded by $k+1$.
We run the polynomial-time approximation algorithm for treedepth, which returns an~elimination forest of \gsep with depth $\Oh(k^2 \log^{3 / 2} k)$~\cite{CzerwinskiNP19}.
We turn it into an $\hh$-elimination forest of $G$ of depth $\Oh(h(k)\cdot k^2\log^{3 / 2} k)$ in polynomial time with Lemma~\ref{lem:treedepth-to-h-depth}.
\end{proof}

By replacing the approximation algorithm for treedepth with an exact one, running in time $2^{\Oh(k^2)}\cdot n^{\Oh(1)}$~\cite{ReidlRVS14}, we can improve the approximation ratio.
We lose the polynomial space guarantee, though.

\begin{lemma}\label{lem:decomp-ed-exact}
Suppose there exists an algorithm $\mathcal{A}$ for \textsc{$(\hh,h)$-separation finding} running in time $f(n,t)$.
Then there is an algorithm that,
given graph $G$ with $\hh$-elimination distance $k$,
runs in time $\br{f(n,k) + 2^{\Oh(k^2)}} \cdot n^{\Oh(1)}$, and returns an $\hh$-elimination forest of $G$ of depth $\Oh(h(k) \cdot k)$.
\end{lemma}

\subsubsection{Restricted separation}
\label{sec:decomposition-restricted-separation}

\bmp{To compute tree $\hh$-decompositions, we will follow the global approach of Lemma~\ref{lem:decomp-ed-polyspace}.}
\bmp{An analog of Lemma~\ref{lem:quotient-depth} also works to bound the treewidth of a quotient graph~$G/T$ in terms of $\hhtw(G)$, as will be shown in Lemma~\ref{lem:quotient-width}.}
However, the provided structure of {a} separation decomposition turns out {to be} too weak to make a counterpart of Lemma~\ref{lem:treedepth-to-h-depth} work for tree $\hh$-decompositions.
\bmp{This is because the separations~$(C_t, S_t)$ given by property (\ref{item:separation:ancestor}) may intersect each other,  which can cause the neighborhood of a set~$V_t \setminus S_t \subseteq C_t$ to be arbitrarily large even though~$|N_G(C_t)| \leq k_2$ by property (\ref{item:separation:k2}). 
We therefore need a stronger property to ensure that~$|N_G(V_t \setminus S_t)|$ does not become too large.}

\begin{definition} \label{def:restricted:separation:decomp}
We call {an} $(\hh,k_1,k_2)$-separation decomposition \textbf{restricted} if it satisfies {the} additional property: for each $t \in V(T)$ there are at most $k_1$ ancestors $s$ of $t$ such that $V_{s} \subseteq C_t$.
\end{definition}

{When building a separation decomposition by repeatedly extracting separations, this additional property states that a vertex set that eventually becomes the separated~$C_t$-side of an $(\hh,k_2)$-separation cannot fully contain too many distinct sets~$V_s$ handled by earlier separations.}
{Intuitively, if we find an~$(\hh,k_2)$-separation $(C_t,S_t)$ that covers $V_s$ for $s$ being an ancestor of $t$, then we could have used it earlier to replace $V_s$ with a~larger set.
However, we have no guarantee that an~$(\hh,k_2)$-separation covering $V_s$ would be found when processing the node $s$, but it can be found later, when processing $t$.
\bmp{Since $V_s$ is guaranteed to be $(\hh,k_1)$-inseparable, this can only happen if $k_1 < k_2$}. 
Unfortunately, we cannot enforce $k_1 = k_2$, mostly because our algorithms for finding separations are approximate. \bmp{However,} we are able to ensure that the number of such sets $V_s$ is small.
}

We introduce a restricted version of the \textsc{$(\hh,h)$-separation finding} problem, tailored for building restricted separation decompositions.

\defparproblem{Restricted $(\hh,h)$-separation finding}{Integers $k \le t$, a graph $G$ of $\hh$-treewidth bounded by $t$, a connected {non-empty} subset $Z \subseteq V(G)$, a family $\mathcal{F}$ of connected, disjoint, $(\hh,k)$-inseparable subsets of $ V(G)$.}{$t$}
{Either return an \hsep{h(t)} $(C,S)$ that {weakly covers} $Z$ {such that} $C$ contains \mic{at most $t$} sets from $\mathcal{F}$, or conclude that $Z$ is $(\hh,k)$-inseparable.}

We are ready to present a general way of constructing a separation decomposition, restricted or not, when supplied by an algorithm for the respective version of $(\hh,h)$-\textsc{separation finding}.
In~particular, this generalizes {and proves}  Lemma~\ref{lem:crown-decomposition}.

\begin{lemma}\label{lem:restricted-decomposition}
Suppose there exists an algorithm $\mathcal{A}$ for \textsc{(Restricted) $(\hh,h)$-separation finding} running in time $f(n,t)$.
Then there is an algorithm that, given an integer $k$ and \mic{a connected graph} $G$ of $\hh$-treewidth at most $k$,
runs in time $f(n,k) \cdot n^{\Oh(1)}$ and returns a (restricted) $(\hh,k, h(k)+1)$-separation decomposition of $G$.
If $\mathcal{A}$ runs in polynomial space, then the latter algorithm does as well.
\end{lemma}
\begin{proof}
We will construct the desired separation decomposition by a recursive process, aided by algorithm~$\mathcal{A}$. In the recursion, we have a partial decomposition covering a subset~$X$ of~$V(G)$, and indicate a connected subset~$H \subseteq V(G) \setminus X$ to be decomposed. Before describing the construction, we show how~$\mathcal{A}$ can be used to efficiently find an \emph{extremal} separation that covers a given subset~$Z$ of~$H$. The extremal property consists of the fact that the output separation weakly covers a set~$Z' \supseteq Z$ that is itself~$(\hh,k)$-inseparable, or all of~$H$.

\begin{claim}
In the time and space bounds promised by the lemma we can solve:

\defparproblem{Extremal (restricted) $(\hh,h(k)+1)$-separation finding}{A graph $G$ of $\hh$-treewidth bounded by $k$, an integer~$k$, a connected non-empty set~$Z \subseteq H \subseteq V(G)$, an~$(\hh, h(k)+1)$-separation~$(C,S)$ in~$G$ weakly covering~$Z$. In the restricted case, additionally a family $\mathcal{F}$ of connected, disjoint, $(\hh,k)$-inseparable subsets of $V(G)$ \bmp{such that~$C$ contains at most~$k$ sets from~$\mathcal{F}$}.}{$k$}
{Output a triple~$(Z',C',S')$ such that $(C', S')$ is an~$(\hh, h(k)+1)$-separation in~$G$ that weakly covers the connected set~$Z' \supseteq Z$ which is either equal to~$H$ or is~$(\hh,k)$-inseparable and contained in~$H$. In the restricted case, ensure \bmp{that~$C'$ contains at most~$k$ sets from~$\mathcal{F}$}.} 
\end{claim}
\begin{innerproof}
Run~$\mathcal{A}$ on~$G$,~$k$, and~$Z$ (and~$\mathcal{F}$, in the restricted case), where we additionally set~$t=k$. 
If~$\mathcal{A}$ reports that~$Z$ is $(\hh,k)$-inseparable: output~$(Z,C,S)$ unchanged. Otherwise, ~$\mathcal{A}$ outputs an $(\hh, h(k))$-separation~$(C',S')$ with~$Z \subseteq C' \cup S'$ (\bmp{with~$C'$ containing} at most~$k$ sets from~$\mathcal{F}$ in the restricted case). Note that~$|S'| \leq h(k)$ while~$|S| \leq h(k)+1$. Let~$Z'$ be the connected component of~$G[H \cap (C' \cup S')]$ that contains~$Z$.
If $Z' = H$, then we may simply output~$(Z', C', S')$. If $Z' \subsetneq H$, then choose an arbitrary vertex~$v \in N_G(Z') \cap H$, which exists since~$H$ is connected. Recursively solve the problem of covering~$Z' \cup \{v\}$ starting from the separation~$(C', S' \cup \{v\})$ {(and~$\mathcal{F}$, in the restricted case)} and output the result. Since the set~$Z$ strictly grows in each iteration, the recursion depth is at most~$n$.
\end{innerproof}

\begin{figure}[tb]
    \centering
    \includegraphics[scale=0.9]{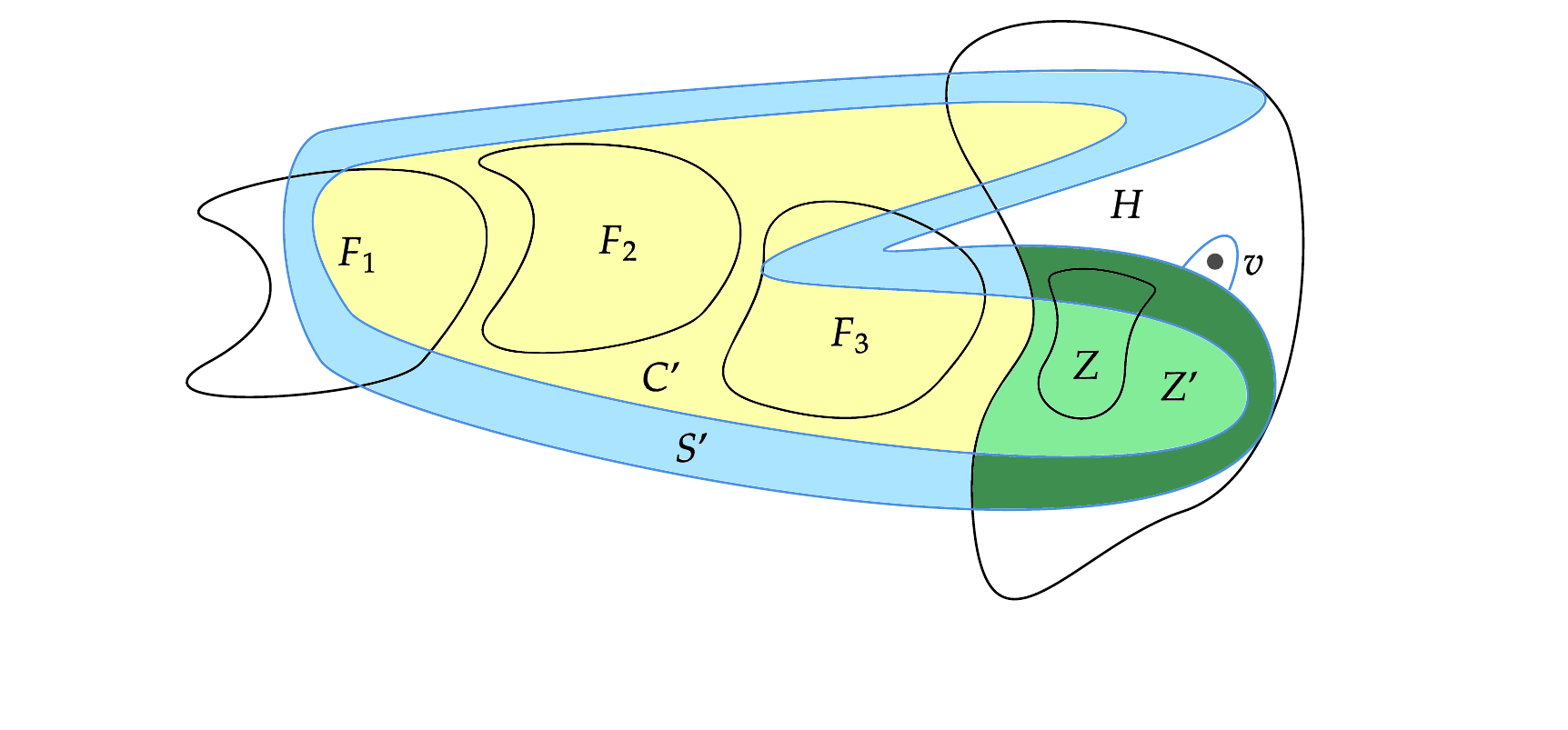}
    \caption{Illustration {of} the proof of Lemma \ref{lem:restricted-decomposition}: a
    call to \textsc{Extremal (restricted) $(\hh,h(k)+1)$-separation finding}
    with the family $\mathcal{F}_t = \{F_1, F_2, F_3\}$.
    An $(\hh,h(k))$-separation $(C', S')$ weakly covers $Z$.
    The set $Z'$ is the entire green area (both dark and light).
    We add vertex $v \in N_G(Z') \cap H$ to $Z'$ to ensure progress of the process.
    The dark green area, i.e., $Z' \cap S'$ forms a separator between $Z' \setminus S'$ and $H \setminus Z'$.
    The set $C'$ contains $F_2$ but in the restricted setting we require that contains at most $k$ such sets.
    }
    \label{fig:separation-decompostion}
\end{figure}

Using the algorithm for finding extremal separations, we  construct the desired separation decomposition recursively. To formalize the subproblem solved by a recursive call, we need the notion of a \emph{partial} (restricted)~$(\hh, k_1, k_2)$-separation decomposition of~$G$. This is a rooted tree~$T$ where each node~$t \in V(T)$ is associated with a triple~$(V_t, C_t, S_t)$ of subsets of~$V(G)$ such that:

\begin{enumerate}
    \item the subsets $V_t$ are vertex disjoint, induce non-empty connected subgraphs of~$G$, and sum up to a subset~$X \subseteq V(G)$ which are the vertices \emph{covered} by the partial separation decomposition,
    \item for each $t \in V({T})$ the pair $(C_t, S_t)$ is an \hsep{k_2} {in~$G$} and $V_t \subseteq C_t \cup S_t$, 
    \item each edge $e \in E(G[X])$ is either contained inside some $G[V_t]$ or there exists $t_1, t_2 \in V({T})$, such that $t_1$ is an ancestor of $t_2$ and $e \in E(V_{t_1} \cap S_{t_1}, V_{t_2})$, \label{item:partial:separation:ancestor}
    \item if $t$ is not a leaf in $T$, then $V_t$ is $(\hh,k_1)$-inseparable,
    \item if~$X \neq \emptyset$, then {for each connected component~$H$ of~$G - X$}, there is a node~$t^* \in V(T)$ such that~$N_G(H) \subseteq \bigcup _{t \in A_{t^*}} (V_t \cap S_t)$, where~$A_{t^*}$ are the ancestors of~$t^*$ in~$T$,
    \label{item:partial:separation:inseparable}
    \item additionally, in the case of a restricted partial separation decomposition: for each $t \in V(T)$ there are at most $k_1$ ancestors $s$ of $t$ such that $V_{s} \subseteq C_t$.
\end{enumerate}     

Note that if the set~$X$ of vertices covered by a partial (restricted)~$(\hh, k_1, k_2)$-separation decomposition is equal to~$V(G)$, then it is a standard (restricted) $(\hh, k_1, k_2)$-separation decomposition following Definitions~\ref{def:separation-decomposition} and~\ref{def:restricted:separation:decomp}.

Using this notion we can now state the subproblem that we solve recursively to build a (restricted) $(\hh, k, h(k)+1)$-separation decomposition of~$G$. The input is a partial (restricted) $(\hh, k, h(k)+1)$-separation decomposition covering some set~$X \subseteq V(G)$ and a connected vertex set~$H \subseteq V(G) \setminus X$. The output is a partial (restricted) separation decomposition covering~$X \cup H$.

The algorithm solving this subproblem is as follows. Create a new node~$t$ in the decomposition tree. If~$X \neq \emptyset$, then \mic{choose}~$t^* \in V(T)$ such that~$N_G(H) \subseteq \bigcup _{t \in A_{t^*}} (V_t \cap S_t)$ and make~$t$ a child of~$t^*$; if~$X = \emptyset$ then~$t$ becomes the root of the tree. Let $\mathcal{F}_t$ denote the family of sets $V_{t_i}$ for the ancestors $(t_i)$ of node $t$. Let~$v$ be an arbitrary vertex from~$H$. Note that~$(C = \emptyset, S = \{v\})$ is an~$(\hh, h(k)+1)$-separation in~$G$ weakly covering~$Z := \{v\} \subseteq H$, so that we can use~$(C,S)$ as input to \textsc{Extremal (restricted) $(\hh, h)$-separation finding}. In the restricted case, we also give the family~$\mathcal{F}_t$ as input \mic{(note that setting $C = \emptyset$ meets the precondition for the restricted case)}. Let~$(C',S')$ be the resulting~$(\hh, h(k)+1)$-separation, weakly covering a connected set~$Z' \subseteq H$ which is equal to~$H$ or~$(\hh,k)$-inseparable. We set~$(V_t, C_t, S_t) := (Z', C', S')$, thereby adding~$Z'$ to the set of vertices covered by the partial decomposition. If~$Z' = H$ then the node~$t$ becomes a leaf of the decomposition and we are done. If~$Z' \subsetneq H$, then let~$H'_1, \ldots, H'_m$ be the connected components of~$G[H \setminus Z']$. One by one we recurse on~$H'_i$ to augment the decomposition tree into one that additionally covers~$H'_i$. Based on the guarantees of the subroutine to find extremal separations and the properties of partial separation decompositions, it is straight-forward to verify correctness of the algorithm.

\mic{To obtain the desired (restricted) $(\hh, k, h(k)+1)$-separation decomposition of~$G$,
it suffices to solve the subproblem starting from an empty decomposition tree covering the empty set~$X$, for~$H = V(G)$.}
As the overall algorithm performs~$n^{\Oh(1)}$ calls to~$\mathcal{A}$ and otherwise consists of operations that run in polynomial time and space, the claimed time and space bounds follow.
\end{proof}

\subsubsection{Constructing a tree $\hh$-decomposition}

We proceed analogously as for constructing an \hhdepthdecomp{}.
We begin with bounding the treewidth of \gsep, so we could employ existing algorithms to find a tree decomposition {for it}.
Then we take advantage of restricted separations to build a \hhtwdecomp{} of $G$ by modifying the tree decomposition of \gsep. 

\begin{lemma}\label{lem:quotient-width}
Let ${T}$ be an $(\hh,k_1,k_2)$-separation decomposition of \mic{a connected graph} $G$ with $\hhtw(G) \leq k_1-1$.
Then the treewidth of \gsep is at most $k_1$.
\end{lemma}
\begin{proof}
Let $\mathcal{E}$ denote the class of edge-less graphs.
We use notation $\mathcal{T}_\hh$ to refer to the tree in a tree $\hh$-decomposition, in order to distinguish it from the separation decomposition tree $T$.
We are going to transform a tree $\hh$-decomposition $(\mathcal{T}_\hh, \chi_\hh, L_\hh)$ of $G$ of width $k_1$ into
a tree $\mathcal{E}$-decomposition $(\mathcal{T_E}, \chi_\mathcal{E}, L_\mathcal{E})$ of \gsep, with $\mathcal{T_E} = \mathcal{T}_\hh$.
{Recall that~$V(G / T) = V(T)$.} 
For a vertex $v \in V(G)$ let $t_v \in V(T)$ denote the node for which $v \in V_{t_v}$.
For $x \in V(\mathcal{T}_\mathcal{E})$, we define $\chi_\mathcal{E}(x) = \{t_v \mid v \in \chi_\hh(x)\}$.
Let $L_\mathcal{E}$ be the set of these nodes $t \in V(T)$ for which $V_t \subseteq L_\hh$.

We need to check that this construction satisfies the properties (\ref{item:tree:h:decomp:connected}){--}(\ref{item:tree:h:decomp:base}) of a tree $\mathcal{E}$-decomposition.
If $t_1t_2 \in E(\gsep)$, then there are $v_1, v_2 \in V(G)$ such that $v_1v_2 \in E(G),\, v_1 \in V_{t_1},\, v_2 \in V_{t_2}$.
Hence, there must be a bag $\chi_\hh(x)$ containing both $v_1, v_2$, so the bag $\chi_\mathcal{E}(x)$ contains both $t_1,t_2$.
The bags containing {an arbitrary} $t \in V(\gsep)$ {consist of the union} of connected subtrees, {that is,} $\chi_\mathcal{E}^{-1}(t) = \{\chi_\hh^{-1}(v) \mid v \in V_t\}$.  
Since $V_t$ is connected, this union is also a connected subtree.
To see property (\ref{item:tree:h:decomp:unique}), consider $t \in V(\gsep)$ with $V_t \subseteq L_\hh$.
By Observation~\ref{obs:basecomponent:neighborhoods}, a~connected subset included in $L_\hh$ must reside in a single bag $\chi_\hh(x)$.
Therefore $t$ belongs only to its counterpart: $\chi_\mathcal{E}(x)$.
Finally, by Lemma~\ref{lem:independent}, $L_\mathcal{E}$ is an independent set in \gsep, so $(\gsep)[\chi_\mathcal{E}(x) \cap L_\mathcal{E}] \in \mathcal{E}$ for any $x \in \mathcal{T_E}$.
Hence, we have obtained a tree $\mathcal{E}$-decomposition.

For a node $x \in V(\mathcal{T_E})$ and $t \in \chi_\mathcal{E}(x) \setminus L_\mathcal{E}$, the set $V_t$ must include a vertex from $\chi_\hh(x) \setminus L_\hh$.
Therefore $|\chi_\mathcal{E}(x) \setminus L_\mathcal{E}| \le |\chi_\hh(x) \setminus L_\hh| \le k_1$.
This means that the width of $(\mathcal{T_E}, \chi_\mathcal{E}, L_\mathcal{E})$ is at most $k_1 - 1$.
Since the edge-less graphs have treewidth 0, the claim follows from Lemma \ref{lem:treewidth-of-hh}.
\end{proof}

The following {lemma} explains the role of the restricted separation decompositions. \bmp{Effectively, the additional property of being restricted implies that for each node~$t \in V(T)$ of an $(\hh,k_1,k_2)$-separation decomposition~$T$, we not only know that~$C_t$ induces an $\hh$-subgraph with a neighborhood of small size ($k_2$), we can also infer that the set~$V_t \setminus S_t = V_t \cap C_t$ induces an $\hh$-subgraph with a neighborhood whose size is bounded in terms of~$k_1$ and~$k_2$.}

\begin{lemma}\label{lem:ancestors}
Let~$t$ be a node in a restricted $(\hh,k_1,k_2)$-separation decomposition~$T$ of a \mic{connected} graph~$G$ and let $A_t \subseteq V(T)$ be the set of nodes $t_i \ne t$ such that $E_G(V_{t_i}, V_t \setminus S_t)$ is non-empty. The following holds.
\begin{enumerate}[(i)]
    \item {$|A_t| \leq k_1 + k_2$.} \label{ancestors:count}
    \item {$N_G(V_t \setminus S_t) \subseteq (V_t \cap S_t) \cup \bigcup_{s \in A_t} (V_s \cap S_s)$.} \label{ancestors:at}
    \item {The pair~$(V_t \setminus S_t, N_G(V_t \setminus S_t))$ is an~$(\hh, k_2(k_1 + k_2 + 1)$-separation in~$G$.}\label{ancestors:separation}
\end{enumerate}
\end{lemma}
\begin{proof}
\textbf{\eqref{ancestors:count}} By property (\ref{item:separation:ancestor}) all nodes {$t_i \in A_t$} must be ancestors of $t$. {Note that~$V_t \setminus S_t \subseteq C_t$, so that~$N_G(V_t \setminus S_t) \subseteq C_t \cup S_t$ since~$(C_t, S_t)$ is an~$(\hh,k_2)$-separation.}

As the decomposition is restricted, there can be at most $k_1$ many $t_i$ for which $V_{t_i} \subseteq C_t$.
{Since each set~$V_{t_i}$ induces a connected subgraph of~$G$, Observation~\ref{obs:sepdec:connected:meets:s} implies that} any other $t_i$ for which $E(V_{t_i}, V_t \setminus S_t) \ne \emptyset$
must satisfy $V_{t_i} \cap S_t \ne \emptyset$. Since $|S_t| \le k_2$ and the sets $V_{t_i}$ are vertex-disjoint,
the claim follows.

\textbf{\eqref{ancestors:at}} {
As the sets~$V_s$ in a separation decomposition sum to~$V(G)$, each neighbor of~$V_t \setminus S_t$ either belongs to~$V_t \cap S_t$ or belongs to a set~$V_s$ for~$s \neq t$, implying~$s \in A_t$. 
}


\textbf{\eqref{ancestors:separation}} {As~$V_t \setminus S_t \subseteq C_t$, we have~$G[V_t \setminus S_t] \in \hh$ since~$G[C_t] \in \hh$ and~$\hh$ is hereditary. Since~\eqref{ancestors:at} shows that~$N_G(V_t \setminus S_t)$ is contained in the union of~$1 + |A_t| \leq 1 + k_1 + k_2$ sets~$S_s$, each of which has size at most~$k_2$, this proves the claim.} 
\end{proof}


{The separation guarantee of the preceding lemma allows us to transform tree decompositions of~$G/T$ for restricted separation decompositions~$T$, into tree $\hh$-decompositions of~$G$.}

\begin{lemma}\label{lem:treewidth-to-h-width}
Let ${T}$ be a restricted $(\hh,k_1,k_2)$-separation decomposition of
\mic{a connected graph} $G$.
Suppose we are given a tree decomposition of \gsep of width $d-1$.
Then we can construct a tree $\hh$-decomposition of $G$ of width $d\cdot k_2 \cdot (k_1+k_2+1)$ in polynomial time.
\end{lemma}
\begin{proof}
For a node $t \in V(\gsep)$, let $A_t$ be the set of nodes $t_i \ne t$ such that $E(V_{t_i}, V_t \setminus S_t)$ is non-empty.
By Lemma~\ref{lem:ancestors} we have $|A_t| \le k_1 + k_2$.

Given a tree decomposition $(\mathcal{T}_\mathbf{tw}, \chi_\mathbf{tw})$ of \gsep of width $d-1$, 
we define $Q = \bigcup_{t \in {V(T)}} (V_t \cap S_t)$.
We reuse the same tree $\mathcal{T}_\hh = \mathcal{T}_\mathbf{tw}$, define $\chi_\hh(x) = \bigcup_{t \in \chi_\mathbf{tw}(x)} \left((V_t \cap S_t) \cup \bigcup_{s \in A_t} (V_s \cap S_s)\right)$,
and $L = V(G) \setminus Q$.
In other words, we replace $t$ with the ``downstairs'' separator $V_t \cap S_t$ and ``upstairs'' separator $\bigcup_{s \in A_t} (V_s \cap S_s)$.
\bmp{Lemma~\ref{lem:ancestors} ensures that} each constructed bag has size at most $\max_{x \in V(\mathcal{T}_\mathbf{tw})} |\chi_\mathbf{tw}(x)|\cdot \max_{t \in V(T)} |S_t| \cdot (\max_{t \in V(T)} |A_t| + 1) \le d\cdot k_2\cdot (k_1 +k_2+1)$.

First, we argue that for any $v \in Q$ the set of bags containing $v$ forms a connected subtree of $\mathcal{T}_\hh$.
\bmp{For~$v \in Q$, let $t_v$ be the node of~$T$} with $v \in V_t \cap S_t$ containing $v$.
The vertex $v$ appears in the bags which contained $t_v$ (then $v$ is a part of the ``downstairs'' separator) or some $t$ with $t_v \in A_t$  (then $v$ is a part of the ``upstairs'' separator). 
Formally, we define $B_v = \{t \in V(\gsep) \mid t_v \in A_{t} \}$.
Note that $B_v \subseteq N_{\gsep}(t_v)$.
We have $\chi_\hh^{-1}(v) = \chi_\mathbf{tw}^{-1}(t_v) \cup \bigcup_{t \in B_v}\chi_\mathbf{tw}^{-1}(t)$. 
Hence, $v$ \mic{appears} in bags where either $t_v$ was located or some of its neighbors from $B_v$. 
This gives a sum of connected subtrees that have non-empty intersections with $\chi_\mathbf{tw}^{-1}(t_v)$. 

To finish the construction, we append the connected components of $G[L] = G - Q$, that is, the subgraphs $V_t \setminus S_t$.
Since $V_t \subseteq C_t \cup S_t$ and $C_t \in \hh$, then also  $V_t \setminus S_t \in \hh$.
{By Lemma~\ref{lem:ancestors}, each set $N_G(V_t \setminus S_t)$ is contained in~$(V_t \cap S_t) \cup \bigcup_{s \in A_t} (V_s \cap S_s)$, i.e., in the set of vertices that was put in place of $t$.} 
We can thus choose any node $x \in \chi_\mathbf{tw}^{-1}(t)$, make a new node $x_t$, connect $x_t$ to $x$, and set $\chi_\hh(x_t) = \chi_\hh(x) \cup (V_t \setminus S_t)$.
This preserves the property that sets $\chi_\hh^{-1}(v)$ are connected.

Finally, we argue that for any edge $uv \in E(G)$, there is a node $x$ such that $\{u,v\} \subseteq \chi_\hh(x)$.
If $u \in V_t \setminus S_t$ for some $t$, then
either $uv \in E(G[V_t])$ or $uv \in E(V_{s} \cap S_{s}, V_{t} \setminus S_t)$ for some $s \in A_t$
(we cannot have any edges in $E(V_{s} \setminus S_{s}, V_{t} \setminus S_t)$ since one of these nodes would be an~ancestor of the other one and this would contradict property (\ref{item:separation:ancestor}) of Definition~\ref{def:separation-decomposition}).
Then both $u,v$ are contained in $\chi_\hh(x_t)$.
The case $v \in V_t \setminus S_t$ is symmetric, so the remaining case is when $u,v \in Q$.
Again, if $uv \in E(G[V_t])$ for some $t$, then $u,v$ are contained in $\chi_\hh(x_t)$.
Otherwise, $uv \in E(V_{t_1} \cap S_{t_1}, V_{t_2} \cap S_{t_2})$ for $t_1t_2 \in E(\gsep)$.
Then there {exists} $x$ for which $t_1, t_2 \in \chi_\mathbf{tw}(x)$ and thus $u,v \in \chi_\hh(x)$.
\end{proof}

\begin{lemma}\label{lem:decomp-tw}
Suppose there exists an algorithm $\mathcal{A}$ for \textsc{Restricted $(\hh,h)$-separation finding} running in time $f(n,t)$.
Then there is an algorithm that,
given graph $G$ with $\hh$-treewidth $k-1$,
runs in time $\br{f(n,k) + 2^{\Oh(k)}} \cdot n^{\Oh(1)}$, and returns a tree $\hh$-decomposition of $G$ of width $\Oh((h(k))^2\cdot k)$.
\end{lemma}
\begin{proof}
\mic{It suffices to process each connected component of $G$ independently, so we can assume that $G$ is connected.}
We use Lemma~\ref{lem:crown-decomposition} to find a restricted $(\hh,k,h(k)+1)$-separation decomposition ${T}$ in time $f(n,k) \cdot n^{\Oh(1)}$.
The graph \gsep is guaranteed by Lemma~\ref{lem:quotient-width} to have treewidth at most $k$.
We find a tree decomposition of \gsep of width $\Oh(k)$ in time $2^{\Oh(k)} \cdot n^{\Oh(1)}$~\cite{BodlaenderDDFMP16, CyganFKLMPPS15}. 
We turn it into a tree $\hh$-decomposition of $G$ of width $h(k)\cdot(k+h(k)+1)\cdot \Oh(k) = \Oh((h(k))^2\cdot k)$ in polynomial time with Lemma~\ref{lem:treewidth-to-h-width}.
\end{proof}

Similarly to the construction of an \hhdepthdecomp{},
we can replace the approximation algorithm for computing the tree decomposition of \gsep with one that works in polynomial space.
For this purpose, we could use the known $\Oh(\sqrt{\log \tw})$-approximation which runs in polynomial time (and space)~\cite[Thm 6.4]{FeigeHL08}.
This construction is not particularly interesting though, because all the considered algorithms exploiting \hhtwdecomp{s} require exponential space.

\subsection{Finding separations}
\label{sec:finding-separators}

We have established a framework that allows us to {construct \hhdepthdecomp{}s and \hhtwdecomp{}s for any hereditary class \hh{}, as long as we can supply an efficient parameterized algorithm for \hhsepfind.}
In this section we provide such algorithms for {various} graph classes.
Their running times and the approximation guarantee $h$ vary over different classes and they govern the efficiency of the decomposition finding procedures.

{A} crucial tool employed in several arguments is the theory of important separators (Definition~\ref{def:imp:sep}). 
We begin with a few observations about them.
For  two disjoint sets $X, Y \subseteq V(G)$, let $\mathcal{S}(X,Y)$ be the set of all important $(X,Y)$-separators and $\mathcal{S}_k(X,Y)$ be the subset of $\mathcal{S}(X,Y)$ consisting of separators of size at most $k$. \bmp{We say that an algorithm \emph{enumerates the set~$\mathcal{S}_k(X,Y)$ in polynomial space} if it runs in polynomial space and outputs a member of~$\mathcal{S}_k(X,Y)$ at several steps during its computation, so that each member of~$\mathcal{S}_k$ is outputted exactly once.}

\begin{thm}[{\cite[{Thm.~8.51}]{CyganFKLMPPS15}}]\label{lem:important-enumerate}
For any disjoint $X, Y \subseteq V(G)$ the set
$\mathcal{S}_k(X,Y)$ can be enumerated in time $\Oh^*(|\mathcal{S}_k(X,Y)|)$ and polynomial space.
\end{thm}

\bmp{
Note that a polynomial-space algorithm cannot store all relevant important separators to output the entire set~$\mathcal{S}_k(X,Y)$ at the end, since the cardinality of~$\mathcal{S}_k(X,Y)$ can be exponential in~$k$. While the original statement of \cref{lem:important-enumerate} does not mention the polynomial bound on the space usage of the algorithm, it is not difficult to see that the algorithm indeed uses polynomial space. At its core, the enumeration algorithm for important separators is a bounded-depth search tree algorithm. Each step of the algorithm uses a maximum-flow computation to find a minimum~$(X,Y)$-separator which is as far from~$X$ as possible, and then branches in two directions by selecting a suitable vertex~$v \notin X \cup Y$ and either adding~$v$ to the separator or adding it to~$X$. In the base case of the recursion, the algorithm outputs an important separator.}

\begin{lemma}[{\cite[Lemma 8.52]{CyganFKLMPPS15}}]\label{lem:important-sum}
For any disjoint $X, Y \subseteq V(G)$ it holds that $\sum_{S \in \mathcal{S}(X,Y)} 4^{-|S|} \le 1$.
\end{lemma}

In particular this implies that $|\mathcal{S}_k(X,Y)| \le 4^k$.
The next observation makes a connection between important separators and $(\hh,k)$-separations, which will allow us to perform branching according to the choice of an important separator.

\begin{lemma}\label{lem:important-subset}
Let $Z \subseteq V(G)$ be a connected set of vertices. Suppose there is an \hsepk $(C,S)$ that covers $Z$ and a {vertex set $F$} such that $S$ is an $(F,Z)$-separator.
Then there exists an \hsepk $(C^*,S^*)$ covering $Z$ such that $S^*$ contains an important $(F,Z)$-separator {$S'$.
Furthermore,
$(C^*,S^* \setminus S')$ is an $(\hh,\, k-|S'|)$-separation covering $Z$ in $G-S'$ and if $(\widehat{C},\widehat{S})$ is an $(\hh,\, k')$-separation covering $Z$ in $G - S'$,
then $(\widehat{C},\widehat{S} \cup S')$ is an $(\hh,\, k'+|S'|)$-separation covering $Z$ in $G$.}
\end{lemma}
\begin{proof}
We first show the existence of $(C^*,S^*)$.
Let $(C,S)$ be an \hsepk covering $Z$ such that $S$ is an $(F,Z)$-separator. This exists by the premise of the lemma. If $S$ contains an important $(F,Z)$-separator, then with $(C^*,S^*) = (C,S)$ {the first claim holds.}
 
So suppose that this is not the case.
{Clearly $N_G(R_S(F)) \subseteq S$ is an $(F,Z)$-separator, so there exists an inclusion-minimal~$(F,Z)$-separator~$S_F \subseteq N_G(R_S(F))$.} As~$S_F \subseteq S$ is not important by assumption, {by Lemma~\ref{lem:sep-to-importantsep}} there exists an important $(F,Z)$-separator $S_F'$ {satisfying} $R_{S_F}(F) \subsetneq R_{S_F'}(F)$ {and $|S_F'| \le |S_F|$.}
We show that $(C^* = R_{S^*}(Z), S^* = (S \setminus S_F) \cup S_F')$ is an \hsepk. Clearly $|S^*| \leq |S| \leq k$ and $(C^*,S^*)$ covers $Z$. 
Since $C^*$ is the set of reachable vertices from $Z$ in $G - S^*$, it follows that $G[C^*]$ is a connected component of $G-S^*$. We show that~$C^* \subseteq C$. Suppose not. Since~$Z$ is connected, so is~$G[C^*]$. Consider a path~$P$ in~$G[C^*]$ from~$Z$ \bmp{to a vertex in~$C^* \setminus C$}. Since~$G[C]$ also contains~$Z$, the first vertex on~$P$ that is not in~$C$, is a vertex~$v\in N_G(C) \cap (C^* \setminus C)$. Then~$v \notin (S \setminus S_F)$, since~$S \setminus S_F \subseteq S^*$. As~$N_G(C) \subseteq S$, we therefore have~$v \in S_F$. Since~$S_F \ni v$ is an  \emph{inclusion-minimal}~$(F,Z)$-separator in~$G$, there exists~$u \in N_G(v) \cap R_{S_F}(F)$. But this leads to a contradiction: as~$R_{S'_F}(F) \supsetneq R_{S_F}(F)$ there is a path from~$F$ to~$u$ in~$G - S'_F$, while~$v \in C^*$ yields a path from~$v$ to~$Z$ in~$G - S'_F$ as~$S'_F \subseteq S^*$, so the combination of these two paths with the edge~$uv$ yields a path from~$F$ to~$Z$ in~$G - S'_F$; a contradiction to the assumption that~$S'_F$ is an~$(F,Z)$-separator. Hence~$C^* \subseteq C$, and since $\hh$ is hereditary we have $G[C^*] \in \hh$.

Next, {let~$S' \subseteq S^*$ and $(C^*,S^*)$ be an~$(\hh,k)$-separation covering~$Z$ in $G$.}
{Since {$N_G(C^*) \subseteq S^*$}, we have~$N_{G-S'}(C^*) \subseteq S^* \setminus S'$, and
therefore $(C^*, S^* \setminus S')$ is an $(\hh,\, k-|S'|)$-separation covering $Z$ in $G-S'$.}

{Analogously, we argue that if $(\widehat{C},\widehat{S})$ is an $(\hh,\, k')$-separation covering $Z$ in $G - S'$, then $(\widehat{C},\widehat{S} \cup S')$ is an $(\hh,\, k'+|S'|)$-separation covering $Z$ in $G$. We have $G[\widehat{C}] \in \hh$ and $|\widehat{S} \cup S'| = k' + |S'|$. Since 
{$N_{G}(\widehat{C}) \subseteq N_{G - S'}(\widehat{C}) \cup S' \subseteq \widehat{S} \cup S'$, the pair~$(\widehat{C}, \widehat{S} \cup S')$ is indeed an~$(\hh, k' + |S'|)$ separation in~$G$, which clearly covers~$Z$.}}
\end{proof}

\subsubsection{Classes excluding induced subgraphs}

Within this section, let $\hh$ be a class given by a finite family $\mathcal{F}$ of forbidden induced subgraphs, so that $G \in \hh$ if and only if $G$ does not contain any graph from $\mathcal{F}$ as an induced subgraph.
Note that such a formulation allows us to express also classes excluding non-induced subgraphs.
Unlike further examples, in this setting we do not need to settle for approximation and we present an exact algorithm for separation finding.
\mic{The algorithm below solves
the \hhsepfind{} problem \bmp{for} $h(x)=x$ with a~single-exponential \bmp{dependence} on the parameter.
Note that
it does not rely on the assumption of bounded $\hh$-treewidth.
}

\begin{lemma}\label{lem:separation-subgraph}
\mic{Let $\hh$ be a class given by a finite family $\mathcal{F}$ of forbidden \jjh{induced subgraphs}.
There is an algorithm that, given an $n$-vertex graph $G$, an integer $k$, and a connected non-empty \bmp{set} $Z \subseteq V(G)$,
runs in time $2^{\Oh(k)} \cdot n^{\Oh(1)}$ and polynomial space, and returns an \hsepk{} covering $Z$ or concludes that~$Z$ is~$(\hh,k)$-inseparable.}
\end{lemma}
\begin{proof}
Let $(G,Z,k)$ be the input. \bmp{We present a recursive algorithm that works under the assumption that 
there exists an \hsepk{} $(C,S)$ covering
$Z$ and outputs an~$(\hh,k)$-separation covering~$Z$. 
If the algorithm fails to output a separation during its execution, we shall conclude that $Z$ is $(\hh,k)$-inseparable.} Since it suffices to consider the connected component containing~$Z$, we may assume that~$G$ is connected.

Let $q$ denote the maximal number of vertices among graphs in $\mathcal{F}$.
We can check in time $\Oh(n^q)$ whether $G$ contains an induced subgraph $F$ isomorphic to one from $\mathcal{F}$.
If not, we can simply return $(V(G), \emptyset)$.

It cannot be that $V(F) \subseteq C$, so \bmp{there exists a vertex} $u \in V(F) \setminus Z$ \bmp{such that} either $u \in S$ or $S$ is a $(u,Z)$-separator.
In the first case $(C, S \setminus u)$ is an answer for the instance $(G - u, Z, k - 1)$,
\mic{which we solve recursively.}  
In the second case, by Lemma~\ref{lem:important-subset}, there exists an important $(u,Z)$-separator $S'$, such that $|S'| \le k$ and there exists an \hsepk $(C^*, S^*)$ covering $Z$ with $S' \subseteq S^*$.
\mic{Note that since we have assumed that $G$ is connected, the separator $S'$ is always non-empty.}
Then $(C^*, S^* \setminus S')$ is an answer for the instance $(G-S', Z, k - |S'|)$.
We enumerate all separators from $\mathcal{S}_k(u,Z)$ and
branch on all these possibilities\bmp{, over all (at most~$q$) ways to choose~$u$.} If any recursive call returns a separation, then we can complete it to an answer for $(G,Z,k)$.

We shall bound the maximal number of leaves $T(k)$ in the {recursion} tree with respect to parameter~$k$.
For each $u \in V(F) \setminus Z$
the first branching rule leads to at most $T(k-1)$ leaves and
the second one gives $\sum_{S' \in \mathcal{S}_k(u,Z)}T(k - |S'|)$.
Let us show inductively that $T(k) \le (5q)^k$, first upper bounding the summand coming from the important separators \bmp{for one choice of~$u$}.

\begin{align*}
\sum_{S' \in \mathcal{S}_k(u,Z)}T(k - |S'|) \le& \sum_{S' \in \mathcal{S}_k(u,Z)}(5q)^{k - |S'|} \le
 \sum_{S' \in \mathcal{S}_k(u,Z)}\left(\frac{5q}{4}\right)^{k - |S'|}4^k\cdot4^{- |S'|} \\
 \le\, &\left(\frac{5q}{4}\right)^{k - 1}4^k\sum_{S' \in \mathcal{S}_k(u,Z)}4^{- |S'|} \le
 4\cdot(5q)^{k - 1} \quad\text{(Lemma~\ref{lem:important-sum})}
\end{align*}
\begin{align*}
    T(k) \le \sum_{u \in V(F) \setminus Z} \left(T(k-1) + \sum_{S' \in \mathcal{S}_k(u,Z)}T(k - |S'|)\right) \le q\cdot\left((5q)^{k-1} + 4\cdot(5q)^{k - 1}\right) = (5q)^{k}
\end{align*}
Since the depth of the recursion tree is $k$, we obtain a bound of $k \cdot (5q)^k$ on the number of recursive calls.
By Theorem~\ref{lem:important-enumerate}, the time spent in each call is proportional to the number of children in the recursion tree times a~polynomial factor. As the sum of the numbers of children, taken over all the nodes in the recursion tree, is bounded by its total size~$k \cdot T(k)$, this proves the lemma.
\end{proof}

\subsubsection{Bipartite graphs}
{Let $\mathsf{bip}$ denote the class of bipartite graphs.
For $\hh = \mathsf{bip}$ we present a~polynomial-time algorithm, which gives a 2-approximation to the task of separation finding.
Therefore, we obtain a~better running time but worse approximation guarantee than in the previous section.
\mic{Again, we do not rely on the bounded $\mathsf{bip}$-treewidth.} \bmp{The 2-approximation for separation finding for bipartite graphs is related to the half-integrality of 0/1/\textsc{all} constraint satisfaction problems as studied by Iwata, Yamaguchi, and Yoshida~\cite{IwataYY18}. To avoid the notational overhead of their framework, we give an elementary algorithm below.}

The presented algorithm works on a~parity graph $G'$:
each vertex $v$ from the original graph $G$ is replaced by two parity-copies $v', v''$ and each edge in $G'$ switches the parity of a~vertex.
A~path in $G$ is odd if and only if the endpoints of its counterpart in $G'$ have different parity.
Hence, we can express the task of hitting odd cycles in terms of finding a~separator in the parity graph.}

For $A \subseteq V(G)$, we say we identify $A$ \bmp{into} a vertex $v \notin V(G)$, if we add vertex $v$ to $G$ with $N_G(v) = \bigcup_{u \in A} N_G(u) \setminus A$, followed by deleting $A$.
\begin{lemma}\label{lem:separation-bipartite}
There is a polynomial-time algorithm for \textsc{($\mathsf{bip}$,\,$h(x) = 2x$)-separation finding}.
\end{lemma}
\begin{proof}
Let $G$ be a graph, $k$ be an integer, and $Z \subseteq V(G)$ be non-empty such that $G[Z]$ is connected.
If $k = 0$, then let $C$ be the connected component of $G$ that contains $Z$. The separation $(C,\emptyset)$ is a ($\mathsf{bip}$,$0$)-separation if $G[C]$ is bipartite, otherwise $Z$ is ($\mathsf{bip}$,$0$)-inseparable. Similarly, if $G[Z]$ is not bipartite then we simply conclude that $Z$ is ($\mathsf{bip}$,$2k$)-inseparable. In the remainder, assume that $G[Z]$ is bipartite and $k > 0$.
Compute a proper 2-coloring $c \colon Z \to [2]$ of $G[Z]$, which is unique up to relabeling since $G[Z]$ is connected. Let $G_Z$ be the graph obtained from~$G$ by {identifying each color class into a single vertex. That is, we identify $c^{-1}(1)$ into a vertex $v_1$ and $c^{-1}(2)$ into a vertex $v_2$.} 
If $Z$ consists of a single vertex, then without loss of generality assume only $v_1$ exists, the rest of the construction remains the same up to ignoring $v_2$.
Next we construct a graph $G_Z'$ that contains a false twin for every vertex in $G_Z$ in order to keep track of parity of paths. More formally, let $V(G_Z') = \{u',u'' \mid u \in G_Z\}$ and $E(G_Z') = \{x'y'' \mid {xy} \in E(G_Z)\}$. Intuitively, an odd {cycle} in $G_Z$ that contains some vertex $u$ corresponds to a path {in~$G'_Z$ from $u'$ to its false-twin} copy $u''$. In what follows we show that breaking such paths for a good choice of endpoints results in the desired $(\mathsf{bip}$,$2k$)-separation.

\begin{claim}
If the set $S \subseteq V(G_Z')$ is a $(\{v_1',v_2''\},\{v_1'',v_2'\})$-separator in $G_Z'$, then the connected component {of $G - \{u \mid u' \in S \vee u'' \in S\}$} that contains $Z$ is bipartite.
\end{claim}
\begin{innerproof}
Suppose $S \subseteq V(G_Z')$ is a $(\{v_1',v_2''\},\{v_1'',v_2'\})$-separator in $G_Z'$. Let $X = \{u \mid u' \in S \vee u'' \in S\}$. For the sake of contradiction, suppose that the component of $Z$ in $G - X$ is not bipartite. Let $Q$ be an odd cycle in said component, and let $Q_Z$ be the walk in $G_Z-X$ obtained by replacing {each} $u \in Z \cap Q$ by {the vertex $v_{c(u)}$ into which it was identified.} 
Clearly $Q_Z$ is a closed odd walk in $G_Z-X$, since there are no two consecutive vertices of $Q$ in $Z \cap Q$ that are in the same color class.
{As $Z$ is connected, we have $v_1v_2 \in E(G_Z)$, so $v_1$ and $v_2$ lie in a single connected component of $G_Z-X$.
We obtain that this component contains an odd cycle.} Let $O = (u_1,\ldots,u_{2k+1})$ be a {shortest} odd cycle in $G_Z-X$ in the component of $\{v_1,v_2\}$. Let $X'= \{u',u'' \mid u \in X\}$; note that $S \subseteq X'$. We do a case distinction on $|V(O) \cap \{v_1,v_2\}|$, in each case we show there exists a path in $G_Z'-S$ that contradicts that $S$ is a $(\{v_1',v_2''\},\{v_1'',v_2'\})$-separator. 

If $|V(O) \cap \{v_1,v_2\}| = 0$, then since $O$ and $\{v_1,v_2\}$ are in the same connected component of $G_Z - X$, there is a path $P$ in $G_Z'-X'$  from $v_1'$ to either $u_1'$ or $u_1''$ depending on the parity of $|P|$. Suppose $P$ ends in $u_1'$ (the argument for $P$ ending in $u_1''$ is symmetric), that is, $P = (v_1',\ldots,u_1')$. Then $P$ concatenated with $(u_2'',u_3',\ldots,u_{2k+1}',u_1'')$ ends at $u_1''$ in $G_Z'-X'$, since $O$ is an odd cycle. But then taking the twins of $P$ in the reverse direction, we arrive at $v_1''$. Since $G_Z'-X'$
is an induced subgraph of $G_Z'-S$, it follows that $S$ does not contain a single vertex of this path, and hence $S$ cannot be a $(\{v_1',v_2''\},\{v_1'',v_2'\})$-separator. 

If $|V(O) \cap \{v_1,v_2\}| = 1$, then {by re-numbering the cycle~$O$ if needed}, we may assume $V(O) \cap \{v_1,v_2\} = \{u_1\}$. {Assume $u_1 = v_1$; the case~$u_1 = v_2$ is symmetric.} Then $P = (v_1',u_2'',u_3',\ldots, \allowbreak u_{2k+1}',v_1'')$ is a $v_1'v_1''$-path in $G_Z'-S$. It follows that $S$ cannot be a $(\{v_1',v_2''\},\{v_1'',v_2'\})$-separator.

Finally, consider the case $|V(O) \cap \{v_1,v_2\}| = 2$. First note that $v_1$ and $v_2$ are consecutive in $O$. To see this, for any odd cycle that contains both $v_1$ and $v_2$, exactly one of the $v_1v_2$-path or the $v_2v_1$-path contains an even number of edges, which can be combined with the edge $v_1v_2$ to form an odd cycle. Since $O$ is {a shortest odd cycle}, $v_1$ and $v_2$ are consecutive in $O$. Without loss of generality assume $u_1 = v_1$ and $u_{2k+1} = v_2$. Then $(v_1'$,$u_2''$,\ldots,$v_2')$ is a path of even length in $G_Z'-S$. It follows that $S$ cannot be a $(\{v_1',v_2''\},\{v_1'',v_2'\})$-separator. 
\end{innerproof}

\begin{claim}
If $Z$ is ($\mathsf{bip}$,$k$)-separable, then $G_Z'$ has a $(\{v_1',v_2''\},\{v_1'',v_2'\})$-separator of size at most~$2k$.
\end{claim}
\begin{innerproof}
Suppose $Z$ is ($\mathsf{bip}$,$k$)-separable and let $(C,X)$ be a ($\mathsf{bip}$,$k$)-separation that covers $Z$. We show that $X' = \{u',u'' \mid u \in X\}$ is a $(\{v_1',v_2''\},\{v_1'',v_2'\})$-separator in $G_Z'$. Clearly $|X'| \leq 2k$. Note that $X' \subseteq V(G_Z') \setminus \{v_1',v_1'',v_2',v_2''\}$ since $Z \subseteq C$.
Suppose $X'$ is not a  $(\{v_1',v_2''\},\{v_1'',v_2'\})$-separator in $G_Z'$. We consider the cases that there is a $v_1'v_1''$-path and a $v_1'v_2'$-path; the other cases are symmetric. 

Let $P'$ be a $v_1'v_1''$-path in $G_Z'-X'$, then $P'$ has an odd number of edges. {The walk in $G_Z$ obtained from $P'$ by dropping all parity information is also odd and has empty intersection with $X$ because the set $X'$ is symmetric. Hence,~$G_Z - X$ contains an odd cycle~$P$.}
\begin{itemize}
    \item If~$P$ contains the edge~$v_1 v_2$, then the remainder of~$P$ is a path of even length connecting~$v_1$ to~$v_2$. As~$v_1$ and~$v_2$ were obtained by identifying the color classes of~$Z$, it follows that there is an even-length walk~$P^*$ in~$G - X$ from a vertex~$z_1 \in Z \cap c^{-1}(1)$ to a vertex~$z_2\in Z \cap c^{-1}(2)$. As~$G[Z]$ is bipartite and connected, with~$z_1$ and~$z_2$ of opposite colors, we can extend this even walk~$P^*$ with an odd-length path from~$z_1$ to~$z_2$ in~$G[Z] - X$, to obtain an odd closed walk in the connected component of~$G - X$ containing~$Z$; a contradiction to the assumption that $(C,X)$ is a ($\mathsf{bip}$,\,$k$)-separation that covers $Z$. 
    \item If~$P$ does not contain the edge~$v_1 v_2$, then we can lift~$P$ to a closed odd walk in the component of~$G - X$ that contains~$Z$, as follows. For every occurrence of some~$v_j \in \{v_1, v_2\}$ on~$P$, let~$p$ and~$q$ be the predecessor and successor of~$v_j$ on~$P$. By the assumption of this case,~$p,q \notin \{v_1, v_2\}$. By construction of~$G_Z$, the fact that~$p,q \in N_{G_Z}(v_j)$ implies that there exist~$z_p \in c^{-1}(j) \cap Z \cap N_G(p)$ and~$z_q \in c^{-1}(j) \cap Z \cap N_G(q)$. Since~$G[Z]$ is connected, there is a path~$P_{p,q}$ from~$z_p$ to~$z_q$ in~$G[Z]$, and as~$G[Z]$ is bipartite and~$z_p,z_q$ belong to the same color class, this path has an even number of edges. Hence we can replace the occurrence of~$v_j$ on walk~$P$ by the path~$P_{p,q}$, which increases the number of edges on the walk by an even number and therefore preserves its parity. Replacing each occurrence of~$v_j \in \{v_1, v_2\}$ on~$P$ by a path~$P_{p,q}$ in this way, we obtain a closed  walk~$P^*$ of odd length in the connected component~$G - X$ that contains~$Z$; a contradiction.
\end{itemize}

Finally let $P'$ be a $v_1'v_2'$-path in $G_Z'-X'$, then $P'$ has an even number of edges. Let $P$ be the walk in $G_Z - X$ obtained from $P'$ by dropping all parity information, $P$ is an even walk in $G_Z - X$ in the component of $\{v_1,v_2\}$. {Similarly as in the first bullet above, turn $P$ into a closed odd walk by making the last step from $v_2$ to $v_1$ with an odd-length path through~$G[Z]$.} Again we reach a contradiction for the fact that $(C,X)$ was a ($\mathsf{bip}$,$k$)-separation that covers $Z$.
\end{innerproof}

Now we can find the desired separation covering $Z$ by first constructing $G_Z'$ and then computing a $(\{v_1',v_2''\},\{v_1'',v_2'\})$-separator{, for example using the Ford-Fulkerson algorithm}. If the latter is larger than $2k$, return that $Z$ is ($\mathsf{bip}$,$2k$)-inseparable. Otherwise, return the separation given by the first claim together with the connected component of~$Z$. Clearly these operations can be performed in polynomial time.
\end{proof}

\subsubsection{Separation via packing-covering duality}
\label{subsec:packing-covering}


In this section we present the most general way of finding separations, which is based on a packing-covering duality.
We show that the existence of an efficient parameterized (or approximation) algorithm for \textsc{$\hh$-deletion} entails the existence of an efficient algorithm for \textsc{$(\hh, \Oh(k^2))$-separation finding}.

The relation between the maximum size of a packing of obstructions to $\hh$ (e.g., odd cycles for the case $\hh = \mathsf{bip}$) and the minimum size of an $\hh$-deletion set is  
called the Erd\H{o}s-P\'{o}sa property {for obstructions to~$\hh$}, after the famous result on the packing-covering duality for cycles~\cite{ErdosP65}.
It is known that graphs with bounded treewidth enjoy such a packing-covering duality
for connected obstructions for various graph classes.
\mic{We show that this relation also holds for graphs with bounded \hhtwfull{} and obstructions to $\hh$.}
This is the only place where the assumption of bounded \hhtwfull{} plays a role for constructing algorithms for \hhsepfind{}.

\begin{lemma}\label{lem:branching:erdos-posa}
Let \hh{} be a hereditary union-closed graph class and $G$ be a graph satisfying
$\hhtw(G) \le k$.
Then $G$ either contains $k+1$ {disjoint vertex sets $V_1, \dots, V_{k+1} \subseteq V(G)$ such that $G[V_i] \not\in \hh$ and $G[V_i]$ is connected \bmp{for all~$i \in [k+1]$}}, or an $\hh$-deletion set of size at most $k(k+1)$. 
\end{lemma}
\begin{proof}
Consider a tree $\hh$-decomposition $(\mathcal{T}, \chi, L)$ of $G$ {of} width at most $k$ and root it at an arbitrary node~$r$. 
For $t \in V(\mathcal{T})$ let $S_t = \chi(t) \setminus L$ and let $V_t$ denote the set of vertices occurring in {bags} in the subtree of $\mathcal{T}$ rooted at $t$ ({including} $\chi(t)$).

\bmp{If~$G \in \hh{}$, then the empty set is an $\hh$-deletion set of size~$0$ and the lemma holds. If not, consider a node} $t$ for which $G[V_t] \not\in \hh$ and $t$ is furthest from the root among all such nodes. 
{By Observation~\ref{obs:basecomponent:neighborhoods} we have \mic{that $N(\chi(t) \cap L) \subseteq S_t$ and by definition $G[\chi(t) \cap L] \in \hh$.}
Also $G[V_{t_i}] \in \hh$ for any child $t_i$ of $t$, by the choice of $t$, and thus $G[V_{t_i} \setminus S_t] \in \hh$.
Since $\hh$ is closed under disjoint unions, we get that the union of all these subgraphs, i.e., $G[V_t \setminus S_t]$ also belongs to $\hh$.}
Hence, removing $S_t$ separates $G$ into $G - V_t$ and a~subgraph from $\hh$.

We define $S_1 = S_t,\,G_1 = G[V_t]$, and iterate this procedure on $G' = G - V_t$.
If after $\ell \le k$ steps we have reached a graph that belongs to $\hh$,
then $S = \bigcup_{i=1}^\ell S_i$ is an $\hh$-deletion set and $|S| \le \ell\cdot\max_{i=1}^\ell |S_i| \le k(k+1)$.
\mic{Otherwise the induced subgraphs $G_1, G_2,\dots, G_{k+1}$ are vertex-disjoint and $G_i \not\in \hh$ for \jjh{each} $i \in [k+1]$.
Since $\hh$ is closed under disjoint unions, for each $i \in [k+1]$ there is a 
subset $V'_i \subseteq V(G_i)$ such that $G[V'_i]$ is
connected and does not belong to $\hh$.}
\end{proof}

So far we have given an existential proof that under the assumption of bounded \hhtwfull{} the \bmp{obstructions to} $\hh$ enjoy a packing-covering duality.
Next, we show that when a graph admits a packing \bmp{of} $k+1$ obstructions to $\hh$ and set $Z$ is $(\hh,k)$-separable, then we can detect one obstruction that can cheaply be separated from $Z$.
\mic{Let $\hhdn{}(G)$ \bmp{(for \emph{$\hh$-deletion number})} denote the minimum size of an $\hh$-deletion set in $G$.} 

\begin{lemma}\label{lem:branching:separated-obstruction}
Let \hh{} be a hereditary union-closed graph class \bmp{such that there exists an algorithm~$\mathcal{A}$ that tests} membership of an $n$-vertex graph in $\hh$ in time $f(n)$. 
Then there is an algorithm that, \mic{given integers $k \le t$, a graph $G$ satisfying
$\hhtw(G) \le t$}, \bmp{and} an~$(\hh,k)$-separable set $Z \subseteq V(G)$, runs in time $4^k \cdot (f(n) + n^{\Oh(1)})$ and polynomial space; it either correctly concludes that $\hhdn{}(G) \le t(t+1)$ or outputs a set $F \subseteq V(G)$ such that (1) $G[F] \not\in \hh$, (2) $G[F]$ is connected, (3) $N_G[F] \cap Z = \emptyset$, and (4) $|N_G(F)| \le k$.
If $\mathcal{A}$ runs in polynomial space, then the latter algorithm does as well. 
\end{lemma}
\begin{proof}
By \cref{lem:branching:erdos-posa} either $\hhdn{}(G) \le t(t+1)$ or there exist $t+1$ disjoint vertex sets $V_1, \dots, V_{t+1}$ such that $G[V_i] \not \in \hh$ and $G[V_i]$ is connected.
We present an algorithm that outputs the requested set $F$ under the assumption that the second scenario holds.
If the algorithm fails to output $F$, we shall conclude that $\hhdn{}(G) \le t(t+1)$.

We first present the algorithm and then argue for its correctness.
For each vertex $v \in V(G) \setminus N_G[Z]$
we enumerate all important $(v,Z)$-separators of size at most $k$.
There are at most $4^k$ such separators (\cref{lem:important-sum}) for a fixed $v$ and they can be enumerated in time $\Oh^*(4^k)$ (\cref{lem:important-enumerate}).
For a~separator $S'$ we consider the set $F_v^{S'} = R_{S'}(\{v\})$ \bmp{(recall Definition~\ref{def:imp:sep})}. We have $N_G(F_v^{S'}) \subseteq S'$ (since the separator $S'$ is important, it is also inclusion-wise minimal, so the equality holds as well) and therefore $F_v^{S'}$ satisfies \bmp{conditions} (2)--(4).
We check whether $G[F_v^{S'}] \in \hh$, in time $f(n)$: if not, we simply output $F = F_v^{S'}$.

Now we analyze the algorithm.
Let $(C,S)$ be an \hsepk{} covering $Z$, which exists by assumption.
We have also assumed the existence of $t+1$ disjoint vertex sets $V_1, \dots, V_{t+1}$ such that $G[V_i] \not \in \hh$ and $G[V_i]$ is connected.
Since $k \le t$,
by a counting argument there exists $j \in [t+1]$ such that $S \cap V_j = \emptyset$.
It cannot be that $V_j \subseteq C$ so,
by connectivity of $G[V_j]$, we infer that $C \cap V_j = \emptyset$.
Let $v$ be an arbitrary vertex from $V_j$.
We have that $S$ is a $(v,Z)$-separator of size at most $k$.
By \cref{lem:sep-to-importantsep} we know that there exists an~important $(v,Z)$-separator $S'$,
such that $R_S(v) \subseteq R_{S'}(v)$ and $|S'| \le |S| \le k$.
Since $\hh$ is hereditary and $V_j \subseteq R_S(v)$, the subgraph induced by $R_{S'}(v)$ does not belong to $\hh$.
This implies that such a set will be detected by the algorithm.
\end{proof}

In the next lemma, we exploit the obstruction $F$ to design a branching rule which shaves off some part of the set $S$ in the \hsepk{} $(C,S)$ covering the set $Z$.
We assume the existence of an algorithm $\mathcal{A}$ for $\hh$-\textsc{deletion} parameterized by the solution size $s$.
We allow $\mathcal{A}$ to be \bmp{a} $\beta$-approximation algorithm for some constant $\beta \ge 1$.
Such an algorithm either outputs an $\hh$-deletion set of size at most $\beta \cdot s$ or correctly concludes that $\hhdn{}(G) > s$.
We describe the running time of $\mathcal{A}$ by a function $f(n,s)$, where $n$ is the input size.
This function may be purely polynomial and independent of $s$; this will be the case when $\hh$-\textsc{deletion} admits a constant-factor polynomial-time approximation.
On the other hand, if $\hh$-\textsc{deletion} admits an exact FPT algorithm, the function $f(n,s)$ may be of the form, e.g., $2^{\Oh(s)} \cdot n^{\Oh(1)}$ and the approximation factor $\beta$ equals 1.

\begin{lemma}\label{lem:branching:general}
Let \hh{} be a hereditary union-closed graph class.
Suppose that $\hh$-\textsc{deletion} admits a $\beta$-approximation algorithm $\mathcal{A}$ $(\beta = \Oh(1)$, $\beta \ge 1)$ parameterized by the solution size $s$ running in time $f(n,s)$ on an $n$-vertex graph, where $f$ is a computable function non-decreasing on both coordinates.
Then there is an $f(n, t(t+1)) \cdot 2^{\Oh(t)} \cdot n^{\Oh(1)}$-time algorithm solving $(\hh, h)$-\textsc{separation finding} for $h(x) = \beta\cdot 2x(2x+1)$.
If $\mathcal{A}$ runs in polynomial space, then the latter algorithm does as well.
\end{lemma}
\begin{proof}
Let $(G,Z,k,t)$ be the given instance.
By the problem definition $G[Z]$ is connected, so whenever $G$ is not connected it suffices to focus on the connected component containing $Z$.
We present a recursive algorithm that works under the assumption that 
there exists an \hsepk{} $(C,S)$ covering
$Z$ \mic{and outputs an $(\hh,\beta (t+k)(t+k+1))$-separation weakly covering $Z$ (recall that $k \le t$ by the problem definition).}
If the algorithm fails to output an $(\hh,\beta (t+k)(t+k+1))$-separation weakly covering $Z$, we shall conclude that $Z$ is $(\hh,k)$-inseparable.
In the corner case $k=0$ it suffices to check whether the connected component containing $Z$ induces a graph from $\hh$.

\bmp{We can use $\mathcal{A}$ to test whether an $n$-vertex graph belongs to $\hh$ in time $f(n,0)$.} 
We execute the algorithm from \cref{lem:branching:separated-obstruction} in
time $4^k \cdot (f(n,0) + n^{\Oh(1)})$ (and polynomial space if $\mathcal{A}$ runs in polynomial space).
If it reports that $\hhdn{} (G) \le t(t+1)$ we
invoke the algorithm $\mathcal{A}$ to find an $\hh$-deletion set $S'$ of size at most $\beta t(t+1)$, \mic{in time $f(n, t(t+1))$.}
We then output the $(\hh,\beta t(t+1))$-separation $(V(G) \setminus S', S')$ which weakly covers $Z$.
Otherwise we obtain a set $F \subseteq V(G)$
such that (1) $G[F] \not\in \hh$, (2) $G[F]$ is connected, (3) $N_G[F] \cap Z = \emptyset$, and (4) $|N_G(F)| \le k$.

In such a case we perform branching: 
we create 
multiple instances of the form $(G-S', Z, k',t)$ where $S'$ is some $(F,Z)$-separator of size at most $k$ and $k' < k$. 
Note that $Z$ belongs to the vertex set of the graph on which we recurse \mic{and $\hhtw(G-S) \le \hhtw(G) \le t$. 
We keep the value of parameter $t$ intact so the invariant $k' \le t$ is preserved.}
If an \hsep{\beta (t+k')(t+k'+1)} $(\widehat C, \widehat S)$ weakly covering $Z$ in $G-S'$ is found, we output $(\widehat C, \widehat S \cup S')$ for the original instance $(G,Z,k,t)$.
Observe that in such a case $|\widehat S \cup S'| \le \beta (t+k')(t+k'+1) + k \le \beta (t+k-1)(t+k) + k \le \beta (t+k)(t+k+1)$ and $N_G(\widehat C) \subseteq \widehat S \cup S'$, so the solution is valid.
It remains to design a branching rule to ensure that in least one call $(G-S', Z, k',t)$ the set $Z$ is $(\hh,k')$-separable in $G-S'$.
Then it follows from induction that an \hsep{\beta (t+k')(t+k'+1)} $(\widehat C, \widehat S)$ weakly covering $Z$ in $G-S'$ will be found. 

The branching rule is divided into 2 parts.
\bmp{In part} (a) we consider all the important $(F,Z)$-separators of size at most $k$.
For each such a separator $S'$ we recurse on the instance $(G-S', Z, k-|S'|, t)$.
\bmp{In part} (b) we produce only one recursive call:
$(G-N_G(F), Z, k-1, t)$.
Note that $N_G(F)$ is an $(F,Z)$-separator of size at most $k$, so this instance obeys the description above.
Next, we advocate the correctness of this rule.
Recall that $(C,S)$ refers to an (unknown) \hsepk{} covering the given set $Z$, which we have assumed to exist.

If $F \cap S = \emptyset$, then by connectivity of $G[F]$ we obtain that either $F \subseteq C$ or $C \cap F = \emptyset$.
The first option is impossible because $G[F] \not\in \hh$,
therefore $S$ is an $(F,Z)$-separator.
From \cref{lem:important-subset}
we get that there exists an \hsepk{} $(C^*, S^*)$ covering $Z$ and an~important $(F,Z)$-separator $S'$, such that $S' \subseteq S^*$.
Furthermore (from the same lemma), we know that $(C^*, S^* \setminus S')$ is an $(\hh, k - |S'|)$-separation covering $Z$ in $G - S'$, hence the instance $(G-S',Z,k - |S'|,t)$ admits a solution.
Recall that we have reduced the problem to the case where $G$ is connected so $|S'| > 0$.
Such an~important $(F,Z)$-separator $S'$ is going to be considered during the branching rule, part (a), and in such a case the algorithm returns an \hsepk{} obeying the requirements for the instance $(G-S',Z,k-|S'|,t)$.

Now consider the case  $F \cap S \ne \emptyset$.
This means that the connected component of $F$ in $G-N_G(F)$ contains at least one vertex from $S$ and thus 
the connected component of $Z$
contains at most $k-1$ vertices from $S$.
More precisely, let $V_Z$ denote the vertex set inducing the connected component of $Z$ in $G-N_G(F)$.
Then $|S \cap V_Z| < k$ and $(C \cap V_Z, S \cap V_Z)$ is an \hsep{k-1} covering $Z$ in $G[V_Z]$, and thus in $G-N_G(F)$.
Hence, this case is covered in the part (b) in the branching rule.

Finally, we analyze the running time.
We bound the number of leaves in the recursion tree.
Let $T(k)$ denote the maximal number of leaves for a call with parameter $k$ -- we shall prove by induction that $T(k) \le 5^k$.
We have $T(0) = 1$.
In part (a) of the branching rule
we consider all the important $(F,Z)$-separators of size at most $k$ and for each such separator $S'$ we create an instance with parameter value $k-|S'|$.
Similarly as in the proof of \cref{lem:separation-subgraph} we take advantage of \cref{lem:important-sum} to upper bound the sum of these terms.
The second summand comes from part~(b), where we create a single instance  with parameter value $k-1$. 

%
\begin{align*}
\bmp{T(k) \le} & \sum_{S' \in \mathcal{S}_k(F,Z)}T(k - |S'|) + T(k-1) \le \sum_{S' \in \mathcal{S}_k(F,Z)}5^{k - |S'|} + 5^{k-1} \\ = &
 \sum_{S' \in \mathcal{S}_k(F,Z)}\left(\frac{5}{4}\right)^{k - |S'|}4^k\cdot4^{- |S'|} + 5^{k-1}
 \le\, \left(\frac{5}{4}\right)^{k - 1}4^k\sum_{S' \in \mathcal{S}_k(F,Z)}4^{- |S'|} + 5^{k-1} \\ \le\,\, &
 4\cdot5^{k - 1} + 5^{k-1} = 5^k \quad\text{(Lemma~\ref{lem:important-sum})}
\end{align*}

The number of nodes in the recursion tree of depth bounded by $k$ is clearly at most $k \cdot T(k)$.
In each node we may invoke the algorithm from \cref{lem:branching:separated-obstruction} running in time $4^k \cdot (f(n,0) + n^{\Oh(1)})$, the algorithm $\mathcal{A}$ running in time $f(n,t(t+1))$, and enumerate the important $(F,Z)$ separators of size at most $k$, which takes time $4^k\cdot n^{\Oh(1)}$ (\cref{lem:important-enumerate}).
\end{proof}

\subsubsection{Finding restricted separations}

As explained in Section \ref{sec:decomposition-restricted-separation}, for the sake of constructing a \hhtwdecomp{} we need algorithms with a stronger guarantee, formalized as a restricted separation.
We recall the problem statement.

\defparproblem{Restricted $(\hh,h)$-separation finding}{Integers $k \le t$, a graph $G$ of $\hh$-treewidth bounded by $t$, a connected {non-empty} subset $Z \subseteq V(G)$, a family $\mathcal{F}$ of connected, disjoint, $(\hh,k)$-inseparable subsets of $ V(G)$.}{$t$}
{Either return an \hsep{h(t)} $(C,S)$ that {weakly covers} $Z$ {such that} $C$ contains at most \mic{$t$} sets from $\mathcal{F}$, or conclude that $Z$ is $(\hh,k)$-inseparable.}

It turns out that as long as we can solve the basic separation problem exactly,
we obtain the stronger guarantee ``for free''. 

\begin{lemma}\label{lem:separation-subgraph-restricted}
For $\hh$ given by a finite family of forbidden induced subgraphs and function $h(x)=x$, \textsc{Restricted $(\hh,h)$-separation finding} is solvable in time $2^{\Oh(t)}n^{\Oh(1)}$ and polynomial space.
\end{lemma}
\begin{proof}
We have $k \le t$.
By \cref{lem:separation-subgraph},
there is an algorithm running in time $2^{\Oh(k)}n^{\Oh(1)}$ and polynomial space,
which given a \bmp{non-empty} connected set $Z$ in a graph $G$, outputs an $(\hh,k)$-separation $(C,S)$ covering $Z$ \bmp{or concludes that~$Z$ is $(\hh,k)$-inseparable}.
\bmp{If }the set $C$ is $(\hh,k)$-separable,
no subset of $C$ can be $(\hh,k)$-inseparable and the condition required in restricted separation is always satisfied. Hence the output for the non-restricted version of the problem is also valid in the restricted setting.
\end{proof}

For other graph classes we present a reduction from  \textsc{Restricted $(\hh,2h)$-separation finding}  to \hhsepfind{} that is burdened with a multiplicative factor $2^{\Oh(t\log t)}$. \bmp{The following lemma applies to hereditary classes~$\hh$, even those which are not union-closed.}

\begin{lemma}\label{lem:restricted-to-unrestricted}
Suppose \textsc{$(\hh,h)$-separation finding} admits an algorithm $\mathcal{A}$ running in time $f(n,t)$
and function $h$ satisfies $h(x) \ge x$.
Then \textsc{Restricted $(\hh, 2h)$-separation finding} admits an algorithm with running time $2^{\Oh(t\log t)} \cdot (f(n,t) + n^{\Oh(1)})$.
If $\mathcal{A}$ runs in polynomial space, then the latter algorithm does as well.
\end{lemma}
\begin{proof}
We shall prove a stronger claim: given an instance $(G, Z, \mathcal{F}, k, t)$, \mic{within the claimed running time we can either find an $(\hh,h(t) + k)$-separation weakly covering $Z$ and covering at most $k$ sets from $\mathcal{F}$, or correctly report that $Z$ is $(\hh,k)$-inseparable.}
This is indeed stronger than the lemma statement because $k \le t \le h(t)$.

We begin with the description of the algorithm.
Since it suffices to consider the connected component containing~$Z$ {and those sets from~$\mathcal{F}$ contained in that component}, we may assume that~$G$ is connected.
The algorithm distinguishes two cases.
First, if $|\mathcal{F}| \le k$ then we just execute the
algorithm $\mathcal{A}$ with input $(G, Z, k, t)$.
If $|\mathcal{F}| > k$
we choose {an} arbitrary subfamily $\mathcal{F}_0 \subseteq \mathcal{F}$ of size $k+1$.
We {branch} on the choice of $F \in \mathcal{F}_0$ and the choice of an important $(F,Z)$-separator of size at most $k$.
For each $S' \in \mathcal{S}_k(F,Z)$ we invoke a recursive call $(G',Z,\mathcal{F}',k',t)$ with
$G' = G - S'$, $k' = k - |S'|$, and $\mathcal{F}'$ being the subfamily of sets from $\mathcal{F}$ that are subsets of $V(G')$.
If an $(\hh, h(t) + k')$-separation $(\widehat C, \widehat S)$ weakly covering $Z$ in $G'$ is found, we return $(\widehat C, \widehat S \cup S')$.
By connectivity we know that the separator $S'$ is non-empty so $k' < k$.
Note that we preserve the invariant $k' \le t$.
If we do not have enough budget to cut away any {set} $F$, the recursive call reports failure.

We prove the correctness of the algorithm by induction on $k$.
Recall the assumption that $G$ is connected.
\mic{If $k=0$ and $\mathcal{F} = \emptyset$ the algorithm $\mathcal{A}$ gets executed and clearly returns a feasible solution because the restricted condition is trivially satisfied.
If $k=0$ and $\mathcal{F} \ne \emptyset$, then we know that $G \not\in \hh$ because otherwise any subset of $V(G)$ would be $(\hh,k)$-separable.
In this case $Z$ is $(\hh,k)$-inseparable because any $(\hh,0)$-separation $(C,S)$ covering $Z$ would need to satisfy $C = V(G)$.
The algorithm will report this fact as there are no important separators of size 0 to perform branching.}

If $k > 0$ and $|\mathcal{F}| \le k$, then any (unrestricted) $(\hh, h(t))$-separation found by the algorithm $\mathcal{A}$ satisfies the restriction trivially.
\bmp{Suppose that $k > 0$ and $|\mathcal{F}| > k$. It is easy to see that the algorithm gives a correct answer when it returns a separation, so it suffices to prove that the algorithm indeed outputs a separation when~$Z$ is $(\hh,k)$-separable}. Consider
some \hsepk $(C, S)$ covering $Z$.
Let $\mathcal{F}_0 \subseteq \mathcal{F}$ be the chosen subfamily of size $k+1$.
None of the sets from $\mathcal{F}_0$ can be contained in $C$ (as they are $(\hh,k)$-inseparable) and $|\mathcal{F}_0| > |S|$ so by connectivity for some $F \in \mathcal{F}_0$ the set $S$ is an $(F,Z)$-separator.
By Lemma~\ref{lem:important-subset} we know that there exists {an} \hsepk{} $(C^*, S^*)$ covering $Z$ such that $S^*$ contains an important $(F,Z)$-separator $S'$.
During the branching the algorithm considers the separator $S'$ and invokes a recursive call $(G',Z,\mathcal{F}',k',t)$ with
$G' = G - S'$, $k' = k - |S'|$, and $\mathcal{F}'$ being the subfamily of sets from $\mathcal{F}$ that are subsets of $V(G')$.
Then $Z$ is $(\hh,k')$-separable in $G'$ and by induction we know that the algorithm will output some $(\hh,h(t) + k')$-separation $(\widehat C, \widehat S)$
weakly covering $Z$ in $G'$
such that there are at most $k'$ sets from $\mathcal{F}'$ contained in $\widehat C$.
Furthermore, $N_G(\widehat C) \subseteq \widehat S \cup S'$ and no set from $\mathcal{F} \setminus \mathcal{F}'$ is contained in $\widehat C$ so
$(\widehat C, \widehat S \cup S')$ is an $(\hh,h(t) + k)$-separation weakly covering $Z$ in $G$ and covering at most $k' \le k$ sets from $\mathcal{F}$.

Finally, we analyze the running time.
The time spent in each recursive node is proportional to the number of direct recursive calls, \mic{modulo a polynomial factor}, due to Theorem~\ref{lem:important-enumerate}. We will bound the number of leaves in the recursion tree {similarly as in \cref{lem:separation-subgraph}}.
Let $T(k)$ denote the maximal number of leaves for a call with argument $k$.
We prove by induction that $T(k) \le (k+1)! \cdot 4^k$, which holds for $k=0$.
For each recursive call, as described above, we estimate
\begin{equation*}
T(k) \le \sum_{F \in \mathcal{F}_0} \sum_{S' \in \mathcal{S}_k(F,Z)} T(k - |S'|) \le k! \cdot 4^k \sum_{F \in \mathcal{F}_0} \sum_{S' \in \mathcal{S}_k(F,Z)} 4^{-|S'|} \stackrel{\text{Lemma }\ref{lem:important-sum}}{\le} (k+1)! \cdot 4^k.
\end{equation*}

The number of calls to $\mathcal{A}$ is thus at most $T(k)$.
Since the depth of the recursion tree is $k$, we obtain a bound of $k \cdot (k+1)! \cdot 4^k$ on the total number of recursive \mic{nodes}.
Together {with the assumption $k\le t$}, we obtain the claimed running time $2^{\Oh(t\log t)} \cdot (f(n,t) + n^{\Oh(1)})$.
\mic{The important separators can be enumerated using only polynomial space, so whenever  $\mathcal{A}$ runs in polynomial space then the entire algorithm does.}
\end{proof}

\subsubsection{Disconnected obstructions}
\mic{So far we have focused on union-closed graph classes, so we could take advantage of \cref{lem:branching:erdos-posa} (with the exception of classes given by forbidden induced subgraphs in \cref{lem:separation-subgraph}).
If the class $\hh$ is given by a family $\mathcal{F}$ of forbidden (topological) minors, then $\hh$ is union-closed if all graphs in $\mathcal{F}$ are connected.}
We shall now explain how to reduce the general problem to this case.
We define the class \dhh{} to be the closure of \hh{} under taking disjoint unions of graphs,
i.e., $G \in \dhh$ if each connected component of $G$ belongs to \hh.
Observe that in the definition of \hhdepthdecomp{} and \hhtwdecomp{} we only require that each base component belongs to $\hh$.
If a~base component $C$ is disconnected we can make several copies of its~bag, one for each connected component of~$C$.
Therefore we are only interested in checking whether some connected graphs belong to \hh{} and for the sake of obtaining decompositions there is no difference in studying \hh{} or \dhh{}.

\begin{observation}\label{obs:disjoint-union-decomposition}
An \hhdepthdecomp{} (resp. \hhtwdecomp{}) of $G$ can be transformed into an $\dhh$-elimination forest (resp. tree $\dhh$-decomposition) of $G$ of the same width, in polynomial time.
In particular, $\ed_\hh(G) = \ed_\dhh(G)$ and $\tw_\hh(G) = \tw_\dhh(G)$.
\end{observation}

Below we explain that given \hh{} \bmp{defined} via a~finite family of forbidden (topological) minors, we can construct an~analogous family for \dhh, in which all the graphs are connected.

\begin{lemma}\label{lem:disjoint-union-minors}
There exists an algorithm that, for \hh{} given by a~finite family of forbidden (topological) minors~$\mathcal{F}$, generates the family $\mathcal{D(F)}$ of 
forbidden (topological) minors for \dhh.
All the graphs in $\mathcal{D(F)}$ are connected.
\end{lemma}
\begin{proof}
Bulian and Dawar~\cite{BulianD17} have presented such an~algorithm for non-topological minors and we summarize it below.
It holds that $\hh = \dhh$ as long as all the graphs in $\mathcal{F}$ are connected.
If $H \in \mathcal{F}$ is not connected
and a~connected graph $G$ contains $H$ as a minor, then
$G$ also contains a minor being a~connected supergraph of $H$ \bmp{on vertex set~$V(H)$}.
\bmp{For each $H \in \mathcal{F}$ we generate all connected supergraphs of~$H$ on vertex set~$V(H)$, and add them to $\mathcal{D(F)}$. Afterwards we trim the family to its minor-minimal elements, to obtain the desired family of forbidden minors for~$\mathcal{D}(H)$.}

\bmp{With a small modification, the same argument also works for topological minors.}
If $H \in \mathcal{F}$ consists of \bmp{multiple} connected components $H_1, \dots, H_m$,
we consider \bmp{the}~family $\mathcal{T}_m$ of all trees $T$ such that $\{1,\dots, m\} \subseteq V(T)$ and all vertices apart from $\{1,\dots, m\}$ have degree at least 3.
Clearly, for $T \in \mathcal{T}_m$ it holds that $|V(T)| \le 2m$, so the family $\mathcal{T}_m$ is finite.
Now consider a~family $\mathcal{D}(H)$ given by choosing $T \in \mathcal{T}_m$, replacing the vertices $\{1,\dots, m\}$ with components $H_1, \dots, H_m$,
so that if $i \in \{1,\dots, m\}$ is an~endpoint of an~edge from $E(T)$, it gets replaced by some vertex from $H_i$ or a~vertex created by subdividing an~edge in $H_i$.
Then $\mathcal{D}(H)$ is also finite and
if a~connected graph $G$ contains $H$ as a~topological minor, then
$G$ also contains a topological minor from $\mathcal{D}(H)$.
The family $\mathcal{D(F)}$ is given by the family of topological-minor-minimal graphs in the~union of $\mathcal{D}(H)$ over $H \in \mathcal{F}$.
\end{proof}

We remark that when \hh{} is not closed under disjoint unions of graphs, we do not exploit the decompositions for solving the respective vertex deletion problems, as these become para-NP-hard (see Section~\ref{subsec:not:closed}).
However, we provided a~general way of constructing decompositions for the sake of completeness and because the resulting algorithms to approximate elimination distance and $\hh$-treewidth are of independent interest.

\subsection{Summary of the decomposition results}
\label{sec:summary}

Supplied with the algorithms for \hhsepfind{} and \hhsepfindres{}, we can plug them into the framework from Section~\ref{sec:decomp-abstract}.
\mic{Each result \bmp{below follows} from the~combination of a~respective algorithm for \hhsepfind{} with a procedure to build either an~\hhdepthdecomp{} or a~\hhtwdecomp{}.
Whenever the first algorithm runs in polynomial space, we formulate an~additional claim about constructing the~\hhdepthdecomp{} in polynomial space, but with a~slightly worse approximation guarantee.
We begin with a formalization of \cref{thm:solving:general}, which allows us to cover several graph classes in a unified way, and then move on to applications.}

\begin{thm}\label{thm:branching:final}
Let \hh{} be a hereditary union-closed graph class.
Suppose that $\hh$-\textsc{deletion} admits an exact or $\Oh(1)$-approximate algorithm $\mathcal{A}$ parameterized by the solution size $s$ running in time $f(n,s)$ on an $n$-vertex graph, where $f$ is a computable function non-decreasing on both coordinates.
There \bmp{exist algorithms} which given a graph $G$ perform the following:
\begin{enumerate}[a)]
       \item if $\hhtw(G) \le k$, the algorithm runs in time $f\left(n,k(k+1)\right) \cdot 2^{\Oh(k \log k)} \cdot n^{\Oh(1)}$ and returns a \hhtwdecomp{}{} of width $\Oh(k^5)$,
    \item if $\hhdepth(G) \le k$, the algorithm runs in time $\left(f(n,k(k+1)) \cdot 2^{\Oh(k)} + 2^{\Oh(k^2)}\right) \cdot n^{\Oh(1)}$ and returns an \hhdepthdecomp{} of depth $\Oh(k^3)$,
    \item (only applicable when $\mathcal{A}$ works in polynomial space) if $\hhdepth(G) \le k$, the algorithm runs in time $f\left(n,k(k+1)\right) \cdot 2^{\Oh(k)} \cdot n^{\Oh(1)}$ and polynomial space, and returns an \hhdepthdecomp{} of depth $\Oh(k^4\log^{3 / 2} k)$.
\end{enumerate}
\end{thm}
\begin{proof}
From \cref{lem:branching:general} we obtain an algorithm for \textsc{$(\hh, \Oh(t^2))$-separation finding} running in time
$g(n,t) = f(n, t(t+1)) \cdot 2^{\Oh(t)} \cdot n^{\Oh(1)}$. 
If $\mathcal{A}$ runs in polynomial space then the algorithm above does as well and we can plug it to \cref{lem:decomp-ed-polyspace} to prove the claim (c) with running time $g(n,k)\cdot n^{\Oh(1)}$.
Similarly, we obtain (b) via  \cref{lem:decomp-ed-exact} to get running time $(g(n,k) + 2^{\Oh(k^2)})\cdot n^{\Oh(1)}$.
We construct an algorithm for \textsc{Restricted $(\hh, \Oh(t^2))$-separation finding} via \cref{lem:restricted-to-unrestricted} at the expense of multiplying the running time $g(n,t)$ by $2^{\Oh(t \log t)}$.
Then the claim (a) follows from \cref{lem:decomp-tw}.
\end{proof}

Instead of formulating a lengthy theorem statement about numerous graph classes, we gather the detailed list of results in \cref{table:decomposition} and provide a combined proof for all of them below.
The cases of bipartite graphs and classes given by \mic{a finite family of} forbidden induced subgraphs are amenable to \cref{thm:branching:final}, but we provide  specialized algorithms achieving better running times and approximation guarantees.

\begin{table}[bt]
\caption{\small Summary of decomposition results. Each result is of the form: there is an~algorithm that, given graph $G$ satisfying $\hhdepth(G) \le k$ (resp. $\hhtw(G) \le k$), runs in time $f(k)\cdot n^{\Oh(1)}$, and returns an \hhdepthdecomp{} of depth $h(k)$ (resp. a~\hhtwdecomp{} of width $h(k)$).
One can read the {bounds} on $h$ (top) and $f$ (bottom) in
the corresponding cell.
For bipartite graphs and $h(k) = \Oh(k^3\log^{3 / 2} k)$ we obtain $f=\Oh(1)$ as the algorithm runs in polynomial time.
} \label{table:decomposition}
\centering

\[\begin{array}{|c|c|c|c|c|}\hline
\text{class } \hh & & \hhtw & \hhdepth & \hhdepth \text{ in poly-space} \\ \hline
\text{bipartite} &h & \Oh(k^3) & \Oh(k^2) & \Oh(k^3\log^{3 / 2} k) \\
 &f & 2^{\Oh(k\log k)} & 2^{\Oh(k^2)} &  \Oh(1) \\ \hline 
 \text{forbidden induced} &h & \Oh(k^3) & \Oh(k^2) & \Oh(k^3\log^{3 / 2} k) \\
 \text{subgraphs}&f & 2^{\Oh(k)} & 2^{\Oh(k^2)} & 2^{\Oh(k)} \\ \hline
\text{chordal} &h & \Oh(k^5) & \Oh(k^3) & \Oh(k^4\log^{3 / 2} k) \\
 &f & 2^{\Oh(k^2\log k)} & 2^{\Oh(k^2 \log k)} & 2^{\Oh(k^2\log k)} \\ \hline
 \text{(proper) interval} &h & \Oh(k^5) & \Oh(k^3) & \Oh(k^4\log^{3 / 2} k) \\
 &f & 2^{\Oh(k \log k)} & 2^{\Oh(k^2)} & 2^{\Oh(k)} \\ \hline
 \text{bipartite permutation} &h & \Oh(k^5) & \Oh(k^3) & \Oh(k^4\log^{3 / 2} k) \\
 &f & 2^{\Oh(k \log k)} & 2^{\Oh(k^2)} & 2^{\Oh(k)} \\ \hline
 \text{(linear) rankwidth $\le 1$} &h & \Oh(k^5) & \Oh(k^3) & \Oh(k^4\log^{3 / 2} k) \\
 &f & 2^{\Oh(k^2)} & 2^{\Oh(k^2)} & 2^{\Oh(k^2)} \\ \hline
 \text{planar} &h & \Oh(k^5) & \Oh(k^3) & - \\
 &f & 2^{\Oh(k^2\log k)} & 2^{\Oh(k^2\log k)} & \text{} \\ \hline
\text{forbidden} &h & \Oh(k^5) & \Oh(k^3) & - \\
 \text{minors}&f & 2^{k^{\Oh(1)}} & 2^{k^{\Oh(1)}} & \text{} \\ \hline
 \text{forbidden topological} &h & \Oh(k^5) & \Oh(k^3) & - \\
 \text{minors}&f & \text{FPT} & \text{FPT} & \text{} \\ \hline
\end{array}\]
\end{table}

\begin{thm}\label{thm:decomposition:full}
For each cell in \cref{table:decomposition} there exists an~algorithm
that, given graph $G$ satisfying $\hhdepth(G) \le k$ (resp. $\hhtw(G) \le k$), runs in time $f(k)\cdot n^{\Oh(1)}$, and returns an \hhdepthdecomp{} of depth $h(k)$ (resp. a~\hhtwdecomp{} of width $h(k)$).
The algorithms specified in the last column work in polynomial space.
\end{thm}
\begin{proof}
For $\hh = \mathsf{bipartite}$, we take advantage of the polynomial-time algorithm for $(\hh, 2t)$-\textsc{separation finding} from \cref{lem:separation-bipartite}.
To compute an $\hh$-elimination forest we combine it with Lemmas~\ref{lem:decomp-ed-polyspace} and~\ref{lem:decomp-ed-exact}.
Note that in the first case we achieve a polynomial-time algorithm.
From Lemmas~\ref{lem:separation-bipartite} and~\ref{lem:restricted-to-unrestricted} we obtain an algorithm for \textsc{Restricted $(\hh, 4t)$-separation finding} running in time $2^{\Oh(t\log t)}n^{\Oh(1)}$ \mic{and polynomial space}.
The algorithm computing a tree $\hh$-decomposition follows then from Lemma~\ref{lem:decomp-tw}.

Let now $\hh$ be given by a finite family of forbidden induced subgraphs, not necessarily connected.
The algorithms for computing an $\hh$-elimination forest are obtained by pipelining the algorithm for $(\hh, t)$-\textsc{separation finding} from Lemma~\ref{lem:separation-subgraph}, running in time $2^{\Oh(t)}n^{\Oh(1)}$ and polynomial space, with Lemmas~\ref{lem:decomp-ed-polyspace} and~\ref{lem:decomp-ed-exact}, respectively.
For the case of a tree $\hh$-decomposition, we supply Lemma~\ref{lem:decomp-tw} with the algorithm from Lemma~\ref{lem:separation-subgraph-restricted}.

The remaining results are obtained by supplying
\cref{thm:branching:final} with known algorithms for $\hh$-\textsc{deletion}, either approximate or exact, parameterized by the solution size.

First let us consider $\hh = \mathsf{chordal}$.
The \textsc{Chordal deletion} problem admits an FPT algorithm parameterized by the solution size $s$ running in time $f(n,s) = 2^{\Oh(s \log s)}n^{\Oh(1)}$~\cite{CaoM16} and polynomial space~\cite{CaoM16-communication}.
The claim follows from plugging this algorithm to \cref{thm:branching:final}.

For $\hh \in \{\mathsf{interval}, \mathsf{proper} \,\, \mathsf{interval}\}$, we take advantage of the existing constant-factor approximation algorithms for $\hh$-\textsc{deletion} \cite{CaoK16, HofV13}.
Pipelining them with \cref{thm:branching:final} yields the claim.
The same argument works for the class of bipartite permutation graphs, for which we can also use a~constant-factor approximation algorithm \cite{BozykDKNO20}.

When $\hh$ is a class of graphs of rankwidth at most 1 or linear rankwidth at most 1 (these graphs are also called {distance-hereditary}), we use algorithms for $\hh$-\textsc{deletion}, parameterized by the solution size $s$, running in time $2^{\Oh(s)}n^{\Oh(1)}$~\cite{EibenGK18, KanteKKP17}.
Both algorithms work in polynomial space~\cite{Kwon-communication}.

For $\hh = \mathsf{planar}$, we use the known algorithm for \textsc{Vertex planarization} parameterized by the solution size $s$, running in time $2^{\Oh(s \log s)}n^{\Oh(1)}$~\cite{JansenLS14}.
This algorithm however does not work in polynomial space so we do not populate the last cell in the row.

As a more general case, consider $\hh$ given by a finite family of forbidden minors.
\cref{thm:branching:final} can be applied directly only when all the graphs in this family are connected as otherwise $\hh$ might not be union-closed.
However, by
\cref{obs:disjoint-union-decomposition} and \cref{lem:disjoint-union-minors} the general problem of computing an $\hh$-elimination forest (or a tree $\hh$-decomposition) can be reduced to the case where all the forbidden minors are connected.
Therefore, we can employ \cref{thm:branching:final} in combination with the algorithm for $\hh$-\textsc{deletion}, parameterized by the solution size $s$, running in time $2^{s^{\Oh(1)}}n^{\Oh(1)}$~\cite{sau20apices}.
Analogous reasoning applies when $\hh$ is given by a finite family of forbidden topological minors.
Here we rely on the FPT algorithm for $\hh$-\textsc{deletion}~\cite{FominLPSZ20} with a computable yet unspecified dependency on the parameter, so we do not specify the function $f$ in the table.
\end{proof}

\section{Solving vertex-deletion problems} \label{sec:solving}

\bmp{We move on to solving particular vertex-deletion problems parameterized by $\hhtw$ or $\hhdepth$.
In~particular, we will be interested in $\hh$-\textsc{deletion} parameterized by graph measures related to $\hh$. We begin by presenting preliminaries specific for solving vertex-deletions problems in Section~\ref{sec:vertexdeletion:prelims}. In Section~\ref{sec:algorithms:adhoc} we show how several existing algorithms can be adapted to work with the parameterizations~$\hhtw$ and~$\hhdepth$, leading to algorithms solving \textsc{Odd Cycle Transversal} and \textsc{Vertex Cover}. In Section~\ref{sec:alg:generic} we present a meta-algorithm and apply it to several problems using new or known bounds on the sizes of minimal representatives of suitable equivalence classes.}

\subsection{Preliminaries for solving vertex-deletion problems} \label{sec:vertexdeletion:prelims}

{Let $T$ be a rooted tree. We denote the subtree of $T$ rooted at $t\in V(T)$ by $T_t$. The distance (in terms of the number of edges) from $t \in V(T)$ to the root of $T$ is denoted by $\depth_T(t)$. The height of a rooted tree~$T$ is the maximum number of edges on a path from the root to a leaf, we denote it by~$\height(T)$.  For an $\hh$-elimination forest $(T,\chi)$, let $V_t = \bigcup\{\chi(t') \mid t \text{ is an ancestor of } t'\}$. }

We begin with a~simple observation that will play {an} important role in several arguments.

\begin{lemma}\label{lem:optimal-set-separable-set}
Suppose $\hh$ is union-closed, $X \subseteq V(G)$ is a minimum $\hh$-deletion set in graph $G$ and $Z \subseteq V(G)$ satisfies $G[Z] \in \hh$. 
Then $|X \cap Z| \le |N(Z)|$.
\end{lemma}
\begin{proof}
Suppose that $|X \cap Z| > |N(Z)|$.
Consider an alternative solution candidate $X' = (X \setminus Z) \cup N(Z)$.
The graph $G - (N[Z] \cup X)$ is an induced subgraph of $G - X$, so it belongs to $\hh$, because the class is hereditary.
Furthermore, $G[Z] \in \hh$ and $G-X'$ is a~disjoint union of these graphs, so $X'$ is an $\hh$-deletion set.
To compare the sizes of $X$ and $X'$,
observe that when defining $X'$
we have removed all vertices from $X \cap Z$ and added at most $|N(Z)|$ new vertices, therefore $|X'| < |X|$.
This contradicts that $X$ is a~minimum deletion set.
\end{proof}

It is particularly useful when we are guaranteed that $Z$ has a~small neighborhood.
This is exactly the case for the base components in both types of decompositions: the base components of an~\hhdepthdecomp{} of depth $k$ (resp. \hhtwdecomp{} of width $k-1$) have neighborhood of size at most $k$.
We could thus assume that the intersection of any minimum $\hh$-deletion set with any base component is not too large.

\begin{corollary}\label{cor:optimal-set-separable-set}
Suppose $\hh$ is closed under disjoint unions of graphs, $X \subseteq V(G)$ is a minimum $\hh$-deletion set in $G$, and $Z \subseteq V(G)$ is a base component of an~\hhdepthdecomp{} of depth $k$ or a~\hhtwdecomp{} of width $k-1$.
Then $|X \cap Z| \le k$.
\end{corollary}

\paragraph{Nice decompositions}

We now introduce a standardized form of tree $\hh$-decompositions which is useful to streamline dynamic-programming algorithms. It generalizes the nice (standard) tree decompositions which are often used in the literature and were introduced by Kloks~\cite{kloks1994treewidth}.
\mic{A~similar construction has been introduced by Eiben et al.~\cite{EibenGHK19}, however we treat the leaves in a~slightly different way for our convenience.}


\begin{definition} \label{def:nice:tree:h:decomp}
Let~$\hh$ be a graph class. A tree $\hh$-decomposition~$(T,\chi,L)$ is \emph{nice} if the tree~$T$ is rooted at a vertex~$r$ such that $\chi(r) \cap L = \emptyset$ and the following properties hold in addition to those of Definition~\ref{def:tree:h:decomp}:
\begin{enumerate}
    \item Each node of~$T$ has at most two children.
    
    \item If~$t \in V(T)$ has two distinct children~$c_1, c_2$, then~$\chi(t) = \chi(c_1) = \chi(c_2)$. (This implies that~$\chi(t) \cap L = \chi(c_1) \cap L = \chi(c_2) \cap L = \emptyset$.) 
    \item If node~$t$ has a single child~$c$, then one of the following holds: 
        \begin{enumerate}
            \item There is a vertex~$v \in V(G) \setminus L$ such that~$\chi(t) = \chi(c) \cup \{v\}$.
            \item There is a vertex~$v \in V(G) \setminus L$ such that~$\chi(t) = \chi(c) \setminus \{v\}$.
            \item \mic{The node~$c$ is a leaf and~$\chi(t) = \chi(c) \setminus L$.}
        \end{enumerate}
\end{enumerate}
\end{definition}


As our approach to dynamic programming over a nice \jjh{tree $\hh$-}decomposition handles all nodes with a single child in the same way, regardless of whether a vertex is being introduced or forgotten compared to its child, we do not use the terminology of \textsc{introduce} or \textsc{forget} node in this work.

We are going to show that any tree $\hh$-decomposition can efficiently be transformed into a~nice one without increasing its width.
The analogous construction for standard tree decompositions is well-known and we shall use it as a black box.
In order to ensure that this process preserves existing bags,
we will rely on the following observation.

\begin{observation}\label{obs:treewidth-clique}
Suppose that $(T,\chi)$ is a tree decomposition of $G$ and the vertex set $A \subseteq V(G)$ forms a clique.
Then there exists a node $t \in V(T)$ such that $A \subseteq \chi(t)$.
\end{observation}

\begin{lemma} \label{lemma:makenice} 
Let~$\hh$ be a graph class. There is an algorithm that, given an $n$-vertex graph~$G$ and a tree $\hh$-decomposition~$(T,\chi,L)$ of~$G$ of width~$k$, runs in time~$\Oh(n + |V(T)|\cdot k^2)$ and outputs a \emph{nice} tree $\hh$-decomposition~$(T', \chi', L)$ of~$G$ of width at most~$k$ satisfying~$|V(T')| = \Oh(kn)$.
\end{lemma}
\begin{proof}
Kloks~\cite{kloks1994treewidth} provided an~algorithm for turning a tree decomposition into a~nice one with the same bounds as in the statement of the lemma.
The (standard) nice tree decomposition satisfies Definition~\ref{def:nice:tree:h:decomp} with $L=\emptyset$. 

Let $G_0 = G - L$ be the graph $G$ without the base components.
The tree $\hh$-decomposition $(T,\chi,L)$ induces a~(standard) tree decomposition $(T_0,\chi_0)$ of $G_0$ of the same width.
We would like to turn $(T_0,\chi_0)$ into a~nice decomposition and then append the base components back.
However, we would like to keep the bags to which the base components are connected. 
In order to ensure this, consider a~graph $G'_0$ obtained from $G_0$ by adding an~edge between each pair of vertices that reside in a~common bag of $(T_0,\chi_0)$.
Note that $(T_0,\chi_0)$ is a~valid tree decomposition of $G'_0$.

By the known construction~\cite{CyganFKLMPPS15, kloks1994treewidth},
we can build a~nice tree decomposition $(T'_0,\chi'_0)$ of $G'_0$
with $\Oh(k\cdot |V(G_0')|)$ nodes and width $k$, in time $\Oh(n + |V(T_0)|\cdot k^2)$.
If a set of vertices $B \subseteq V(G_0)$ forms a~bag in $(T_0,\chi_0)$, then $B$ is a clique in $G_0'$.
By Observation~\ref{obs:treewidth-clique}, there must be a bag in $(T'_0,\chi'_0)$ that contains $B$.

Let $C$ be a base component in $(T,\chi,L)$, that resides in the bag of node $t \in V(T)$.
By the observation above, there is a node $t' \in V(T'_0)$, such that $N(C) \subseteq \chi'_0(t')$.
For each base component $C$ we proceed as follows: (1) identify such a node $t_C \in V(T'_0)$, (2) create 2 copies $t_1, t_2$ of $t_C$, (3) move the original children (if any) of $t_C$ to be children of $t_1$, (4) make $t_C$ the parent of $t_1, t_2$, (5) create a~child node $t_3$ of $t_2$ with a~bag $\chi_0'(t_C) \cup C$.
Such a~modification of a~nice \hhtwdecomp{} preserves the conditions (1-3) and never modifies the bag of the root node.
For each base component we create only a~constant number of new nodes, so the upper bound is preserved.
\end{proof}

To exploit the separators which are encoded in a tree $\hh$-decomposition, the following definition is useful.

\begin{definition} \label{def:triseparation}
A \emph{tri-separation} in a graph~$G$ is a partition~$(A,X,B)$ of~$V(G)$ such that no vertex in~$A$ is adjacent to any vertex of~$B$. The set~$X$ is the \emph{separator} corresponding to the tri-separation. The \emph{order} of the tri-separation is defined as~$|X|$.
\end{definition}

The following notions will be used to extract tri-separations from rooted tree decompositions. Recall that~$T_t$ denotes the subtree of~$t$ rooted at~$T$.

\begin{definition} \label{def:kappa}
Let~$\hh$ be a graph class and let~$(T,\chi,L)$ be a nice tree $\hh$-decomposition of a graph~$G$, rooted at some node~$r$. For each node~$t \in V(T)$ of the decomposition, we define the functions~$\pi_{T,\chi}, \kappa_{T,\chi} \colon V(T) \to 2^{V(G)}$:
\begin{enumerate}
    \item For a non-root node~$t$ with parent~$s$, we define~$\pi_{T,\chi}(t) := \chi(s)$. For~$r$ we set~$\pi_{T,\chi}(r) = \emptyset$.
    \item For an arbitrary node~$t$, we define~$\kappa_{T,\chi}(t) := (\bigcup _{t' \in T_t} \chi(t')) \setminus \pi(t)$.
\end{enumerate}
We omit the subscripts~$T,\chi$ when they are clear from context.
\end{definition}

{Intuitively,~$\pi(t)$ is the bag of the parent of~$t$ (if there is one) and~$\kappa(t)$ consists of those vertices that occur in bags in the subtree~$T_t$ but not in the parent bag of~$t$. The following observations about~$\kappa$ will be useful later on.

\begin{observation} \label{obs:kappa}
If~$(T,\chi,L)$ is a nice tree $\hh$-decomposition of a graph~$G$, then the following holds.
\begin{enumerate}
    \item For the root~$r$ of the decomposition tree,~$\kappa(r) = V(G)$.
    \item If node~$t'$ is a child of node~$t$, then~$\kappa(t') \subseteq \kappa(t)$. 
    \item If~$c_1, c_2$ are distinct children of a node~$t$, then~$\kappa(c_1) \cap \kappa(c_2) = \emptyset$.
\end{enumerate}
\end{observation}
}

Tri-separations can be deduced from tree decompositions using~$\kappa$ and~$\pi$.

\begin{observation} \label{obs:triseparation:from:td}
Let~$\hh$ be a graph class and let~$(T,\chi,L)$ be a nice tree $\hh$-decomposition of graph~$G$, where~$t$ is rooted at some node~$r$. Let~$t \in V(T)$, let~$A := \kappa(t)$ and let~$X := \chi(t) \cap \pi(t)$. Then~$(A, X, V(G) \setminus (A \cup X))$ is a tri-separation in~$G$.
\end{observation}

\subsection{Incorporating hybrid parameterizations into existing algorithms} \label{sec:algorithms:adhoc}

\subsubsection{Odd cycle transversal}
In this section we show how to obtain an FPT algorithm for the \textsc{Odd cycle transversal (oct)} problem when given either a $\mathsf{bip}$-elimination forest or a tree $\mathsf{bip}$-decomposition. 
In the \textsc{Odd cycle transversal} problem we are given a graph $G$, and ask for the (size of a) minimum {vertex} set $S \subseteq V(G)$ such that $G-S$ is bipartite. Such a set $S$ is called an odd cycle transversal. We introduce some more terminology. Recall that a graph is bipartite if and only if it admits a proper 2-coloring, this 2-coloring is referred to as a bipartition. For a bipartite graph $G$ we say that $A,B \subseteq V(G)$ occur in opposite sides of a bipartition of $G$, if there is a proper 2-coloring $c$ of $G$ such that $A \subseteq c^{-1}(1)$ and $B \subseteq c^{-1}(2)$. Another characterization is that a graph is bipartite if and only if it does not contain an odd cycle, explaining the name \textsc{Odd cycle transversal}.

In most classical branching algorithms on elimination forests or {dynamic-programming} algorithms on tree decompositions, some trivial constant time computation is done in the leaves. In our setting we do more work in the leaves, without changing the dependency on the parameter.

\paragraph{Given a $\mathsf{bip}$-elimination forest}
The algorithms we present sometimes need to be able to decide that an instance has no solution {and then the algorithm should return $\bot$.} Our union operator $\cup$ has the additional feature that it propagates this information, that is, $A \cup B = \bot$ if $A = \bot$ or $B = \bot$. 
Before we present the algorithm, we introduce the following problem that will be used as a subroutine later on.

\defproblem{Annotated bipartite coloring (abc)}{A bipartite graph $G$, two sets $B_1,B_2 \subseteq V(G)$, and an integer $k$.}{Return {a} minimum-cardinality set $X \subseteq V(G)$ such that $G-X$ has a bipartition with $B_1 \setminus X$ and $B_2 \setminus X$ on opposite sides, or return $\bot$ if no such $X$ exists of size at most $k$.}

The sets $B_1$ and $B_2$ can be seen as precolored vertices that any coloring after deletion of some vertices should respect. The problem is solvable in polynomial time.

\begin{lemma}[{\cite[Lemma 4.15]{CyganFKLMPPS15}}]
\label{lem:abc}
\textsc{Annotated bipartite coloring} can be solved in $\Oh(k(n+m))$ time.
\end{lemma}

We present a branching algorithm for computing an odd cycle transversal of minimum size in Algorithm~\ref{alg:oct-elim}. The general idea is to branch on the top-most vertex $t$ of a $\mathsf{bip}$-elimination forest, and recursively call annotated subinstances. Since each edge has an ancestor-descendant relation in the elimination forest, we only need to annotate the ancestors of $t$ to be able to solve the problem. In the case of \textsc{oct}, a vertex $v \in V(G)$ is either in the deletion set or in one of two color classes in the resulting bipartite graph. We either return a minimum odd cycle transversal that respects the annotations, or conclude that {they} cannot lead to an optimal solution {for} the initial graph. The idea of branching on a top-most vertex of an elimination forest {has been used by, e.g.,} Chen et al.~\cite{ChenRRV18} to show that \textsc{Dominating set} is FPT {when} parameterized by treedepth. They also {provide} branching strategies for \textsc{$q$-coloring} and \textsc{Vertex cover}.

\begin{algorithm}
\caption{\textsc{oct} (Graph $G$, $t \in V(T)$, $\mathsf{bip}$-elimination forest $(T,\chi)$, $(S_1,S_2,S_X)$)}
\label{alg:oct-elim}
\begin{algorithmic}[1]
\Input 
\Statex $(T,\chi)$ is a $\mathsf{bip}$-elimination forest of $G$ that consists of a single tree. 
\Statex $S = S_1 \cup S_2 \cup S_X = \{ \chi(t') \mid t' \text{ is a proper ancestor of } t\}$.
\Statex $S_1$, $S_2$, and $S_X$ are pairwise disjoint.
\Output
\Statex An odd cycle transversal $X_t \subseteq V_t$ of $G[V_t \cup S_1 \cup S_2]$ of minimum size such that (1) $S_1$ and $S_2$ are in opposite sides of a bipartition of $G[V_t \cup S_1 \cup S_2] - X_t$ and (2) $|X_t \cap \chi(t')| \leq \depth_T(t')$ for each leaf $t' \in V(T_t)$. If no such $X_t$ exists, {then} return $\bot$.
\If {$t$ is a leaf}
\If{$\exists_{u,w \in S_1} uw \in E(G) \vee  \exists_{u,w \in S_2} uw \in E(G)$}
\State \Return $\bot$
\Else
\State \Return \textsc{abc}$(G[\chi(t)],N_G(S_1) \cap \chi(t),N_G(S_2) \cap \chi(t),\depth_T(t))$
\EndIf
\Else
\State Let $U_t$ be the children of $t$.
\State Let $X_X = \chi(t) \cup \bigcup_{t' \in U_t}\textsc{oct}(G, t', (T,\chi), (S_1,S_2,S_X\cup \chi(t)))$.
\State Let $X_1 = \bigcup_{t' \in U_t}\textsc{oct}(G, t', (T,\chi), (S_1 \cup \chi(t),S_2,S_X))$.
\State Let $X_2 = \bigcup_{t' \in U_t}\textsc{oct}(G, t', (T,\chi), (S_1,S_2\cup \chi(t),S_X))$.
\State \Return
smallest of $X_X, X_1, X_2$ that is not $\bot$, or $\bot$ otherwise.
\EndIf
\end{algorithmic}
\end{algorithm}


\begin{thm}\label{thm:octhhdepth}
\textsc{Odd cycle traversal} can be solved in time $\Oh(3^d \cdot d^2 \cdot (n+m))$ time and polynomial space when given a $\mathsf{bip}$-elimination forest of depth~$d$.
\end{thm}
\begin{proof}
Let~$(F,\chi)$ be an $\mathsf{bip}$-elimination forest of depth~$d$ of the input graph~$G$. We solve \textsc{Odd cycle transversal} by taking the union of \textsc{oct}$(G[V_{r_T}],r_T,(T,\chi_T),(\emptyset,\emptyset,\emptyset))$ over all trees $T$ in the forest $F$. Here $r_T \in V(T)$ is the root of $T$, and $\chi_T$ is the function $\chi$ restricted to $V(T)$. To see that this solves the problem, by Corollary~\ref{cor:optimal-set-separable-set} any optimal odd cycle transversal $Y \subseteq V(G)$ of $G$ would {remove} at most $\depth_T(t')$ vertices from $\chi(t')$ for every leaf $t' \in V(T)$.
We argue the correctness and running time of Algorithm~\ref{alg:oct-elim} by induction on $\height(T_t)$. Consider a call \textsc{oct}$(G,t,(T,\chi),(S_1,S_2,S_X))$ that satisfies the input requirements. 

If $\height(T_t) = 0$, then $t$ is a leaf and $G[\chi(t)]$ is a bipartite base component. For any set $Y_t \subseteq V_t$ that satisfies the output requirements, we must have that $S_1$ and $S_2$ are independent sets. Thus if this is not the case, no such $Y_t$ exists and hence we can return $\bot$. Otherwise, in a proper 2-coloring of the remaining graph, every vertex in $N_G(S_1) \cap \chi(t)$ that is not deleted by $Y_t$ needs to get the color opposite to $S_1$. Similarly, every vertex in $N_G(S_2) \cap \chi(t)$ that is not deleted by $Y_t$ needs to get the a color opposite to $S_2$. Therefore, {any solution to the annotated \textsc{oct} problem also solves the} \textsc{Annotated bipartite coloring} instance on $G[\chi(t)]$ with precolored vertex sets $N_G(S_1) \cap \chi(t)$ and $N_G(S_2) \cap \chi(t)$ with budget $\depth_T(t)$. {Conversely, for any solution~$X$ to the \textsc{abc} instance, the graph~$G[\chi(t)] - X$ has a proper 2-coloring with~$N_G(S_1) \cap \chi(t)$ and~$N_G(S_2) \cap \chi(t)$ in opposite color classes. By flipping the color classes if needed, we may assume~$N_G(S_1) \cap \chi(t)$ is colored~$2$ and~$N_G(S_2) \cap \chi(t)$ is colored~$1$. Since~$S_1$ and~$S_2$ are independent sets, we can extend this $2$-coloring of~$G[\chi(t)] - X$ to a $2$-coloring of~$G[V_t \cup S_1 \cup S_2] - X = G[\chi(t) \cup S_1 \cup S_2] - X$ by simply assigning~$S_1$ color~$1$ and~$S_2$ color~$2$. Hence the call to \textsc{abc} results in the correct answer to the annotated \textsc{oct} problem.}

In the case that $t$ is not a leaf, {fix a solution $Y_t \subseteq V_t$ that satisfies the output requirements if there is one; otherwise let~$Y_t = \bot$. The} single vertex in $\chi(t)$ must either be in the odd cycle transversal, or in one of two color classes of the remaining bipartition. We branch on each of these cases, and construct solutions $X_X$, $X_1$, and $X_2$. Initially we add $\chi(t)$ to the transversal $X_X$, in the other cases $\chi(t)$ will be annotated to be in a specific color class of the resulting bipartition. In each branch, we recursively call the algorithm for every child $t'$ and take the union of their results. Since $\depth_T(t') < \depth_T(t)$, we can assume that these recursive calls are correct by the induction hypothesis. We make the following claims.

\begin{claim}
If $X_i \neq \bot$ for some $i \in \{1,2,X\}$, then $X_i$ satisfies the output requirements.
\end{claim}
\begin{innerproof}
We show the argument for $i = 1$. Note that the set of leaves of $T_t$ is the union of the set of leaves of $T_{t'}$ for $t' \in U_t$, where $U_t$ is the set of children of $t$. Therefore the second requirement is satisfied by $X_1$. Clearly $X_1 \subseteq V_t$. Let $X_t'$ be the odd cycle transversal of $G[V_{t'} \cup S_1 \cup \chi(t) \cup S_2]$ for each child $t'$ of $t$, and let $c_{t'}$ be a 2-coloring of $G[V_{t'} \cup S_1 \cup \chi(t) \cup S_2] - X_{t'}$ that colors $S_1 \cup \chi(t)$ with color $1$ and $S_2$ with color $2$. We construct a 2-coloring $c_t$ of $G[V_{t} \cup S_1 \cup S_2] - X_1$. For $v \in V_t \cup S_1 \cup S_2$ let $c_t(v) = c_{t'}(v)$ if $v \in V_{t'}$, $c_t(v) = 1$ if $v \in S_1 \cup \chi(t)$, and $c_t(v) = 2$ if $v \in S_2$. Since any edge in $E(G)$ has an ancestor-descendant relationship in $T$, it follows that $c_t$ is a proper 2-coloring of $G[V_t \cup S_1 \cup S_2] - X_1$. Therefore, $X_1$ is an odd cycle transversal. Note that the first requirement is also satisfied, since $c_t$ has $S_1$ and $S_2$ in opposite sides of the bipartition. The reasoning for $X_2$ {is symmetric. For~$X_X$, the only difference is that the unique vertex in~$\chi(t)$ is removed instead of colored.} The claim follows.
\end{innerproof}
From the above claim we have that $|Y_t| \leq |X_i|$ in the case that $X_i \neq \bot$ for $i \in \{1,2,X\}$.
\begin{claim}
If $Y_t \neq \bot$ is a solution that satisfies the output requirements, then $X_i \neq \bot$ and $|X_i| \leq |Y_t|$ for some $i \in \{1,2,X\}$.
\end{claim}
\begin{innerproof}
Let $v_t$ be the unique vertex in $\chi(t)$. Either $v_t \in Y_t$, or $v_t$ should be in one of two color classes of the remaining bipartite graph. 
Let $c$ be a 2-coloring of $G[V_t \cup S_1 \cup S_2] - Y_t$ that has $S_1$ and $S_2$ colored 1 and 2 respectively. First consider the case that $v_t \notin Y_t$ and $c(v_t) = 1$. We argue that for each child $t' \in U_t$, the set $Y_t \cap V_{t'}$ is some solution to the call $\textsc{oct}(G, t', (T,\chi), (S_1 \cup \chi(t),S_2,S_X))$. Since each leaf in $T_{t'}$ is also a leaf in $T_t$, the second requirement is satisfied by $Y_t \cap V_{t'}$. Clearly $Y_t \cap V_{t'}$ is  an odd cycle transversal of $G'= G[V_{t'} \cup S_1 \cup \chi(t) \cup S_2]$ such that $S_1 \cup \chi(t)$ and $S_2$ are in opposite sides of a bipartition; the coloring $c$ restricted to $G'$ certifies this. Since $Y_t \cap V_{t'}$ is some solution \mic{for $t'$}, \bmp{an optimal solution for~$t'$ is not larger, so }
by induction it follows that the output $X_{t'}$ of the recursive \mic{call} satisfies $|X_{t'}| \leq |Y_t \cap V_{t'}|$. By observing that $Y_t = \bigcup_{t' \in U_t} Y_t \cap V_{t'}$, we get that $X_1 = \bigcup_{t' \in U_t} X_{t'}$ satisfies $|X_1| \leq |Y_t|$. 

The argument for $v_t \notin Y_t$ and $c(v_t) = 2$ is symmetric. The argument for $v_t \in Y_t$ is similar; here $Y_t = \{v_t\} \cup \bigcup_{t' \in U_t} Y_t \cap V_{t'}$ and therefore $X_X = \{v_t\} \cup \bigcup_{t' \in U_t} X_{t'}$ satisfies $|X_X| \leq |Y_t|$.  
%
\end{innerproof}

From the two claims above the correctness follows. For $x \in V(T)$, let $n_x$ and $m_x$ denote the number of vertices and edges of $G[V_x]$, respectively. Recall that~$d$ is the depth of the overall elimination forest~$(F,\chi)$, so that~$\height(T_t), \depth_T(t) \leq d$ for all~$t \in V(T)$. To help bound the running time, we use the number of edges in the subtree~$T_t$ as an additional measure of complexity.

\begin{claim}
There is a constant~$\alpha$ such that \textsc{oct}$(G, t, (T,\chi), (S_1,S_2,S_X))$ can be implemented to run in time bounded by $\alpha \cdot 3^{\height(T_t)} \cdot d^2 \cdot (n_t+m_t + |E(T_t)|)$.
\end{claim}
\begin{innerproof}
Proof by induction on~$\height(T_t)$. 

First consider the case that $\height(T_t) = 0$. Note that $G[\chi(t)]$ is a base component of $T$ and therefore $G[\chi(t)]$ is bipartite. To verify whether~$S_1$ and~$S_2$ are independent sets, we need to be able to efficiently test whether one vertex is a neighbor of another. This would be trivial using an adjacency-matrix representation of the graph, but it would require quadratic space. Instead, we enrich the data structure storing the $\mathsf{bip}$-elimination forest~$(F,\chi)$ to allow efficient adjacency tests. This enrichment step is performed before the first call to~$\textsc{oct}$. With each internal node~$t$ of~$F$, which represents a single vertex~$\{v\} = \chi(t)$, we will store a linked list of those edges which connect~$v$ to a vertex stored in the bag of an ancestor of~$t$. Note that, by definition, this list has length at most~$\depth_F(t)$. It is easy to see that this enrichment step can be done in time~$\Oh(n + m)$. 

Using this additional information, we can verify whether~$S_1$ and~$S_2$ are independent sets in time~$\Oh(|S|^2) = \Oh(d^2)$. For each~$i \in [2]$, for each vertex~$s$ of~$S_i$ (stored in the bag~$\chi(t') = \{s\}$ of an ancestor~$t'$ of~$t$), iterate through the list of edges connecting~$s$ to an ancestor. For each of those edges~$s s'$ in the list, we can determine whether~$s' \in S_i$, by setting a local flag in all vertex objects for~$S_i$ to indicate their membership. As each list has length at most~$d$, we can test in~$\Oh(|S| \cdot d) = \Oh(d^2)$ time whether~$S_1$ and~$S_2$ are indeed independent. If so, we solve an \textsc{abc} instance. By Lemma~\ref{lem:abc} this takes $\Oh(|S|(n_t+m_t)) = \Oh(d(n_t+m_t))$ time. Hence there exists a constant~$\alpha$ such that the base case is handled in~$\alpha \cdot 3^{\height(T_t)} \cdot d^2 \cdot (n_t+m_t + |E(T_t)|)$ time.

Next consider the case that $\height(T_t) > 0$. We branch on three instances with the unique vertex in $\chi(t)$ added to each of $S_1$, $S_2$, and $S_X$ respectively. In each branch, we call the algorithm recursively for each child $t'$ of $t$. Since $\height(T_{t'}) \leq \height(T_t) - 1$, by the induction hypothesis each instance can be solved in $\alpha \cdot 3^{\height(T_{t'})} \cdot d^2 \cdot (n_{t'}+m_{t'} + |E(T_{t'})|)$ time. In addition to solving the three recursive calls, the algorithm does some additional work in taking the union of result sets for child subtrees and determining which of these is largest. By storing solutions in linked lists, the union can be computed in time linear in the number of children by concatenating the lists. Hence the additional work that has to be done on top of the recursive calls, can be handled in time~$\Oh(|U_t|)$, linear in the number of children. Hence if~$\alpha$ is sufficiently large, the additional work can be done in time~$\alpha \cdot |U_t|$. Note that the edges from~$t$ to its children are present in~$T_t$ but not in the subtrees of its children, so that~$|E(T_t)| = |U_t| + \sum _{t' \in U_t} |E(T_{t'})|$. We can now bound the total running time~$N$ as follows:
\begin{align*}
N & \leq \alpha \cdot |U_t| + 3 \cdot \sum _{t' \in U_t} \alpha \cdot 3^{\height(T_{t'})} \cdot d^2 \cdot (n_{t'}+m_{t'} + |E(T_{t'})|) & \mbox{By induction} \\
& \leq \alpha \cdot |U_t| + 3 \cdot \alpha \cdot d^2 \cdot \sum _{t' \in U_t} 3^{\height(T_t) - 1} \cdot (n_{t'}+m_{t'} + |E(T_{t'})|) & \mbox{$t'$ is a child of~$t$} \\
& \leq \alpha \cdot |U_t| + \alpha \cdot 3^{\height(T_t)} \cdot d^2 \cdot \sum _{t' \in U_t} \cdot (n_{t'}+m_{t'} + |E(T_{t'})|) & \mbox{Rewriting} \\
& \leq \alpha \cdot |U_t| + \alpha \cdot 3^{\height(T_t)} \cdot d^2 \cdot (n_t + m_t + |E(T_t)| - |U_t|) & \mbox{Bound on~$|E(T_t)|$} \\
& \leq \alpha \cdot |U_t| + \left (\alpha \cdot 3^{\height(T_t)} \cdot d^2 \cdot (n_t + m_t + |E(T_t)|) \right) - \alpha \cdot |U_t|) & \mbox{Rewriting} \\
& \leq \alpha \cdot 3^{\height(T_t)} \cdot d^2 \cdot (n_t + m_t + |E(T_t)|) & \mbox{Canceling terms}
\end{align*}
This proves the claim.
\end{innerproof}
Since we initially call the algorithm on each tree of the elimination forest, which we can compute in linear time, the result follows. To see the polynomial space {guarantee}, note that the recursion tree has depth at most $d+1$, and for each instance we store a single node $t \in V(T)$, annotations $(S_1,S_2,S_X)$ of size $\Oh(d)$, and construct three odd cycle transversals $X_X$, $X_1$, and $X_2$ of $\Oh(n)$ vertices. Furthermore the {call to} \textsc{abc} runs in polynomial time and therefore uses polynomial space too.
\end{proof}

In the algorithm above we assumed we were given some decomposition. Since we have seen that we can compute {approximate} $\mathsf{bip}$-elimination forests {in FPT time}, we have the following.

\begin{corollary}\label{thm:octelimcombine}
\textsc{Odd cycle transversal} is FPT parameterized by $k = \bipdepth(G)$, furthermore it is solvable in $n^{\Oh(1)} + 2^{\Oh(k^{3}\log k)} \cdot (n+m)$ time and polynomial space.
\end{corollary}
\begin{proof}
Let $k = \bipdepth(G)$. \mic{By \cref{thm:decomposition:full},} we can compute a $\mathsf{bip}$-elimination forest of depth $\Oh(k^3 \log k)$ {in polynomial time and space}. The result follows by supplying this elimination forest to Theorem~\ref{thm:octhhdepth}.
\end{proof}

We remark that this running time can be slightly improved by using the other result \mic{from \cref{thm:decomposition:full}} as we could compute a $\mathsf{bip}$-elimination forest of depth $\Oh(k^2)$, but then we lose the polynomial space guarantee.

\paragraph{Given a tree $\mathsf{bip}$-decomposition}
Fiorini et al.~\cite{FioriniHRV08} give an explicit algorithm for computing the size of a minimum odd cycle transversal in a tree decomposition of width $\tw$ in time $\Oh(3^{3\tw} \cdot \tw \cdot n)$. We follow their algorithm to show that a similar strategy works when given a tree $\mathsf{bip}$-decomposition. Most of the reasoning overlaps with~\cite{FioriniHRV08}, we include some of it for completeness.

\begin{thm}\label{thm:oct-dp}
\textsc{Odd cycle traversal} can be solved in time $3^{3k}\cdot n^{\Oh(1)}$ when given a tree $\mathsf{bip}$-decomposition of width $k-1$ consisting of~$n^{\Oh(1)}$ nodes.
\end{thm}
\begin{proof}
{First, we apply Lemma~\ref{lemma:makenice} to turn the given tree $\mathsf{bip}$-decomposition~$(T,\chi,L)$ into a nice one.
This takes time~$\Oh(n + |V(T)|\cdot k^2)$ and generates a {nice} tree $\mathsf{bip}$-decomposition of the same width, such that ~$|V(T)| = \Oh(kn)$.
In particular, we can assume that $T$ is a binary tree.}
 
Let $t^* \in V(T)$ be the root of $T$. Furthermore, for $t \in V(T)$, let $T_t$ be the subtree of $T$ rooted at $t$. For each $e = uv \in E(G)$, assign a specific node $t(e) \in V(T)$, such that $u,v \in \chi(t)$. Let $E_t$ be the set of edges between the vertices in $\chi(t)$ for $t \in V(T)$. Let us define graph $G(t)$ with vertex set $\chi(t)$ and edge set $E_t$, and the graph $G(T_t)$ with vertex set $\bigcup_{t' \in T_t} \chi(t')$ and edge set $\bigcup_{t' \in T_t}E_{t'}$. We associate with $t \in V(T)$ a set $\mathcal{A}_t$ of ordered triples $\Pi_t = (L_t,R_t,W_t)$, which forms a 3-partition of $\chi(t) \setminus L$.  Note that $|\mathcal{A}_t|$ is at most $3^k$.
 
The algorithm works from the leaves up, and for each partition $\Pi_t$ stores the size of a minimum odd cycle transversal $\hat{W}_t$ in $G(T_t)$ such that $W_t \subseteq \hat{W}_t$ and $L_t$ and $R_t$ are in opposite color classes of $G(T_t) - \hat{W}_t$. This value is stored in $f(\Pi)$, and is infinity if no such transversal exists. For a non-leaf $t$ with a child $r$, a partition of $\Pi_r$ is said to be consistent with $\Pi_t$, denoted $\Pi_r \sim \Pi_t$, if $W_t \cap V(S_r) \subseteq W_r$, $L_t \cap V(S_r) \subseteq L_r$, and $R_t \cap V(S_r) \subseteq R_r$.
 Let $r, s$ be children of $t$; if there is only one child, we omit the terms for $s$.
 If both $L_t$ and $R_t$ are independent sets in $G(t)$, we have

\[f(\Pi_t) = \min_{\Pi_r \sim \Pi_t, \Pi_s \sim \Pi_t} f(\Pi_r) + f(\Pi_s) + |W_t - (W_r \cup W_s)| - |W_r \cap W_s|\]

Otherwise, set $f(\Pi_t) = \infty$. For a leaf $t \in V(T)$, we set 

\[f(\Pi_t) = |W_t| + \textsc{abc}(G[\chi(t)\cap L],N(L_t) \cap \chi(t), N(R_t) \cap \chi(t), n)\]
Since $\chi(t) \cap L = \emptyset$ for non-leaf nodes $t \in V(T)$, the correctness follows as in~\cite{FioriniHRV08}. In each leaf node, similarly as in Theorem~\ref{thm:octhhdepth},
we obtain an instance of \textsc{Annotated bipartite coloring} (here with budget $n$), which can be solved in polynomial time (Lemma~\ref{lem:abc}). 

The total work in the internal nodes remains $\Oh(k \cdot 3^{3k} \cdot n)$ as in~\cite{FioriniHRV08}.  The total work in the leaves takes $3^k \cdot n^{\Oh(1)}$ time. 
\end{proof}

We note that the algorithm of Fiorini et al.~\cite{FioriniHRV08} we adapted above does not have the best possible dependency on treewidth. {The dependency on $\hhtw[\mathsf{bip}]$ of our algorithm can be improved by exploiting more properties of nice tree $\hh$-decompositions, to improve the exponential factor to~$3^k$. Similarly, the polynomial in the running time can be improved to~$k^{\Oh(1)} \cdot (n + m)$. Two additional ideas are needed to achieve this speed-up. First of all, the computations for \textsc{Annotated bipartite coloring} in the leaf nodes can be cut off. By Corollary~\ref{cor:optimal-set-separable-set}, an optimal solution contains at most~$k+1$ vertices from any base component of a $\mathsf{bip}$-decomposition of width~$k$. Hence it suffices to set the budget for the \textsc{abc} instance to~$k+1$, and assign a table cell the value~$\infty$ if no solution was found. While this causes some entries in the table to become~$\infty$, the optimum solution and final answer are preserved. The second additional idea that is needed is an efficient data structure for testing adjacencies in graphs of bounded $\mathsf{bip}$-treewidth, similarly as explained in~\cite{BodlaenderBL13}. We chose to present a simpler yet slower algorithm for ease of readability.}


Since we can compute tree $\mathsf{bip}$-decompositions approximately, we conclude with the following.

\begin{corollary}\label{thm:octdpcombine}
\bmp{The \textsc{Odd cycle transversal} problem can be solved in time $2^{\Oh(k^{3})} \cdot n^{\Oh(1)}$ when parameterized by~$k = \biptw(G)$.}
\end{corollary}
\begin{proof}
\mic{By \cref{thm:decomposition:full}} we can compute a tree $\mathsf{bip}$-decomposition of width $\Oh(k^3)$ in $2^{\Oh(k\log k)}n^{\Oh(1)}$ time. The result follows by supplying this decomposition to Theorem~\ref{thm:oct-dp}. {The algorithm that computes the minimum size of an odd cycle transversal can be turned into an algorithm that constructs a minimum solution by standard techniques or self-reduction, in the same asymptotic time bounds.}
\end{proof}

\subsubsection{Vertex cover}
In this section we give algorithms for the \textsc{Vertex cover} problem when given a decomposition with base components of some hereditary graph class $\hh$ in which the problem is polynomial-time solvable. In the \textsc{Vertex cover} problem, we are given a graph $G$, and ask for the (size of a) minimum vertex set $S \subseteq V(G)$ such that for each $uv \in E(G)$, $S \cap \{u,v\} \neq \emptyset$. Such a set $S$ is called a vertex cover.

Note that the \textsc{Vertex cover} problem can be stated as hitting forbidden induced subgraph, namely $K_2$, which we {will treat in Section~\ref{sec:hitting:subgraphs}.}
{However for the special case of \textsc{Vertex cover}, we are able to handle the problem for a broader spectrum of parameterizations.}

\paragraph{Given an $\hh$-elimination forest}
We present a branching algorithm for \textsc{Vertex cover} in Algorithm~\ref{alg:vc-elim} that is similar to Algorithm~\ref{alg:oct-elim} for \mic{\textsc{Odd cycle transversal}}. Again we branch on the top-most vertex of an elimination forest. In the case of \textsc{Vertex cover} a vertex is either in the vertex cover (and added to $S_I$) or outside the vertex cover (and added to $S_O$). The intuition for the computation in the base components is as follows. For any ancestor that was chosen to not be in the vertex cover, any neighbor contained in the base component must be in the vertex cover. Since $\hh$ is hereditary, the base component where these vertices are deleted {still belongs to} $\hh$. Therefore we can run an algorithm on the remaining graph that solves vertex cover in polynomial time on graphs from $\hh$. Let $\textsc{A}_\hh$ be a polynomial-time algorithm that finds a~minimum {vertex cover} in a~graph $G \in \hh$.

\begin{algorithm}
\caption{\textsc{vc} (Graph $G$, $t \in V(T)$, $\hh$-elimination forest $(T,\chi)$, $(S_I,S_O)$)}
\label{alg:vc-elim}
\begin{algorithmic}[1]
\Input 
\Statex $(T,\chi)$ is an $\hh$-elimination forest of $G$ that consists of a single tree. 
\Statex $S = S_I \cup S_O = \{ \chi(t') \mid t' \text{ is a proper ancestor of } t\}$. 
\Statex $S_I$ and $S_O$ are disjoint.
\Output
\Statex A vertex cover $X_t \subseteq V_t$ of $G[V_t \cup S_O]$ of minimum size. If no such $X_t$ exists return $\bot$.
\If {$t$ is a leaf}
\If{$\exists_{u,w \in S_O} uw \in E(G)$}
\State \Return $\bot$
\Else
\State Use $\textsc{A}_\hh$ to compute a minimum vertex cover~$X$ of~$G[\chi(t) \setminus N_G(S_O)]$
\State \Return $X \cup (N_G(S_O) \cap \chi(t))$
\EndIf
\Else
\State Let $U_t$ be the children of $t$.
\State Let $X_I = \chi(t) \cup \bigcup_{t' \in U_t}\textsc{vc}(G, t', (T,\chi), (S_I \cup \chi(t),S_O))$.
\State Let $X_O = \bigcup_{t' \in U_t}\textsc{vc}(G, t', (T,\chi), (S_I,S_O \cup \chi(t)))$.
\State \Return smallest of $X_I$ and $X_O$ that is not $\bot$, or $\bot$ otherwise.
\EndIf
\end{algorithmic}
\end{algorithm}

We note that the branching can be improved using the following idea. Whenever we put some vertex $v$ in $S_O$ (so $v$ is not part of the cover), then we must have all of $N_G(v)$ in the vertex cover. We do not include it here for ease of presentation.

\begin{thm}\label{thm:vcelim}
Let $\hh$ be a hereditary graph class in which \textsc{Vertex cover} is polynomial time solvable. Then given an $\hh$-elimination forest $(F,\chi)$ with {depth} 
$d$, \textsc{Vertex cover} is solvable in $2^d \cdot n^{\Oh(1)}$ time {and polynomial space}.
\end{thm}
\begin{proof}
We solve \textsc{Vertex cover} by summing the results of \textsc{vc}$(G[V_{r_T}],r_T,(T,\chi_T),(\emptyset,\emptyset,\emptyset))$ over all trees $T$ in the forest $F$. Here $r_T \in V(T)$ is the root of $T$, and $\chi_T$ is the function $\chi$ restricted to $V(T)$. We argue the correctness of Algorithm~\ref{alg:vc-elim} by induction on $\height(T_t)$. Consider a call \textsc{vc}$(G,t,(T,\chi),(S_I,S_O))$ that satisfies the input requirements. 

If $\height(T_t) = 0$, then $t$ is a leaf and $G[\chi(t)] \in \hh$. For any set $Y_t \subseteq V_t$ that satisfies the output requirements, the set $S_O$ must be an independent set. If this is not the case, we correctly return $\bot$. Otherwise for $Y_t$ to be a vertex cover, every vertex from $N_G(S_O) \cap \chi(t)$ must be in the vertex cover. Therefore we return the union of $N_G(S_O) \cap \chi(t)$ and the result of $\textsc{A}_\hh(G[\chi(t) \setminus N_G(S_O)])$.

Otherwise, if $\height(T_t) > 0$, then $t$ is not a leaf. We branch on the unique vertex in $\chi(t)$ to be either in the vertex cover or not. In each branch, we recursively call the algorithm for each child $t'$ of $t$. Since $\height(T_{t'}) \leq \height(T_t) - 1$, by the induction hypothesis these recursive calls are correct. Since there is an ancestor descendant relation for every edge in $E(G)$, the union of vertex covers for in the recursively called subinstances form a vertex cover for the current instance. Since this branching is exhaustive, an optimal vertex cover is found in one of the branches if it exists. If $X_I = X_O = \bot$, then we can return $\infty$ for the current instance as well.

{Finally, we bound the running time and space requirements, by analyzing the recursion tree generated by the initial calls on each tree~$T$ of the forest~$F$. For each node~$t \in V(T)$, there are exactly~$2^{\depth_T(t)}$ recursive calls to Algorithm~\ref{alg:vc-elim} with~$t$ as the second parameter, one for every partition of the nodes in the bags of ancestors of~$t$ into~$S_I$ and~$S_O$. Within one call, we either verify that $S_O$ is an independent set in $\Oh(n^2)$ time and call $\textsc{A}_\hh$ that runs in $n^{\Oh(1)}$ time, or compute the union of the partial results given by other calls in $n^{\Oh(1)}$ time. Hence the total size of the recursion tree is bounded by~$2^d \cdot |V(F)|$ and the work per call is polynomial in~$n$. As we store a polynomial amount of information in each call, this yields the desired bounds on running time and space requirements.}
%
\end{proof}

In the algorithm above we assumed we are given some $\hh$-elimination forest. One example where $\textsc{Vertex cover}$ is polynomial time solvable, but no $\hh$-elimination forest can be computed in FPT time is the class of perfect graphs, which we prove in Theorem~\ref{thm:inapprox:perfectdepthwidth}. For some graph classes $\hh$ we do know how to compute such decompositions, which gives the following.

\begin{corollary}
\label{thm:vcelimcombine}
Let~$\hh$ be a graph class such that~$\hh \in \{\text{bipartite, chordal, interval}\}$ or~$\hh$ is characterized by a finite family of forbidden induced subgraphs and \textsc{Vertex cover} is polynomial-time solvable on~$\hh$. Then \textsc{Vertex cover} is FPT parameterized by $\hhdepth(G)$. Furthermore it can be solved in $2^{\customdepth{\hh}(G)^{\Oh(1)}}n^{\Oh(1)}$ time using polynomial space. 
\end{corollary}
\begin{proof}
It is well known that \textsc{Vertex cover} is polynomial time solvable on bipartite, chordal, and interval graphs, {as these are subclasses of perfect graphs~\cite{GrotschelLS81}}. Let $k = \hhdepth(G)$. \mic{By \cref{thm:decomposition:full}, we can compute an $\hh$-elimination forest of depth {$k^{\Oh(1)}$} in $2^{k^{\Oh(1)}} \cdot n^{\Oh(1)}$ time
and polynomial space, for each considered class $\hh$}. The result follows by supplying the computed decomposition to Theorem~\ref{thm:vcelim}.
\end{proof}

{One example of a hereditary graph class~$\hh$ that satisfies the premises of the theorem is the class of claw-free (induced $K_{1,3}$-free) graphs. A maximum independent set (and therefore, its complement which forms a minimum vertex cover) can be found in polynomial-time in claw-free graphs~\cite{Minty80,Sbihi80}.}

{Taking the idea of relaxed parameterizations for \textsc{Vertex cover} one step further, one could even combine multiple graph classes in which the problem can efficiently be solved. For example, if~$\hh$ is the class of graphs where each connected component is bipartite, or chordal, or claw-free, then \textsc{Vertex cover} is still FPT parameterized by~$\hhdepth$. It is easy to see that the decomposition algorithms can be adapted to work with the combination of these graph classes, while Theorem~\ref{thm:vcelim} directly applies. This leads to even smaller parameter values. Using the $\hh$-treewidth counterpart (Theorem~\ref{thm:vc:width}), we also obtain fixed-parameter tractability for the parameterization by~$\hhtw$.}

\paragraph{Given a tree $\hh$-decomposition}
Niedermeier~\cite{Niedermeier06} gives an algorithm to solve \textsc{Vertex cover} given a tree decomposition of width $d-1$ in time $2^d \cdot n^{\Oh(1)}$ using a standard dynamic programming approach. We adapt this algorithm to work with $\hh$-decompositions for hereditary graph classes $\hh$ in which vertex cover is polynomial time solvable similar to the $\hh$-elimination forest case. 

\begin{thm} \label{thm:vc:width}
Let $\hh$ be a hereditary graph class on which \textsc{Vertex cover} is {polynomial-time} solvable. \textsc{Vertex cover} can be solved in time $2^k\cdot n^{\Oh(1)}$ when given a tree $\hh$-decomposition of width $k-1$ consisting of~$n^{\Oh(1)}$ nodes.
\end{thm}
\begin{proof}
We adapt the algorithm by Niedermeier~\cite[Thm.~10.14]{Niedermeier06}. Choose an arbitrary root for the given tree $\hh$-decomposition $(T,\chi,L)$. For $t \in V(T)$, let $T_t$ be the subtree rooted at $t$. We keep a table for every $t \in V(T)$, with a row for every assignment $C_t \colon \chi(t) \setminus L \to \{0,1\}$. For every such assignment, we store a value $m_t(C_t)$ that denotes the number of vertices in a~minimum vertex cover $V'$ of the graph induced by $\bigcup_{t' \in T_t}\chi(t')$ such that $V' \cap (\chi(t) \setminus L) = C_t^{-1}(1)$. That is, we consider solutions which contain all vertices in~$C_t^{-1}(1)$ and none of~$C_t^{-1}(0)$. Not every assignment leads to a solution: $m_t(C_t) = \infty$ if there is some edge $uv$ with $C_t(u) = C_t(v) = 0$. 

We only have to adapt the computation for the leaf nodes~$t$ of the decomposition; the update steps remain identical as before~\cite{Niedermeier06} since~$\chi(t) \cap L = \emptyset$ for non-leaf nodes $t \in V(T)$. 

For a leaf node~$t$ and assignment~$C_t$, define the auxiliary graph~$G' := G[(\chi(t) \cap L) \setminus N_G(C_t^{-1}(0))]$ corresponding to the base component without the vertices which are forced to be in the solution to cover edges incident on~$C_t^{-1}(0)$. As the solutions we are asking for contain all of~$C_t^{-1}(1)$, none of~$C_t^{-1}(0)$ and therefore all of~$N_G(C_t^{-1}(0)) \cap (\chi(t) \cap L)$, and furthermore have to cover all edges of~$G'$, the value of~$m_t$ can be computed as follows:
\begin{equation*}
    m_t(C_t) = \begin{cases}
        + \infty & \mbox{if $\exists uv \in E(G) \colon u,v \in C_t^{-1}(0)$} \\
        |C_t^{-1}(1)| + |N_G(C_t^{-1}(0)) \cap \chi(t) \cap L| + \mathsf{minvc}(G') & \mbox{otherwise.}
    \end{cases}
\end{equation*}
Note that since $\hh$ is hereditary and $G[\chi(t) \cap L] \in \hh$, it follows that $G' \in \hh$. Therefore we can compute the size $\mathsf{minvc}(G')$ of a minimum vertex cover in $G'$ in $n^{\Oh(1)}$ time. 
The running time of $2^k \cdot n^{\Oh(1)}$ follows.
\end{proof}

As before, for certain $\hh$ we can compute tree $\hh$-decompositions to arrive at fixed parameter tractability results.

\begin{corollary}\label{thm:vcdpcombine}
Let~$\hh$ be a graph class such that~$\hh \in \{\text{bipartite, chordal, interval}\}$ or~$\hh$ is characterized by a finite family of forbidden induced subgraphs and \textsc{Vertex cover} is polynomial-time solvable on~$\hh$. Then \textsc{Vertex cover} is FPT parameterized by $\hhtw(G)$ and can be solved in $2^{\customtw{\hh}(G)^{\Oh(1)}}n^{\Oh(1)}$ time.
\end{corollary}
\begin{proof}
It is well known that \textsc{Vertex cover} is polynomial time solvable on bipartite, chordal, and interval graphs, {as these are subclasses of perfect graphs~\cite{GrotschelLS81}}. Let $k = \hhtw(G)$.
\mic{By \cref{thm:decomposition:full}, we can compute a tree $\hh$-decomposition of width {$k^{\Oh(1)}$} in time $2^{k^{\Oh(1)}} \cdot n^{\Oh(1)}$, for each considered class $\hh$.}
 The result follows by supplying the computed decomposition to Theorem~\ref{thm:vc:width}.
\end{proof}


\subsubsection{Hitting forbidden cliques}\label{sec:forbiddencliques}
{We say that a graph is $K_\ell$-free if it does not contain $K_\ell$ as a~subgraph, or equivalently, as an~induced subgraph.} In this section we show how to obtain an FPT algorithm for the \textsc{$K_\ell$-free deletion} problem {where one is given} an ($K_\ell$-free)-elimination forest, for some $\ell \geq 2$,
{and the goal is to compute the (size of a) minimum vertex set $S \subseteq V(G)$ such that $G-S$ is $K_\ell$-free.} Here $\ell$ is assumed to be a~constant. Note that the \textsc{Vertex cover} problem is equivalent to \textsc{$K_2$-free deletion}.

The fact that \textsc{$K_\ell$-free deletion} problem has an FPT algorithm parameterized by the depth of an ($K_\ell$-free)-elimination forest will follow {from the more general Theorem~\ref{thm:meta-induced:main} on hitting forbidden induced subgraphs, but} the branching algorithm we present here has the additional feature that it runs in polynomial space. 
{It is} similar to Algorithm~\ref{alg:vc-elim} for \textsc{Vertex cover}, where we branch on the top-most vertex of an elimination forest. As with \textsc{Vertex cover}, a vertex is either in the solution (and added to $S_I$), or not in the solution (and added to $S_O$). The {key idea} is that a clique {can occur only on} a single root-to-leaf path in an elimination forest.
{The base component might have some ancestors in $S_O$ that can still be part of some $K_\ell$. In the case of $\ell = 2$ (\textsc{Vertex cover}), there is only one remaining option to hit each $K_2$, namely to remove each neighbor of $S_O$ in the base component. In general there are can be as many as $\ell-1$ options to consider.}
Since the base component is $K_\ell$-free, by Corollary~\ref{cor:optimal-set-separable-set} it follows that any optimal solution uses at most $d$ vertices from the base component, where $d$ is the depth of the ($K_\ell$-free)-elimination forest. 

Consider the \textsc{$K_\ell$-free deletion with forbidden vertices} (\textsc{$K_\ell$-fdfv}) problem, where we are given a graph $G$, a set of undeletable vertices $S \subseteq V(G)$, and integer $k$, whereas $\ell$ is treated as a constant. The task is to output a minimum set $X \subseteq V(G) \setminus S$ such that $G-X$ has no $K_\ell$ as induced subgraph and $|X| \leq k$, or $\bot$ if no such set exists. We can solve this problem by finding an induced $K_\ell$ in $\Oh(n^\ell)$ time, and then branching on the at most $\ell$ ways to delete a non-forbidden vertex to hit this subgraph. If a solution has not been found after branching for $k$ levels deep, return~$\bot$. Hence we get the following observation; a formal proof follows from the more general algorithm of Lemma~\ref{lem:meta-induced:undeletable} which is presented later.

\begin{observation}
\textsc{$K_\ell$-fdfv}$(G,S,k)$ is solvable in $\ell^k \cdot n^{\Oh(1)}$ time and polynomial space.
\end{observation}

\begin{algorithm}
\caption{\textsc{$K_\ell$-free deletion} (Graph $G$, $t \in V(T)$, $(T,\chi)$, $(S_I,S_O)$)}
\label{alg:klfree-elim}
\begin{algorithmic}[1]
\Input 
\Statex $(T,\chi)$ is a ($K_\ell$-free)-elimination forest $G$ that consists of a single tree. 
\Statex $S = S_I \cup S_O = \{ \chi(t') \mid t' \text{ is a proper ancestor of } t\}$.
\Statex $S_I$ and $S_O$ are disjoint.
\Output
\Statex A $K_\ell$-free deletion set $X_t \subseteq V_t$ of $G[V_t \cup S_O]$ of minimum size, such that $|X_t \cap \chi(t')| \leq \depth_T(t')$ for each leaf $t' \in V(T_t)$. If no such $X_t$ exists return $\bot$.
\If {$t$ is a leaf}
\If{$G[S_O]$ contains $K_\ell$ as induced subgraph}
\State \Return $\bot$
\Else
\State \Return \textsc{$K_\ell$-fdfv}($G[\chi(t) \cup S_O],S_O,\depth_T(t))$
\EndIf
\Else
\State Let $U_t$ be the children of $t$.
\State Let $X_I = \chi(t) \cup \bigcup_{t' \in U_t}\textsc{$K_\ell$-free deletion}(G, t', (T,\chi), (S_I \cup \chi(t),S_O))$.
\State Let $X_O = \bigcup_{t' \in U_t}\textsc{$K_\ell$-free deletion}(G, t', (T,\chi), (S_I,S_O \cup \chi(t)))$.
\State \Return smallest of $X_I$ and $X_O$ that is not $\bot$, or $\bot$ otherwise.
\EndIf
\end{algorithmic}
\end{algorithm}

\begin{thm}\label{thm:Kl-free-elim}
Given a ($K_\ell$-free)-elimination forest $(F,\chi)$ with depth $d$ for some constant $\ell \geq 2$, \textsc{$K_\ell$-free deletion} is solvable in $(2\ell)^d \cdot n^{\Oh(1)}$ time and polynomial space.
\end{thm}
\begin{proof}
We solve \textsc{$K_\ell$-free deletion} by summing the results of \textsc{$K_\ell$-free deletion}$(G[V_{r_T}],r_T,\allowbreak (T,\chi_T),(\emptyset,\emptyset,\emptyset))$ over all trees $T$ in the forest $F$. Here $r_T \in V(T)$ is the root of $T$, and $\chi_T$ is the function $\chi$ restricted to $V(T)$. To see that this correctly solves the problem, note that any optimal solution takes at most $\depth_T(t')$ vertices from $\chi(t')$ for every leaf $t' \in V(T)$ by Corollary~\ref{cor:optimal-set-separable-set}. 
We argue the correctness of Algorithm~\ref{alg:klfree-elim} by induction on $\height(T_t)$. Consider a call \textsc{$K_\ell$-free deletion}$(G,t,(T,\chi),(S_I,S_O))$ that satisfies the input requirements. 

If $\height(T_t) = 0$, then $t$ is a leaf and $G[\chi(t) \cap L]$ is $K_\ell$-free.
If graph $G[S_O]$ is not $K_\ell$-free,
then there is no $X_t \subseteq V_t$ that satisfies the output requirements.
Otherwise, the problem is equivalent to an~instance of \textsc{$K_\ell$-free deletion with forbidden vertices} for arguments $(G[\chi(t) \cup S_O], S_O, \depth_T(t))$.

In the remaining case $\height(T_t) > 0$, so $t$ is not a leaf. Recall that there is an ancestor-descendant relation for every edge in $E(G)$. Therefore any clique occurs on a single root-to-leaf path in the decomposition. We branch on the unique vertex in $\chi(t)$ to be either in the solution or not, and recursively call the algorithm for each child $t'$ of $t$. Since $\height(T_{t'}) \leq \height(T_t) - 1$, these calls are correct by the induction hypothesis. Since a clique occurs on a single root to leaf path, the union of solutions for the recursively called subinstances forms a solution for the current instance. Since this branching is exhaustive, the minimum size solution occurs on one of them if it exists. If all subinstances return~$\bot$, then we can return $\bot$ for the current instance as well.

{Finally, we bound the running time and space requirements, by analyzing the recursion tree generated by the initial calls on each tree~$T$ of the forest~$F$. For each node~$t \in V(T)$, there are exactly~$2^{\depth_T(t)}$ recursive calls to Algorithm~\ref{alg:vc-elim} with~$t$ as the second parameter, one for every partition of the nodes in the bags of ancestors of~$t$ into~$S_I$ and~$S_O$. Within one call, we either verify that $G[S_O]$ does not contain $K_\ell$ as induced subgraph in $n^{\Oh(\ell)} = n^{\Oh(1)}$ time and call $K_\ell$-\textsc{fdfv} that runs in $\ell^{\depth_T(t)} \cdot n^{\Oh(1)}$ time, or compute the union of the partial results given by other calls in $n^{\Oh(1)}$ time. Hence the total size of the recursion tree is bounded by~$2^d \cdot |V(F)|$ and the work per call takes at most~$\ell^d \cdot n^{\Oh(1)}$ time. Note that $K_\ell$-\textsc{fdfv} can be solved using polynomial space. Therefore, as we store a polynomial amount of information in each call, this yields the desired bounds on running time and space requirements.}
\end{proof}

\begin{corollary}\label{thm:klfree-elimcombine}
For each fixed $\ell \geq 2$, \textsc{$K_\ell$-free deletion} is FPT parameterized by $k = \customdepth{\text{$K_\ell$-free}}(G)$. Furthermore it can be solved in $2^{\Oh(k^3 \log k)} \cdot n^{\Oh(1)}$ time using polynomial space.
\end{corollary}
\begin{proof}
Let $k = \customdepth{\text{$K_\ell$-free}}(G)$. \mic{By Theorem~\ref{thm:decomposition:full},} we can compute a ($K_\ell$-free)-elimination forest of depth $\Oh(k^3 \log k)$ {in $2^{\Oh(k)} \cdot n^{\Oh(1)}$ time and} polynomial space. The result follows by supplying this decomposition to Theorem~\ref{thm:Kl-free-elim}.
\end{proof}

\subsection{Generic algorithms via $A$-exhaustive families} \label{sec:alg:generic}
\subsubsection{$A$-exhaustive families and boundaried graphs} \label{sec:repsets}

\mic{In this section we introduce the main tools needed to present our most general framework for dynamic programming
on tree \hh-decompositions. 
We follow the ideas of gluing graphs and finite state property dating back to the results of Fellows and Langston~\cite{Fellows89} (cf. \cite{Arnborg91, Bodlaender96reduction}).


The following definition formalizes a key idea for the dynamic programming routine: to compute a restricted set of partial solutions out of which an optimal solution can always be built.} 

\begin{definition} \label{def:representative}
Let~$\hh$ be a graph class, let~$G$ be a graph, and let~$A \subseteq V(G)$. Then we say that a family~$\mathcal{S} \subseteq 2^A$ of subsets of~$A$ is $A$-exhaustive for \textsc{$\hh$-deletion} on~$G$ if for every minimum-size set~$S \subseteq V(G)$ for which~$G - S \in \hh$, there exists~$S_A \in \mathcal{S}$ such that for~$S' := (S \setminus A) \cup S_A$ we have~$|S'| \leq |S|$ and~$G - S' \in \hh$.
\end{definition}

As a consequence of this definition, if~$\mathcal{S}$ is $A$-exhaustive for \textsc{$\hh$-deletion} then there exists an optimal solution~$S$ to \textsc{$\hh$-deletion} with~$S \cap A \in \mathcal{S}$.

The notion of $A$-exhaustive families is similar in spirit to that of \emph{$q$-representative families} for matroids, which have been used in recent algorithms working on graph decompositions~\cite{FominLPS16,KratschW20} (cf.~\cite[\S 12.3]{CyganFKLMPPS15}) to trim the set of partial solutions stored by a dynamic program while preserving the existence of an optimal solution. As the desired outcome of the replacement process in our case is not defined in terms of independence in a matroid and there is no particular importance of a given integer~$q$, we use different terminology for our concept of exhaustive families.
We reserve the name of a~\emph{representative family} to refer to a set of representatives for each equivalence class in the relation introduced in Definition~\ref{def:boundaried:eqv}.

{The following observation shows how $A$-exhaustive families can be extended to supersets~$A' \supseteq A$ by brute force. If~$|A' \setminus A|$ is bounded, the increase in the size of the exhaustive family can be controlled.}

\begin{observation} \label{obs:representative:introduce}
Let~$\hh$ be a graph class, let~$G$ be a graph, let~$A \subseteq V(G)$ and let~$\mathcal{S} \subseteq 2^A$ be $A$-exhaustive for \textsc{$\hh$-deletion} on~$G$. Then for any set~$A' \supseteq A$, the family~$\mathcal{S'} \subseteq 2^{A'}$ defined as follows has size at most~$|\mathcal{S}| \cdot 2^{|A' \setminus A|}$ and is $A'$-exhaustive for \textsc{$\hh$-deletion} on~$G$:
\begin{equation*}
    \mathcal{S}' := \{ S_1 \cup S^* \mid S_1 \in \mathcal{S} \wedge S^* \subseteq (A' \setminus A)\}.
\end{equation*}
\end{observation}

{A similar brute-force extension can be done when merging exhaustive families for disjoint subsets~$A_1, A_2$ into a an $A$-exhaustive family for a common superset~$A \supseteq A_1 \cup A_2$. We will use this step to handle a variation of standard $\mathsf{join}$ bags in a tree decomposition. 
}

\begin{lemma} \label{lem:representative:join}
Let~$\hh$ be a graph class and let~$G$ be a graph. Let~$A_1,A_2 \subseteq V(G)$ be disjoint sets and let~$\mathcal{S}_1, \mathcal{S}_2$ be $A_1$-exhaustive (respectively, $A_2$-exhaustive) for \textsc{$\hh$-deletion} on~$G$. Then for any set~$A' \supseteq A_1 \cup A_2$, the family~$\mathcal{S'} \subseteq 2^{A'}$ defined as follows has size at most~$|\mathcal{S}_1| \cdot |\mathcal{S}_2| \cdot 2^{|A' \setminus (A_1 \cup A_2)|}$ and is $A'$-exhaustive for \textsc{$\hh$-deletion} on~$G$:
\begin{equation*}
    \mathcal{S}' := \{ S_1 \cup S_2 \cup S^* \mid S_1 \in \mathcal{S}_1 \wedge S_2 \in \mathcal{S}_2 \wedge  S^* \subseteq (A' \setminus (A_1 \cup A_2))\}.
\end{equation*}
\end{lemma}
\begin{proof}
The bound on~$|\mathcal{S}'|$ is clear from the definition. Consider an arbitrary optimal solution~$S \subseteq V(G)$ to~\textsc{$\hh$-deletion} on~$G$; we will show there exists~$\widehat{S} \in \mathcal{S}'$ such that~$(S \setminus A') \cup \widehat{S}$ is an optimal solution. We use a two-step argument.

Since~$\mathcal{S}_1$ is $A_1$-exhaustive, there exists~$S_1 \in \mathcal{S}_1$ such that~$S' := (S \setminus A_1) \cup S_1$ is again an optimal solution.

Applying a similar step to~$S'$, as~$\mathcal{S}_2$ is $A_2$-exhaustive there exists~$S_2 \in \mathcal{S}_2$ such that~$S'' := (S' \setminus A_2) \cup S_2$ is an optimal solution.

Since~$A_1$ and~$A_2$ are disjoint, we have~$S'' \cap A_1 = S_1$ and~$S'' \cap A_2 = S_2$. Let~$S^* := S'' \cap (A' \setminus (A_1 \cup A_2))$. It follows that the set~$\widehat{S} = S_1 \cup S_2 \cup S^*$ belongs to~$\mathcal{S}'$. Now note that~$S \setminus A' = S'' \setminus A'$ as we have only replaced parts of the solution within~$A_1$ and~$A_2$, while~$A' \supseteq A_1 \cup A_2$. Hence~$(S \setminus A') \cup \widehat{S} = S''$ is an optimal solution, which concludes the proof.
\end{proof}

Note that Lemma~\ref{lem:representative:join} implies Observation~\ref{obs:representative:introduce} by letting~$A_2 = \emptyset,\, \mathcal{S}_2 = \{ \emptyset\}$.


\paragraph{\bmp{Boundaried graphs}}
For a function~$f \colon A \to B$ and a set~$A' \subseteq A$, we denote by~$f_{|A'} \colon A' \to B$ the restriction of~$f$ to~$A'$. 
A $k$-boundaried graph is a triple $\widehat{G} = (G,X,\lambda)$, where $G$ is a graph, $X \subseteq V(G)$, and $\lambda \colon [k] \to X$ is a bijection.
\mic{We define $V(\widehat{G}) = V(G)$.}
Two $k$-boundaried graphs $(G_1, X_1, \lambda_1)$, $(G_2, X_2, \lambda_2)$ are compatible if $\lambda_2 \circ \lambda_1^{-1}$ is a graph isomorphism between $G_1[X_1]$ and $G_2[X_2]$.
Two $k$-boundaried graphs $(G_1, X_1, \lambda_1)$, $(G_2, X_2, \lambda_2)$ are isomorphic if there is a graph isomorphism $\pi \colon V(G_1) \to V(G_2)$,
such that $\pi_{|X_1} = \lambda_2 \circ \lambda_1^{-1}$.
The gluing operation $\oplus$ takes two compatible $k$-boundaried graphs  $(G_1, X_1, \lambda_1)$, $(G_2, X_2, \lambda_2)$, creates their disjoint
union, and then identifies the vertices $\lambda_1[i], \lambda_2[i]$ for each $i \in [k]$.
A tri-separation $(A,X,B)$ of order $k$ in graph $G$ can \bmp{be naturally} decomposed as $(G[A \cup X], X, \lambda) \oplus (G[B \cup X], X, \lambda)$ for an~arbitrary bijection $\lambda \colon [|X|] \to X$.

\begin{definition} \label{def:boundaried:eqv}
We say that two $k$-boundaried graphs $(G_1, X_1, \lambda_1)$, $(G_2, X_2, \lambda_2)$ 
are $(\hh,k)$-equivalent if \mic{they are compatible and} for every compatible $k$-boundaried graph $(G_3, X_3, \lambda_3)$, it holds that $(G_1, X_1, \lambda_1) \oplus (G_3, X_3, \lambda_3) \in \hh \Longleftrightarrow (G_2, X_2, \lambda_2) \oplus (G_3, X_3, \lambda_3) \in \hh$.
\end{definition}



\mic{Observe that if $(G_1, X_1, \lambda_1)$ \jjh{and} $(G_2, X_2, \lambda_2)$ are compatible and $G_1, G_2$ do not belong to $\hh$, they are $(\hh,k)$-equivalent because we only consider hereditary classes of graphs.
Therefore in each equivalence class of compatibility
 there can be only one $(\hh,k)$-equivalence class that is comprised of $k$-boundaried \bmp{graphs which} do not belong to $\hh$.}

The relevance of \bmp{the} $(\hh,k)$-equivalence relation for solving $\hh$\textsc{-deletion} can be seen from the observation below.

\begin{observation}\label{obs:boundaried-deletion-set}
Consider $k$-boundaried graphs $(G_1, X_1, \lambda_1)$, $(G_2, X_2, \lambda_2)$, $(H, X_3, \lambda_3)$, such that 
$(G_1, X_1, \lambda_1)$, $(G_2, X_2, \lambda_2)$ are $(\hh,k)$-equivalent {and compatible with  $(H, X_3, \lambda_3)$.}
Let $S \subseteq V(H) \setminus X_3$.
Then $(H - S, X_3, \lambda_3) \oplus (G_1, X_1, \lambda_1) \in \hh$ if and only if $(H - S, X_3, \lambda_3) \oplus (G_2, X_2, \lambda_2) \in \hh$.
\end{observation}

\bmp{The $\hh$\textsc{-membership} problem is the problem of deciding whether a given graph belongs to~$\hh$.} We say that the $\hh$\textsc{-membership} problem is finite state if the relation of $(\hh,k)$-equivalence has finitely many equivalence classes for each $k$.
A $k$-boundaried graph {$(G,X,\lambda)$} is called a \emph{minimal representative} in the relation of $(\hh,k)$-equivalence if {$G \in \hh$ and} there is no other $k$-boundaried graph, which is $(\hh,k)$-equivalent and has a smaller number of vertices.
{We remark that in the literature it is more common to define representatives also for equivalence classes in which the underlying graphs do not belong to $\hh$ but for our purposes such a restriction is sufficient.
It will allow us to (a) avoid non-interesting corner cases and (b) exploit the properties of some classes $\hh$ to bound the number of minimal representatives.}

A family
$\rr^\hh_k$ is called $(\hh,k)$-representative if it {contains a minimal representative from each $(\hh,k)$-equivalence class where the underlying graphs belong $\hh$.} 
It will not matter how the ties are broken.
A family $\rr^\hh_{\le k}$ is called $(\hh,\le k)$-representative if it is a~union of $(\hh,t)$-representative families for all $t \in [k]$.
We define $\texttt{vol}(\rr^\hh_{\le k}) = \sum_{{R} \in \rr^\hh_{\le k}} |V({R})|$.
Note that even though there may be many
ways to construct such a family, the sum above is well-defined,
as well as the size of the family.
Let us note that we can effectively test  $(\hh,k)$-equivalence as long as some $(\hh,k)$-representative family is {provided.} 
\begin{observation}\label{obs:boundaried-testing}
\mic{Let $\rr^\hh_{k}$ be an $(\hh,k)$-representative family.}
Suppose that $k$-boundaried graphs $(G_1, X_1, \lambda_1)$, $(G_2, X_2, \lambda_2)$ are compatible 
and for every compatible $k$-boundaried graph $(G_3, X_3, \lambda_3)$ from $\rr^\hh_k$ it holds that $(G_1, X_1, \lambda_1) \oplus (G_3, X_3, \lambda_3) \in \hh \Longleftrightarrow (G_2, X_2, \lambda_2) \oplus (G_3, X_3, \lambda_3) \in \hh$.
Then $(G_1, X_1, \lambda_1)$, $(G_2, X_2, \lambda_2)$ are $(\hh,k)$-equivalent.
\end{observation}

\mic{
We will be interested in upper bounding the maximal number of vertices of a graph in an $(\hh, \le k)$-representative family $\rr^\hh_{\le k}$;
we will denote this quantity by $r_\hh(k)$.
Since there are $2^{\Oh((r_\hh(k))^2)}$ different graphs on $r_\hh(k)$ vertices,
this gives an immediate bound on $|\rr^\hh_{\le k}|$ and $\texttt{vol}(\rr^\hh_{\le k})$.
What is more, we can construct an $(\hh, \le k)$-representative family within the same running time.

\begin{lemma}\label{lem:representative-generation-general}
Consider a polynomially recognizable graph class $\hh$ such that $\hh$\textsc{-membership} is finite state and there is a time-constructible\footnote{A function $r \colon \mathbb{N} \to \mathbb{N}$ is time-constructible if there exists a Turing machine that, given a string $1^k$, outputs the binary representation of $r(k)$ in time $\Oh(r(k))$. We add this condition so we can assume that the value of $r(k)$ is known to the algorithm. All functions of the form $r(k) = \alpha\cdot k^c$, where $\alpha, c$ are positive integers, are time-constructible.} function $r(k) \ge k$, such that 
for every integer $k$ and for every $t$-boundaried graph, where $t\in [k]$, there exists an $(\hh,t)$-equivalent $t$-boundaried graph on at most $r(k)$ vertices.
Then there exists an algorithm that, given an integer $k$, runs in time $2^{\Oh((r(k))^2)}$
and returns an $(\hh, \le k)$-representative family.
\end{lemma}
\begin{proof}
Consider a process in which we begin from an empty graph on $n$ vertices, fix $t$ boundary vertices, and choose an adjacency matrix determining which pairs of vertices share an edge.
By iterating over all possible adjacency matrices, we can generate all $t$-boundaried graphs on $n$ vertices.

We enumerate all $2^{\Oh((r(k))^2)}$ $t$-boundaried graphs {from $\hh$} on at most $r(k)$ vertices, for all $t \in [k]$.
For each pair of $t$-boundaried graphs which are compatible, we perform gluing and check (in time polynomial in $r(k)$) whether the obtained graph belongs to~$\hh$.
By Observation~\ref{obs:boundaried-testing}, this provides sufficient information to divide the generated $t$-boundaried graphs into $(\hh,t)$-equivalence classes.
We construct \mic{an $(\hh,t)$-representative family} $\rr^\hh_{t}$, for each $t \in [k]$,
by picking any minimal representative from each class; as noted before, it does not matter how the ties are broken.
\end{proof}
}

\subsubsection{Dynamic programming with $A$-exhaustive families}

We move on to designing a meta-algorithm
for solving $\hh$\textsc{-deletion} parameterized by \hhtwfull{}.
Let $S$ be \bmp{an} optimal solution and $t$ be a node
in a \bmp{tree $\hh$-decomposition}.
We will associate some tri-separation $(A,X,B)$ with $t$,
where $A$ stands for the set of vertices introduced below $t$ and $|X| \le k$ for $k-1$ being the width of the decomposition.
The main idea is to consider all potential tri-separations $(A_S,X_S,B_S)$ \bmp{that may} be obtained after removing $S$.
We enumerate all the non-equivalent $k$-boundaried graphs which may be isomorphic with $(G[B_S \cup X_S], X_S)$ and seek a minimal deletion set within $A$ so that the resulting graph belongs to $\hh$.
We will show that it suffices to consider all the representatives in the relation of $(\hh,k)$-equivalence to construct an $A$-exhaustive family of partial solutions.
We shall use the following problem to describe a subroutine \bmp{for} 
finding such deletion sets in the leaves of the decomposition.

\defparproblem{Disjoint \hh-deletion}
{A graph $G$, integers $s, \ell$, and a subset $U \subseteq V(G)$ of size at most $\ell$, which is an $\hh$-deletion set in $G$.}{$s,\ell$}
{Either return a minimum-size $\hh$-deletion set $S \subseteq V(G) \setminus U$ or conclude that no such set of size at most $s$ exists.}

\mic{
We introduce two parameters $s, \ell$ that control the solution size and the size of the set $U$.
In the majority of cases we can adapt known algorithms for \textsc{\hh-deletion} to solve \textsc{Disjoint \hh-deletion} even without the assumptions on the size and structure of $U$.
It is however convenient to impose such requirements on $U$ because (1) this captures the subproblem we need to solve in Lemma~\ref{lem:meta-uniform:undeletable} and (2) this allows us to adapt the known algorithm for $\hh = \mathsf{planar}$ (Theorem~\ref{thm:meta-minors:main}).}

The following lemma shows that the requirement of having undeletable vertices is easy to overcome if $\hh$ is closed under the addition of true twins. We say that vertices $u,v \in V(G)$ are \emph{true-twins} if $N_G[u] = N_G[v]$. 
Note that in the following lemma the value of parameter $\ell$ is insignificant.

\begin{lemma} \label{lem:meta-uniform:undeletable}
Let $\hh$ be a graph class closed under the addition of true twins and such that~\textsc{$\hh$-deletion} \jjh{parameterized by the solution size $s$} can be solved in $f(n,s)$ time.
Then the problem \textsc{Disjoint \hh-deletion} can be solved in $f(s\cdot n, s)$ time.
\end{lemma}
\begin{proof}
Let~$G'$ be the graph obtained from~$G$ by making~$s+1$ new true-twin copies of every vertex in~the undeletable set $U$. Then for every set~$S \subseteq V(G) \setminus U$ for which~$G - S$ is a member of $\hh$, the graph~$G' - S$ is also in $\hh$. Conversely, any minimum-size set~$S \subseteq V(G')$ of size at most~$s$ for which~$G' - S$ in $\hh$ does not contain any copies of vertices in~$U$, as the budget of~$s$ vertices is insufficient to contain all copies of a vertex, while a solution that avoids one copy of a vertex can avoid all copies, since members of $\hh$ are closed under the addition of true twins. Consequently, an optimal solution in~$G'$ of size at most~$s$ is disjoint from~$U$, and is also a solution in the induced subgraph~$G$ of~$G'$. Hence to compute a set~$S$ as desired, it suffices to compute an optimal $\hh$-deletion set in~$G'$ if there is one of size at most~$s$, which can be done by assumption. 
\end{proof}

\begin{lemma}\label{lem:meta-uniform:base}
Suppose that \bmp{$\hh$\textsc{-membership} is finite state},~$\hh$ is hereditary and union-closed, and \textsc{Disjoint \hh-deletion} admits an algorithm with running time $f(s,\ell)\cdot n^{\Oh(1)}$.
Then there is an algorithm that, given a tri-separation~$(A,X,B)$ of order~$k$ in a graph~$G$ such that $G[A] \in \hh$, and \mic{an $(\hh,\le k)$-representative family} $\rr^\hh_{\le k}$, runs in time~$2^k \cdot f(k, r_\hh(k)) \cdot \texttt{vol}(\rr^\hh_{\le k})^{\Oh(1)} \cdot n^{\Oh(1)}$ and outputs a family~$\mathcal{S}$ of size at most~$2^k \cdot |\rr^\hh_{\le k}|$ that is $A$-exhaustive for \hh\textsc{-deletion} on~$G$.
\end{lemma}
\begin{proof}
%
Initialize $\mathcal{S} = \emptyset$.
For each subset $X' \subseteq X$, \bmp{fix an~arbitrary bijection $\lambda \colon [|X'|] \to X'$} and consider the graph $G - (X \setminus X')$. 
\bmp{It} admits a tri-separation $(A, X', B)$.
For each representative ${R} \in \rr^\hh_{t}$, where $t = |X'|$, which is compatible with $(G[B \cup X'], X', \lambda)$,
we perform the gluing operation
$G_R = (G[A \cup X'], X', \lambda) \oplus {R}$
and execute the algorithm for \textsc{Disjoint \hh-deletion} on $G_R$ with the set of undeletable vertices $U = V(G_R) \setminus A$ and parameters $(k, r_\hh(k))$.
In other words, we seek 
a~minimum-size deletion set $A' \subseteq A$ of size at most $k$.
If such a set is found, we add it to $\mathcal{S}$.
Note that $U$ is indeed an $\hh$-deletion set in $G_R$
because $G_R - U = G[A]$ which belongs to $\hh$ by the assumption.
By the definition of an $(\hh,\le k)$-representative family we have that $|U| \le |V(R)| \le r_\hh(k)$, so the created instance meets the specification of \textsc{Disjoint \hh-deletion}.

The constructed family clearly has size at most $2^k \cdot |\rr^\hh_{\le k}|$.
The running time can be upper bounded by $2^k \cdot \sum_{{R} \in \rr^\hh_{\le k}} f(k,r_\hh(k)) (n+|V({R})|)^{\Oh(1)} = 2^k \cdot f(k,r_\hh(k)) \cdot \texttt{vol}(\rr^\hh_{\le k})^{\Oh(1)} \cdot n^{\Oh(1)}$.
It remains to show that $\mathcal{S}$ is indeed $A$-exhaustive.

Consider a minimum-size solution~$S$ to \textsc{$\hh$-deletion} on~$G$. Define~$S_A, S_X, S_B$ as~$S \cap A, S \cap X, S \cap B$, respectively, and let~$X' := X \setminus S_X$, $|X'| = t$.
\mic{Fix an~arbitrary bijection $\lambda \colon [t] \to X'$.}
Since $\hh$ is union-closed and $G[A] \in \hh$, by Lemma~\ref{lem:optimal-set-separable-set} we know that $|S_A| \le k$.
{The graph $G[B \cup X'] - S_B$ is an induced subgraph of $G-S$ so it belongs to $\hh$.}
The set $\rr^\hh_{t}$ contains a $t$-boundaried graph $R$ that is $(\hh,t)$-equivalent to $(G[B \cup X'] - S_B, X',\lambda)$.
By Observation~\ref{obs:boundaried-deletion-set},
$S_A$ is an $\hh$-deletion set for $(G[A \cup X'], X', \lambda) \oplus R$.
As $(G[A \cup X'], X', \lambda) \oplus R$ contains an $\hh$-deletion set within $A$ of size at most $k$, some set $S'_A$ with this \bmp{property} has been added to $\mathcal{S}$.
Furthermore, $S'_A$ is a minimum-size solution, so $|S'_A| \le |S_A|$.
Again by Observation~\ref{obs:boundaried-deletion-set}, $S'_A$ is an $\hh$-deletion set for $(G[A \cup X'], X', \lambda) \oplus (G[B \cup X'] - S_B, X',\lambda)$.
It means that $S' = (S \setminus A) \cup S'_A = S_B \cup S_X \cup S'_A$ is an $\hh$-deletion set in $G$ and $|S'| \le |S|$, which finishes the proof.
\end{proof}

The next step is to propagate the partial solutions along the decomposition in a bottom-up manner.
As we want to grow the sets $A$ for which $A$-exhaustive families are computed, we
can take advantage of Observation~\ref{obs:representative:introduce} and Lemma~\ref{lem:representative:join}.
However, after processing several nodes, the size of $A$-exhaustive families computed this way can become arbitrarily large.
In order to circumvent this, we shall prune the $A$-exhaustive families after each step.

\begin{lemma}\label{lem:meta-uniform:pruning}
Suppose that \bmp{$\hh$\textsc{-membership}} is finite state and graphs in the class $\hh$ can be recognized in polynomial time.
There is an algorithm that, given a tri-separation~$(A,X,B)$ of order~$k$ in a graph~$G$, a family~$\mathcal{S}' \subseteq 2^A$ that is $A$-exhaustive for \hh-\textsc{deletion} on~$G$, and \mic{an $(\hh,\le k)$-representative family} $\rr^\hh_{\le k}$, runs in time~$2^k \cdot |\mathcal{S}'| \cdot \texttt{vol}(\rr^\hh_{\le k})^{\Oh(1)} \cdot n^{\Oh(1)}$ and outputs a family~$\mathcal{S} \subseteq \mathcal{S}'$ of size at most~$2^k \cdot |\rr^\hh_{\le k}|$ that is $A$-exhaustive for \hh-\textsc{deletion} on~$G$.
\end{lemma}
\begin{proof}
Initialize $\mathcal{S} = \emptyset$.
For each subset $X' \subseteq X$, \bmp{fix an~arbitrary bijection $\lambda \colon [|X'|] \to X'$} and consider the graph $G - (X \setminus X')$. 
\bmp{It} admits a tri-separation $(A, X', B)$.
For each graph $R \in \rr^\hh_{t}$, where $t = |X'|$, which is compatible with $(G[B \cup X'], X', \lambda)$,
we perform the gluing operation
\bmp{$G_R = (G[A \cup X'], X', \lambda) \oplus R$}. 
\bmp{Using the polynomial-time recognition algorithm} we choose \bmp{a} minimum-size set $S_A \in \mathcal{S'}$ which is an $\hh$-deletion set for $G_R$, if one exists, and add it to $\mathcal{S}$.

We construct at most $2^k \cdot |\rr^\hh_{\le k}|$ graphs~$G_R$. \bmp{For each graph~$G_R$ we add at most one set to $\mathcal{S}$ and spend~$|\mathcal{S'}| \cdot (n + |R|)^{\Oh(1)}$ time. In total, we perform at most $2^k \cdot |\rr^\hh_{\le k}| \cdot |\mathcal{S'}| \cdot \sum_{{R} \in \rr^\hh_{\le k}} (n+|V({R})|)^{\Oh(1)} = 2^k \cdot |\mathcal{S'}| \cdot \texttt{vol}(\rr^\hh_{\le k})^{\Oh(1)} \cdot n^{\Oh(1)}$ operations.}
It remains to show that $\mathcal{S}$ is indeed $A$-exhaustive.

Consider a minimum-size solution~$S$ to \textsc{$\hh$-deletion} on~$G$. Define~$S_A, S_X, S_B$ as~$S \cap A, S \cap X, S \cap B$, respectively, and let~$X' := X \setminus S_X$, $|X'| = t$. Since $\mathcal{S'}$ is $A$-exhaustive on~$G$, there exists~$\widehat{S}_A \in \mathcal{S}'$ such that~$\widehat{S} = (S \setminus A) \cup \widehat{S}_A$ is also a minimum-size solution to~$G$, implying that~$|\widehat{S}_A| \leq |S_A|$. 
\mic{Fix an~arbitrary bijection $\lambda \colon [t] \to X'$.}
The set $\rr^\hh_{t}$ contains a $t$-boundaried graph $R$ that is $(\hh,t)$-equivalent to $(G[B \cup X'] - S_B, X',\lambda)$.
By Observation~\ref{obs:boundaried-deletion-set}, a set $A' \subseteq A$ is an $\hh$-deletion set for $G' = (G[A \cup X'], X',\lambda) \oplus (G[B \cup X'] - S_B, X',\lambda)$ if and only if $A'$ is 
an $\hh$-deletion set for $(G[A \cup X'], X',\lambda) \oplus R$.
Since~$G' - \widehat{S}_A = G - \widehat{S} \in \hh$, we know that $\widehat{S}_A \in \mathcal{S'}$ is such a set. Hence by the construction above, there \bmp{exists some (possibly different)} $S'_A \in \mathcal{S}$ with this property and minimum size;
hence $|S'_A| \le |\widehat{S}_A| \le |S_A|$ and
$S' = (S \setminus A) \cup S'_A = S_B \cup S_X \cup S'_A$ is an $\hh$-deletion set in $G$ and $|S'| \le |S|$. The claim follows.
\end{proof}

We are ready to combine the presented subroutines in a general meta-algorithm to process tree $\hh$-decompositions for every class $\hh$ which satisfies three simple conditions. The following statement formalizes and, \mic{together with \cref{lem:representative-generation-general}}, proves Theorem~\ref{thm:metathm:informal} from Section~\ref{sec:outline:deletion}.

\begin{thm}\label{thm:meta-uniform:main}
Suppose that the class $\hh$ satisfies the following:
\begin{enumerate}
    \item $\hh$ is hereditary and union-closed,
    \item \textsc{Disjoint \hh-deletion} admits an algorithm with running time $f(s,\ell)\cdot n^{\Oh(1)}$, \label{item:meta-uniform:undeletable}
    \item $\hh$\textsc{-membership} is finite state and there is an algorithm computing an $(\hh,\le t)$-representative family with running time $v(t)$.
    
\end{enumerate}
Then \textsc{$\hh$-deletion} can be solved in time~$2^{\Oh(k)} \cdot f(k, r_\hh(k)) \cdot v(k)^{\Oh(1)} \cdot n^{\Oh(1)}$ when given a~tree $\hh$-decomposition of width~$k-1$ consisting of~$n^{\Oh(1)}$ nodes.
\end{thm}
\begin{proof}

\mic{First, we construct \mic{an $(\hh,\le k)$-representative family} $\rr^\hh_{\le k}$ in time $v(k)$.
Since the output size of the algorithm cannot exceed its running time, we have $|\rr^\hh_{\le k}| \le \texttt{vol}(\rr^\hh_{\le k}) \le v(k)$.
}

The algorithm is based on a variant of dynamic programming in which bounded-size sets of partial solutions are computed, with the guarantee that at least one of the partial solutions which are stored can be completed to an optimal solution. More formally, for each node~$t \in V(T)$ we are going to compute (recall~$\kappa, \pi$ from Definition~\ref{def:kappa}) a set of partial solutions~$\mathcal{S}_t \subseteq 2^{\kappa(t)}$ of size at most~$2^k \cdot |\rr^\hh_{\le k}|$ which is~$\kappa(t)$-exhaustive for \textsc{$\hh$-deletion} in~$G$. As~$\kappa(r) = V(G)$ for the root node~$r$ by Observation~\ref{obs:kappa}, any minimum-size set~$S \in \mathcal{S}_r$ for which~$G - S \in \hh$ is an optimal solution to the problem, and the property of $\kappa(r)$-exhaustive families guarantees that one exists.
We do the computation bottom-up in the tree decomposition, using Lemma~\ref{lem:meta-uniform:pruning} to prune sets of partial solutions at intermediate steps to prevent them from becoming too large.

Let~$(T,\chi,L)$ be the given tree $\mathsf{\hh}$-decomposition of width~$k-1$. By Lemma~\ref{lemma:makenice} we may assume that the decomposition is nice and is rooted at some node~$r$.
For~$t \in V(T)$, define~$L_t := L \cap \chi(t)$. Process the nodes of~$T$ from bottom to top. We process a node~$t$ after having computed exhaustive families for all its children, as follows. Let~$X_t := \chi(t) \cap \pi(t)$, let~$A_t := \kappa(t)$ and let~$B_t := V(G) \setminus (A_t \cup X_t)$. By Observation~\ref{obs:triseparation:from:td}, the partition~$(A_t,X_t,B_t)$ is a tri-separation of~$G$. The way in which we continue processing~$t$ depends on the number of children it has. As~$T$ is a nice decomposition, node~$t$ has at most two children.

\textbf{Leaf nodes} For a leaf node~$t \in V(T)$, we construct an exhaustive family of partial solutions~$\mathcal{S}_t \subseteq 2^{\kappa(t)}$ as follows.
By Definition~\ref{def:tree:h:decomp}, vertices of~$L_t$ do not occur {in other bags than~$\chi(t)$}.
Because the decomposition is nice, we have $\chi(t) \setminus L_t = \pi(t)$.
Therefore $\kappa(t) = L_t$ and we have ~$(A_t,X_t,B_t) = (L_t, \chi(t) \setminus L_t, V(G) \setminus \chi(t))$.
Furthermore, $|X_t| \leq k$ since the width of the decomposition is~$k-1$.
As $G[L_t] \in \hh$, we can process the tri-separation $(A_t,X_t,B_t)$ with Lemma~\ref{lem:meta-uniform:base}
within running time $2^k \cdot f(k, r_\hh(k)) \cdot v(k)^{\Oh(1)} \cdot n^{\Oh(1)}$. 
We obtain a $\kappa(t)$-exhaustive family of size at most $2^k \cdot |\rr^\hh_{\le k}|$.

\textbf{Nodes with a unique child} Let~$t$ be a node that has a unique child~$c$, for which a $\kappa(c)$-exhaustive family~$\mathcal{S}_c$ of size~$2^k \cdot |\rr^\hh_{\le k}|$ has already been computed. Recall that vertices of~$L$ only occur in leaf bags, so that~$L_t = \emptyset$ and therefore~$|\chi(t)| \leq k$. Observe that~$\kappa(t) \setminus \kappa(c) \subseteq \chi(t)$, so that~$|\kappa(t) \setminus \kappa(c)| \leq k$. \bmp({A tighter bound is possible by exploiting the niceness property, which we avoid for ease of presentation.}) Compute the following set of partial solutions:
 \begin{equation*}
     \mathcal{S}'_t := \{ S_c \cup S^* \mid S_c \in \mathcal{S}_c, S^* \subseteq \kappa(t) \setminus \kappa(c) \}.
 \end{equation*}
Since the number of choices for~$S_c$ is~$2^k \cdot |\rr^\hh_{\le k}|$, while the number of choices for~$S^*$ is~$2^{k}$, the set~$\mathcal{S}'_t$ has size at most~$2^{2k} \cdot |\rr^\hh_{\le k}|$ and can be computed in time~$2^{2k} \cdot |\rr^\hh_{\le k}| \cdot n^{\Oh(1)}$. Since~$\kappa(c) \subseteq \kappa(t)$ due to Observation~\ref{obs:kappa}, we can invoke Observation~\ref{obs:representative:introduce} to deduce that the family~$\mathcal{S}'_t$ is $\kappa(t)$-exhaustive for \textsc{$\hh$-deletion} on~$G$. As the last step for the computation of this node, we compute the desired exhaustive family~$\mathcal{S}_t$ as the result of applying Lemma~\ref{lem:meta-uniform:pruning} to~$\mathcal{S}'_t$ and the tri-separation~$(A_t,X_t,B_t)$ of~$G$, which is done in time~$2^{3k} \cdot v(k)^{\Oh(1)} \cdot n^{\Oh(1)}$ because $|\rr^\hh_{\le k}| \le v(k)$. 
As~$A_t = \kappa(t)$, the lemma guarantees that~$\mathcal{S}_t$ is $\kappa(t)$-exhaustive and it is sufficiently small.

\textbf{Nodes with two children}
The last type of nodes to handle are those with exactly two children. So let~$t \in V(T)$ have two children~$c_1, c_2$. Since~$t$ is not a leaf we have~$L_t = \emptyset$. Let~$K := \kappa(t) \setminus (\kappa(c_1) \cup \kappa(c_2))$ and observe that~$K \subseteq \chi(t) \setminus L$. Therefore~$|K| \leq k$. 

Using the $\kappa(c_1)$-exhaustive set~$\mathcal{S}_{c_1}$ and the~$\kappa(c_2)$-exhaustive set~$\mathcal{S}_{c_2}$ computed earlier in the bottom-up process, we define a set~$\mathcal{S}'_t$ as follows:
\begin{equation*}
    \mathcal{S}'_t := \{ S_1 \cup S_2 \cup S^* \mid S_1 \in \mathcal{S}_{c_1}, S_2 \in \mathcal{S}_{c_2}, S^* \subseteq K \}.
\end{equation*}
As~$\mathcal{S}_{c_1}$ and~$\mathcal{S}_{c_2}$ both have size~$2^k \cdot |\rr^\hh_{\le k}|$, while~$|K| \leq 2^{k}$, we have~$|\mathcal{S}'_t| = 2^{3k} \cdot |\rr^\hh_{\le k}|^2$.
{By Observation~\ref{obs:kappa} we have that $\kappa(c_1) \cap \kappa(c_2) = \emptyset$ and $\kappa(c_1) \cup \kappa(c_2) \subseteq \kappa(t)$, so we can apply Lemma~\ref{lem:representative:join} to obtain that} the family~$\mathcal{S}'_t$ is $\kappa(t)$-exhaustive for \textsc{$\hh$-deletion} on~$G$. The desired exhaustive family~$\mathcal{S}_t$ is obtained by applying Lemma~\ref{lem:meta-uniform:pruning} to~$\mathcal{S}'_t$ and the tri-separation~$(A_t,X_t,B_t)$ of~$G$,
which is done in time $2^{4k} \cdot v(k)^{\Oh(1)} \cdot n^{\Oh(1)}$

\textbf{Wrapping up} Using the steps described above we can compute, for each node of~$t \in V(T)$ in a bottom-up fashion, a $\kappa(t)$-exhaustive family~$\mathcal{S}_t$ of size~$2^k \cdot |\rr^\hh_{\le k}|$. Since the number of nodes of~$t$ is~$n^{\Oh(1)}$ 
\mic{the overall running time follows.} As discussed in the beginning of the proof, an optimal solution can be found by taking any minimum-size solution from the family~$\mathcal{S}_r$ for the root~$r$.
\end{proof}

\subsubsection{Hitting forbidden connected minors}

As a first application of the meta-theorem, we consider classes defined by a finite set of forbidden connected minors.
The seminal results of Robertson and Seymour~\cite{robertson2004wagner} \bmp{state} that every minor-closed family $\hh$ can be defined by a finite set of forbidden minors.
The \hh\textsc{-deletion} problem is FPT for such classes when parameterized by the solution size~\cite{AdlerGK08, sau20apices} or by treewidth~\cite{baste20hitting}.
The requirement that all the forbidden minors are connected holds whenever $\hh$ is union-closed.
In Section~\ref{subsec:not:closed} we argue why this limitation is necessary.

Unlike the next sections, here we do not need to prove any claims about structure of minor-closed classes and we can just take advantage of known results in a black-box manner.
In order to apply Theorem~\ref{thm:meta-uniform:main},
we first need to bound the sizes of representatives in $\rr^\hh_{k}$.
To this end, we shall take advantage of the recent result of
Baste,  Sau, and Thilikos \cite{baste20complexity},
who have studied optimal running \bmp{times} for \hh\textsc{-deletion} parameterized by treewidth. 
They define a relation of $(\le h,k)$-equivalence: two $k$-boundaried graphs $(G_1, X_1, \lambda_1)$, $(G_2, X_2, \lambda_2)$ are $(\le h,k)$-equivalent if {they are compatible and for every graph $F$ on at most $h$ vertices and for every compatible} $k$-boundaried graph $(G_3, X_3, \lambda_3)$, $F$ is a minor of $(G_1, X_1, \lambda_1) \oplus (G_3, X_3, \lambda_3)$ if and only if $F$ is a minor of $(G_2, X_2, \lambda_2) \oplus (G_3, X_3, \lambda_3)$. 

\begin{thm}[{\cite[{Thm.~6.2}]{baste20complexity}}] 
\label{thm:meta-minors:minor-representatives}
{There is a computable function $f$, so that
if $(R, X, \lambda)$ is \bmp{a} $k$-boundaried graph and $R$ is $K_q$-minor-free, then
there exists a $k$-boundaried graph $(R', X', \lambda')$ which is $(\le h,k)$-equivalent to $(R, X, \lambda)$ and
$|V(R')| \le f(q,h) \cdot k$.}
\end{thm}

Let $\hh$ be a class defined by a family of forbidden minors, which are all connected and have at most $h$ vertices.
Then whenever two $k$-boundaried graphs are $(\le h,k)$-equivalent, they are also $(\hh,k)$-equivalent.
{As we consider only representatives whose underlying graphs belong to $\hh$, they must exclude $K_h$ as a minor.}
This leads to the following corollary.


\begin{corollary}\label{cor:meta-minors:representatives}
Let $\hh$ be a class defined by a finite family of forbidden minors.
\mic{
There exists a constant $d_\hh$ such that
for every minimal representative $R$ in the relation of $(\hh,k)$-equivalence we have
$|V(R)| \le d_\hh \cdot k$.}
\end{corollary}

\begin{lemma}\label{lem:representative-generation-minor}
Let $\hh$ be a class defined by a finite non-empty family of forbidden minors.
There exists an algorithm that, given an integer $k$, runs in time $2^{\Oh(k\log k)}$
and returns an $(\hh, \le k)$-representative family.
\end{lemma}
\begin{proof}
We take advantage of the fact that graphs in $\hh$ are sparse, that is, there exists a constant $c_\hh$ such that if $G \in \hh$ then $|E(G)| \le c_\hh \cdot |V(G)|$~\cite{Mader1967HomomorphieeigenschaftenUM}.
Any graph $G \in \hh$ on $n$ vertices can be represented by a set of at most $c_\hh \cdot n$ pairs of vertices which share an edge.
Therefore the number of such graphs is  $2^{\Oh(n\log n)}$.

We proceed similarly as in \cref{lem:representative-generation-general} by generating all $2^{\Oh(k\log k)}$ $t$-boundaried graphs on at most $d_\hh \cdot k$ vertices (see \cref{cor:meta-minors:representatives}), whose underlying graphs belong to $\hh$, for all $t \in [k]$.
For each pair of $t$-boundaried graphs which are compatible, we perform gluing and check  whether the obtained graph belongs to~$\hh$.
Then
by Observation~\ref{obs:boundaried-testing} it suffices to pick any minimal representative from each computed equivalence class.
\end{proof}

Furthermore, it turns out that the known algorithms for \hh\textsc{-deletion} parameterized by the solution size can be adapted to work with undeletable vertices.

\begin{thm}[{\cite{sau20apices}\footnote{The details can be found in the full version of the article \cite[Section 7.2]{sau20apices-full}.}}]
\label{lem:meta-minors:undeletable}
Let $\hh$ be a class defined by a finite family of forbidden minors.
Then \textsc{Disjoint $\hh$-deletion} admits an algorithm with running time $2^{s^{\Oh(1)}}\cdot n^3$, where $s$ is the solution size.
\end{thm}

In order to obtain a better final guarantee for the most important case $\hh = \mathsf{planar}$, we need a~concrete bound on the exponent in the running time for \textsc{Disjoint planar deletion}.
An~algorithm with this property was proposed by  Jansen, Lokshtanov, and Saurabh~\cite{JansenLS14} as a subroutine in the iterative compression step for solving \textsc{Planar deletion}. 
This is the only place where we rely on the assumption that that the undeletable set $U$ is
an $\hh$-deletion set of bounded size.

\begin{thm}[{\cite{JansenLS14}}]
\label{lem:meta-planar:undeletable}
\textsc{Disjoint planar deletion} admits an algorithm with running time $2^{\Oh((\ell+s) \log (\ell+s))}\cdot n$,
where $s$ is the solution size and $\ell$ is the size of the undeletable set.
\end{thm}

We are ready to combine all the ingredients and apply the meta-theorem.

\begin{thm}\label{thm:meta-minors:main}
Let $\hh$ be a class defined by a finite family of forbidden connected minors.
Then \textsc{$\hh$-deletion} can be solved in time~$2^{k^{\Oh(1)}} \cdot n^{\Oh(1)}$ when given a~tree $\hh$-decomposition of width~$k-1$ consisting of~$n^{\Oh(1)}$ nodes.
In the special case of $\hh = \mathsf{planar}$ the running time is $2^{\Oh(k \log k)} \cdot n^{\Oh(1)}$.
\end{thm}
\begin{proof}
We check the conditions of Theorem~\ref{thm:meta-uniform:main}.
The class $\hh$ is hereditary and union-closed because the forbidden minors are connected.
By \cref{lem:meta-minors:undeletable}, \textsc{Disjoint $\hh$-deletion} can be solved in time $2^{s^{\Oh(1)}}\cdot n^3$
and, by \cref{lem:representative-generation-minor},
there is an algorithm computing an $(\hh, \le k)$-representative family in time $v(k) = 2^{\Oh(k \log k)}$.

For the case $\hh = \mathsf{planar}$ we additionally take advantage of Theorem~\ref{lem:meta-planar:undeletable} to solve \textsc{Disjoint planar deletion} in time $f(s,\ell)\cdot n$ where $f(s,\ell) = 2^{\Oh((\ell+s) \log (\ell+s))}$.
By
Corollary~\ref{cor:meta-minors:representatives}
we can bound $r_\hh(k)$ by $\Oh(k)$.
Hence, the running time in Theorem~\ref{thm:meta-uniform:main} becomes $2^{\Oh(k)} \cdot f(k, r_\hh(k)) \cdot 2^{\Oh(k \log k)} \cdot n^{\Oh(1)} = 2^{\Oh(k \log k)} \cdot n^{\Oh(1)}$. 
\end{proof}

Finally, we invoke the algorithm for computing a tree $\hh$-decomposition of approximate width to infer the general tractability result.

\begin{corollary}\label{thm:meta-minors:final}
Let $\hh$ be a class defined by a finite family of forbidden connected minors.
Then \hh\textsc{-deletion} can be solved in time $2^{k^{\Oh(1)}} \cdot n^{\Oh(1)}$ where $k = \mathbf{tw}_\hh(G)$.
In the special case of $\hh = \mathsf{planar}$ the running time is $2^{\Oh(k^{5} \log k)} \cdot n^{\Oh(1)}$.
\end{corollary}
\begin{proof}
{By Theorem~\ref{thm:decomposition:full} we can find a tree $\hh$-decomposition of width $\Oh\big((\mathbf{tw}_\hh(G))^5\big)$
in time $2^{k^{\Oh(1)}} \cdot n^{\Oh(1)}$.
For $\hh = \mathsf{planar}$ the running time for constructing the decomposition is $2^{\Oh(k^{2}\log k)} \cdot n^{\Oh(1)}$.
It remains to apply \cref{thm:meta-minors:main} with parameter $k' = \Oh\big((\mathbf{tw}_\hh(G))^5\big)$.}
\end{proof}

\subsubsection{Hitting forbidden connected induced subgraphs} \label{sec:hitting:subgraphs}
In this section we deal with graph classes~$\hh$ defined by a finite family~$\mathcal{F}$ of forbidden induced subgraphs.
The problem of hitting finite forbidden (induced) subgraphs, parameterized by treewidth, has been studied by several authors~\cite{CyganMPP17,Pilipczuk11,SauS20}.

We introduce some terminology for working with induced subgraphs and isomorphisms. For graphs~$G$ and~$H$, a function~$f \colon V(H) \to V(G)$ is an \emph{induced subgraph isomorphism from~$H$ to~$G$} if~$f$ is injective and satisfies~$xy \in E(H) \Leftrightarrow f(x)f(y) \in E(G)$ for all~$x,y \in V(H)$.
A function~$f' \colon A \to B$ is an \emph{extension} of a function~$f \colon A' \to B$ if~$f'_{|A'} = f$. 

\begin{definition}
Let~$\mathcal{F}$ be a finite family of graphs. The operation of \emph{$\mathcal{F}$-pruning} a $k$-boundaried graph~$(G, X, \lambda)$ is defined as follows:
\begin{itemize}
    \item Initialize all vertices of~$V(G) \setminus X$ as unmarked.
    \item For each~$F \in \mathcal{F}$, for each tri-separation~$(A_F,X_F,B_F)$ of~$H$ with~$B_F \neq \emptyset$, for each induced subgraph isomorphism~$f$ from~$H[X_F]$ to~$G[X_G]$, if there exists an induced subgraph isomorphism from~$H[X_F \cup B_F]$ to~$G$ that is an extension of~$f$, then mark the vertex set~$\{f'(b) \mid b \in B_F\}$ for one such extension~$f'$, chosen arbitrarily.
    \item Remove all vertices of~$V(G) \setminus X$ which are not marked at the end of the process.
\end{itemize}
\end{definition}

Observe that graph resulting from the operation of \emph{$\mathcal{F}$-pruning} depends on the choices made for~$f'$. Our statements and algorithms are valid regardless how these ties are broken
as long as the graphs produced by $\mathcal{F}$-pruning two isomorphic graphs are also isomorphic.
We remark that there is no need to consider implementation aspects of {$\mathcal{F}$-pruning} because we only use it for an existential bound on the sizes of representatives.


\begin{lemma} \label{lem:meta-induced:equivalent}
Let~$\hh$ be a class defined by a finite family $\mathcal{F}$ of forbidden induced subgraphs.
Let~$(G_1, X, \lambda)$ be a $k$-boundaried graph and 
suppose that~$(G_2, X,\lambda)$ was obtained by $\mathcal{F}$-pruning~$(G_1, X,\lambda)$.
Then $(G_1, X,\lambda)$ and $(G_2, X,\lambda)$ are $(\hh,k)$-equivalent.
\end{lemma}
\begin{proof}
First observe that these graphs are compatible because $\mathcal{F}$-pruning removes a subset of vertices from $V(G_1) \setminus X$.
As~$G_2$ is an induced subgraph of~$G_1$, the \mic{forward implication of Definition~\ref{def:boundaried:eqv}} is trivial. 
We prove that \mic{for any} compatible $k$-boundaried graph $\widehat{H}$ it holds that \bmp{if} $\widehat{H} \oplus (G_2, X, \lambda)$ is induced-$\mathcal{F}$-free, then also $\widehat{H} \oplus (G_1, X, \lambda)$ is induced-$\mathcal{F}$-free.

Assume for a contradiction that $G = \widehat{H} \oplus (G_1, X, \lambda)$ contains an induced subgraph isomorphic to~$F$ for some~$F \in \mathcal{F}$.
Let $(A, X, B)$ be the tri-separation of $G$, so that $(G[A \cup X],X,\lambda)$ is isomorphic with $\widehat{H}$ and $(G[B \cup X],X,\lambda)$ is isomorphic with $(G_1,X,\lambda)$.
Let $B' \subseteq B$ be the set of vertices marked during $\mathcal{F}$-pruning $(G[B \cup X], X, \lambda)$.
Let~$f$ 
be an induced subgraph isomorphism from~$F$ to~$G$. We define a tri-separation of~$F$ based on~$f$: let~$A_F := \{v \in V(F) \mid f(v) \in A\}$, let~$X_F := \{v \in V(F) \mid f(v) \in X\}$, and let~$B_F := \{v \in V(F) \mid f(v) \in B\}$. 
\mic{Observe that if $B_F = \emptyset$, then the image of $F$ is fully contained in $A \cup X$, and so \bmp{$G[A \cup X \cup B'] = \widehat{H} \oplus (G_2, X, \lambda)$} contains $F$ as an induced subgraph, which gives a~contradiction.
Assume from now \bmp{on} that $B_F \ne \emptyset$.}

Note that~$f_{|X_F}$ is an induced subgraph isomorphism from~$F[X_F]$ to~$G[X]$ \mic{(if $X_F = \emptyset$ this is an empty isomorphism, which is also considered during $\mathcal{F}$-pruning)}, and that~$f_{|X_F \cup B_F}$ is an extension of~$f_{|X_F}$ that forms an induced subgraph isomorphism from~$F[X_F \cup B_F]$ to~$G[X \cup B]$. Consequently, in the process of $\mathcal{F}$-pruning $(G[B \cup X],X,\lambda)$ we considered~$F$, the function~$f_{|X_F}$, and an~extension~$f'$ of~$f_{|X_F}$ that is an induced subgraph isomorphism from~$F[X_F \cup B_F]$ to~$G[X \cup B]$. Hence the vertices~$\{f'(b) \mid b \in B_F\}$ were marked during the pruning process and are preserved in~$B'$. It follows that~$f'$ is also an induced subgraph isomorphism from~$F[X_F \cup B_F]$ to~${G}[X \cup B']$. Now consider the function~$f^* \colon V(F) \to V(G)$ such that~$f^*_{|A_F \cup X_F} = f$ and~$f^*_{|B_F} = f'$, and recall that~$f_{|X_F} = f'_{|X_F}$. As the tri-separation of~$G$ ensures that~$f(x)f(y) \notin E(G)$ for any~$x \in A_F$ and~$y \in B_F$, it can easily be verified that~$f^*$ is an induced subgraph isomorphism from~$F$ to~${G}[A \cup X \cup B']$.
This graph is isomorphic with $\widehat{H} \oplus (G_2, X, \lambda)$,
so it also contains an induced subgraph isomorphic to~$F$: a~contradiction to the assumption that~$\widehat{H} \oplus (G_2, X, \lambda)$ is induced-$\mathcal{F}$-free.
\end{proof}

Since $\mathcal{F}$-pruning preserves the $(\hh,k)$-equivalence class, we can assume that the minimal representatives \bmp{cannot be reduced by $\mathcal{F}$-pruning}. 
It then suffices to estimate the maximal number of vertices left after $\mathcal{F}$-pruning.
For a graph~$F$ on~$c$ vertices, there are at most~$3^c$ tri-separations of~$F$ as each vertex either belongs to~$A_F, X_F$, or~$B_F$. The number of induced subgraph isomorphisms from~$X_F$ to~$X_G$ is bounded by~$k^c$, where~$k = |X|$, as for each of the at most~$c$ vertices in~$X_F$ there are at most~$k$ options for their image. For each such induced subgraph isomorphism we mark at most~$c$ vertices.

\begin{observation}\label{obs:meta-induced:pruning:size}
Let~$\mathcal{F}$ be a finite family of graphs on at most~$c$ vertices each. The operation of~$\mathcal{F}$-pruning a $k$-boundaried graph~$(G, X, \lambda)$ removes all but~$|\mathcal{F}| \cdot 3^c \cdot k^c \cdot c$ vertices from~$V(G) \setminus X$.
\end{observation}

\begin{corollary}\label{cor:meta-induced:representatives}
Let~$\hh$ be a graph class defined by a finite set~$\mathcal{F}$ of forbidden induced subgraphs on at most~$c$ vertices each.
\mic{If $(R,X, \lambda)$ is a~minimal representative in the relation of $(\hh,k)$-equivalence, then $|V(R)| = \Oh(k^c)$.}
\end{corollary}

As the next step, we need to provide an algorithm
for  \hh\textsc{-deletion} parameterized by the solution size, which works with undeletable vertices.
This can be done via a straightforward application of the technique of bounded-depth search trees.

\begin{lemma} \label{lem:meta-induced:undeletable}
Let~$\hh$ be a graph class defined by a finite set~$\mathcal{F}$ of forbidden induced subgraphs on at most~$c$ vertices each.
Then \textsc{Disjoint $\hh$-deletion} admits an algorithm with running time $c^s \cdot n^{\Oh(1)}$, where $s$ is the solution size.
\end{lemma}
\begin{proof}
Given an input~$(G,s,\ell,U)$ (the parameter $\ell$ is unused here), we start by finding an induced subgraph isomorphism~$f$ from some~$H \in \mathcal{F}$ to~$G$, if one exists.
As $c$ is constant, this can be done in time~$n^{\Oh(1)}$ by brute force. If no such~$f$ exists, then output the empty set as the optimal solution. Otherwise, let~$T := \{ f(v) \mid v \in V(H)\}$ denote the vertices in the range of~$f$. Any valid solution has to include a vertex of~$T \setminus U$. If~$s = 0$ or~$T \subseteq U$, then clearly no solution of size at most~$s$ exists and we report failure. Otherwise, for each of the at most~$c$ vertices~$v \in T \setminus U$ we recurse on the instance~$(G - v, s-1, \ell, U)$. If all recursive calls report failure, then we report failure for this call as well. If at least one branch succeeds, then we take a minimum-size solution~$S'$ returned by a recursive call and add~$v$ to it to form the output.

Since the branching is exhaustive, it is easy to see that the algorithm is correct. As the depth of the recursion tree is at most~$s$, while the algorithm branches on~$|T \setminus U| \leq c$ vertices at every step, the claimed running time follows.
\end{proof}

We can now combine all the ingredients and plug them into the meta-theorem.
Even though $c$ is constant, we can keep track of how it affects the exponent at $k$, as it follows easily from the claims above.
Observe that so far we never had to assume that the graphs in the family $\mathcal{F}$ are connected, but
this requirement is crucial for the tractability (see Section~\ref{subsec:not:closed}). 

\begin{thm}\label{thm:meta-induced:main}
Let~$\hh$ be a graph class defined by a finite set~$\mathcal{F}$ of forbidden induced subgraphs on at most~$c$ vertices each, which are all connected.
Then \textsc{$\hh$-deletion} can be solved in time~$2^{\Oh(k^{2c})} \cdot n^{\Oh(1)}$ when given a~tree $\hh$-decomposition of width~$k-1$ consisting of~$n^{\Oh(1)}$ nodes.
\end{thm}
\begin{proof}
We check the conditions of Theorem~\ref{thm:meta-uniform:main}.
The class $\hh$ is hereditary and closed under disjoint union of graphs because the forbidden subgraphs are connected.
By Lemma~\ref{lem:meta-induced:undeletable}, \textsc{Disjoint $\hh$-deletion} can be solved in time $c^s\cdot n^{\Oh(1)}$. 
Finally, by \cref{cor:meta-induced:representatives} 
\mic{and \cref{lem:representative-generation-general},
an $(\hh,\le k)$-representative family can be computed in time $v(k) = 2^{\Oh(k^{2c})}$.}
\end{proof}

\begin{corollary}\label{thm:final-induced-tw}
For any graph class~$\hh$ which is defined by a finite set~$\mathcal{F}$ of connected forbidden induced subgraphs on at most~$c$ vertices each, \textsc{$\hh$-deletion} can be solved in time~$2^{\Oh(k^{6c})} \cdot n^{\Oh(1)}$ when parameterized by~$k = \hhtw(G)$, \mic{and in time~$2^{\Oh(k^{4c})} \cdot n^{\Oh(1)}$ when parameterized by~$k = \hhdepth(G)$.}
\end{corollary}
\begin{proof}
We use \cref{thm:decomposition:full} to find a tree $\hh$-decomposition of width $\Oh(k^3)$, {which takes time $2^{\Oh(k)} \cdot n^{\Oh(1)}$}, and plug it into Theorem~\ref{thm:meta-induced:main}.
To see the second claim, observe that \cref{thm:decomposition:full} allows us to find an $\hh$-elimination forest of width $\Oh(k^2)$ in time $2^{\Oh(k^2)} \cdot n^{\Oh(1)}$.
This gives us a tree $\hh$-decomposition of the same width (see Lemma~\ref{lem:treedepth-treewidth}), which can be again supplied to Theorem~\ref{thm:meta-induced:main}.
\end{proof}

\subsubsection{Chordal deletion}
\label{sec:chordal}
In this section we develop a dynamic-programming algorithm that solves \textsc{Chordal deletion} using a tree $\mathsf{chordal}$-decomposition.
We again want to use the meta-algorithm presented in Theorem~\ref{thm:meta-uniform:main}.
To this end, we need to bound the \bmp{sizes} of representatives in the relation of $(\mathsf{chordal}, k)$-equivalence.
We obtain it through a new criterion that tests whether a graph~$G$ is chordal based on several properties of a tri-separation~$(A,X,B)$ in~$G$. We therefore first develop some theory of chordal graphs.

A \emph{hole} in a graph~$G$ is an induced cycle of length at least four. A graph is chordal if it does not contain any holes.
We need the following observation, which follows easily from the alternative characterization of chordal graphs as intersection graphs of the vertex sets of subtrees of a tree~\cite{Gavril74}.

\begin{observation}\label{obs:meta-chordal:contraction}
Chordal graphs are closed under edge contractions.
\end{observation}

A vertex~$v$ in a graph~$G$ is \emph{simplicial} if the set~$N_G(v)$ forms a clique in~$G$.
Since a hole does not contain any simplicial vertices, we have the following.

\begin{observation}\label{obs:meta-chordal:simplicial}
If~$v$ is a simplicial vertex in~$G$, then~$G$ is chordal if and only if~$G - v$ is chordal. Consequently, if a graph~$G'$ is obtained from a chordal graph~$G$ by inserting a new vertex whose neighborhood is a clique in~$G$, then~$G'$ is chordal.
\end{observation}

\begin{lemma}[{\cite[Thm.~5.1.1]{BrandstadtLS99}}] \label{lem:meta-chordal:simplicial}
Every chordal graph contains a simplicial vertex.
\end{lemma}

The following structural property of chordal graphs will be used to bound the sizes of representatives,
once their structure is revealed.

\begin{lemma} \label{lem:meta-chordal:simplicial:indset}
If~$G$ is a chordal graph and~$A \cup B$ is a partition of~$V(G)$ such that~$B$ is an independent set in~$G$ and no vertex of~$B$ is simplicial in~$G$, then~$|B| < |A|$.
\end{lemma}
\begin{proof}
Proof by induction on~$|A|$. If~$|A| = 1$ then~$B = \emptyset$, as vertices in the independent set~$B$ can either be isolated or adjacent to the unique vertex in~$A$, which would make them simplicial.

For the induction step, let~$|A| > 1$ and let~$v \in V(G)$ be a simplicial vertex, which exists by Lemma~\ref{lem:meta-chordal:simplicial}. By the precondition to the lemma,~$v \in A$. Let~$B_v := N_G(v) \cap B$. Since~$B_v$ is a clique as~$v$ is simplicial, while~$B \supseteq B_v$ is an independent set by assumption, we have~$|B_v| \leq 1$. Let~$G' := G - (\{v\} \cup B_v)$. Then the vertices of~$B \setminus B_v$ are not simplicial in~$G'$, as the non-edges in their neighborhood do not involve~$v$. By inductive assumption the graph~$G'$ with its partition into~$A' := A \setminus \{v\}$ and~$B' := B \setminus B_v$ satisfies~$|B'| = |B \setminus B_v| < |A'| = |A| - 1$. As~$|B_v| \leq 1$, this implies the lemma.
\end{proof}

Recall that a \emph{walk} from a vertex~$u$ to a vertex~$v$ in a graph~$G$ is a sequence of (not necessarily distinct) vertices, starting with~$u$ and ending with~$v$, such that consecutive vertices are adjacent in~$G$. The vertices~$u$ and~$v$ are the \emph{endpoints} of the walk; all other vertices occurring on the walk are \emph{internal} vertices. The following observation gives an easy way to certify that a graph is not chordal.

\begin{observation}[{\cite[Proposition 3]{Marx10}}] \label{obs:meta-chordal:find-hole}
If a graph $G$ contains a vertex $v$ with two nonadjacent neighbors $u_1,u_2 \in N_G(v)$, and a walk from $u_1$ to $u_2$ with all internal vertices in $V(G) \setminus N_G[v]$, then $G$ contains a hole passing through $v$.
\end{observation}

We are ready to formulate an operation used to produce representatives of bounded size.

\begin{definition}
Let~$(G,X,\lambda)$ be a $k$-boundaried graph.
The operation of \emph{condensing} $(G,X,\lambda)$ is defined as follows.
\begin{itemize}
    \item For each connected component~$B_i$ of~$G-X$ for which~$G[N_G(B_i)]$ is a clique, called a \emph{simplicial component}, remove all vertices of~$B_i$.
    \item For each connected component~$B_i$ of~$G-X$ for which~$G[N_G(B_i))]$ is \emph{not} a clique, called a \emph{non-simplicial component}, contract~$B_i$ to a single vertex.
\end{itemize}
\end{definition}

We want to show that condensing boundaried graphs preserves its equivalence class.
To this end, we show that we can harmlessly contract any edge which is not incident with the separator~$X$. 

\begin{lemma}\label{lem:meta-chordal:contraction-reverse}
Let~$G$ be a graph, $(A,X,B)$ be a tri-separation in~$G$, and let $u,v \in B$, $uv \in E(G)$. 
Then~$G$ is chordal if and only if the following conditions hold:
\begin{enumerate}
    \item The graphs~$G[A \cup X]$ and~$G[B \cup X]$ are chordal.
    \item The graph~$G / uv$ is chordal.
\end{enumerate}
\end{lemma}
\begin{proof}
If~$G$ is chordal, then since chordal graphs are hereditary and closed under edge contractions by Observation~\ref{obs:meta-chordal:contraction}, both conditions are satisfied. 

We prove the reverse implication.
Suppose that $G[A \cup X], G[B \cup X]$ and $G / uv$ are chordal but $G$ is not, that is, it contains a hole $H$.
As both~$G[A \cup X]$ and~$G[B \cup X]$ are chordal, hole~$H$ contains some~$a \in A$ and some~$b \in B$. 

Consider a vertex~$a \in A$ that lies on hole~$H$, and let~$p,q$ be the predecessor and successor of~$a$ on the hole. Then~$p,q \in N_G(a)$ and therefore~$p,q \notin B$ by the properties of a tri-separation.  Furthermore,~$pq \notin E(G)$ since a hole is chordless. The subgraph~$P := H - \{a\}$ forms an induced path between~$p$ and~$q$ in~$G$. Since~$u,v \in B$ and~$a \in A$, we have~$u,v \notin N_G(a)$. As contractions preserve the connectivity of subgraphs, when contracting edge~$uv$ the path~$P$ turns into a (possibly non-simple) walk between nonadjacent~$p,q \in N_G(a)$ in~$G / uv$, whose internal vertices avoid~$N_G[a]$ since neither of the contracted vertices is adjacent to~$a$. By Observation~\ref{obs:meta-chordal:find-hole}, this implies~$G / uv$ contains a hole; a contradiction.
\end{proof}

\begin{lemma}\label{lem:meta-chordal:equivalent}
Let~$(G_1,X,\lambda)$ be a $k$-boundaried graph, so that $G_1$ is chordal,
and let $(G_2,X,\lambda)$ be obtained by condensing $(G_1,X,\lambda)$. Then $(G_1,X,\lambda)$ and $(G_2,X,\lambda)$ are $(\mathsf{chordal}, k)$-equivalent.
\end{lemma}
\begin{proof}
The condensing operation does not affect the boundary $X$ so these graphs are compatible.
Since $G_1$ is chordal, the same holds for $G_2$ because chordal graphs are closed under contracting edges and removing vertices.
Consider a $k$-boundaried graph $\widehat{H}$ compatible with $(G_1,X, \lambda)$.
By the same argument as above, if $\widehat{H} \oplus (G_1,X, \lambda)$ is chordal, then $\widehat{H} \oplus (G_2,X,\lambda)$ is as well.

We now prove the second implication.
Suppose that  $\widehat{H} \oplus (G_2,X,\lambda)$ is chordal, so \mic{the underlying graph in $\widehat{H}$} is chordal as well.
Let $(A,X,B)$ be a tri-separation of $G = \widehat{H} \oplus (G_1,X, \lambda)$, so that $(G[A \cup X], X, \lambda)$ is isomorphic with $\widehat{H}$ and $(G[B \cup X], X, \lambda)$ is isomorphic with $(G_1,X, \lambda)$.
By the definition of condensing, $\widehat{H} \oplus (G_2,X, \lambda)$ can be obtained from $G$ by contracting each connected component of $G[B]$ to a single vertex -- let us refer to this graph as $G'$ -- and then removing some simplicial vertices.
Since~$G'$ can be obtained from the chordal graph~$\widehat{H} \oplus (G_2,X,\lambda)$ by inserting simplicial vertices, then~$G'$ is chordal by Observation~\ref{obs:meta-chordal:simplicial}.
Let $G = G^1, G^2, \dots, G^m = G'$ be the graphs given by the series of edge contractions that transforms $G$ into $G'$.
We prove that if $G^{i+1}$ is chordal, then $G^{i}$ is as well.
The graph $G^{i}$ admits a tri-separation $(A,X,B^{i})$, so that $G^i[B^i \cup X]$ is obtained from $G[B \cup X]$ via edge contractions, therefore  $G^i[B^i \cup X]$ is chordal.
Moreover, $G^i[A \cup X]$ is isomorphic with \mic{the underlying graph in $\widehat{H}$}, so it is also chordal, and $G^{i+1}$ is obtained \mic{from $G^i$} by contracting an edge in $B^i$.
We can thus apply Lemma~\ref{lem:meta-chordal:contraction-reverse} to infer that $G^i$ is chordal.
It follows that $G = \widehat{H} \oplus (G_1,X,\lambda)$ is chordal, which finishes the proof.
\end{proof}

\begin{corollary}\label{cor:meta-chordal:representatives}
\mic{If $(R,X, \lambda)$ is a minimal representative in the relation of $(\mathsf{chordal},k)$-equivalence}, {$k > 0$,}
then {$|V(R)|\le 2k - 1$.} 
\end{corollary}
\begin{proof}
By Lemma~\ref{lem:meta-chordal:equivalent} we obtain
that condensing preserves $(\mathsf{chordal}, k)$-equivalence.
If $(R,X, \lambda)$ is a {minimal representative (so $R$ is chordal)},
it must be condensed.
Therefore $V(R) \setminus X$ is an independent set
and \bmp{no} $v \in V(R) \setminus X$ is simplicial.
From Lemma~\ref{lem:meta-chordal:simplicial:indset}
we get that $|V(R) \setminus X| < |X|$ and therefore $|V(R)| \le 2k - 1$.
\end{proof}

The machinery developed so far is sufficient to obtain an FPT algorithm for \textsc{Chordal deletion} on a standard tree decomposition. To be able to accommodate tree $\mathsf{chordal}$-decompositions, which can contain leaf bags with arbitrarily large chordal base components, we need to be able to efficiently compute exhaustive families for such base components. Towards this end, we will use the algorithm by Cao and Marx for the parameterization by the solution size as a subroutine.

\begin{thm}[{\cite[Thm.~1.1]{CaoM16}}] \label{thm:meta-chordal:caomarx:ChD}
There is an algorithm that runs in time~$2^{\Oh(k \log k)} \cdot n^{\Oh(1)}$ which decides, given a graph~$G$ and integer~$k$, whether or not~$G$ has a chordal deletion set of size \mic{at most}~$k$.
\end{thm}

By self-reduction and some simple graph transformations, the above algorithm can be
adapted to our setting.

\jjh{
\begin{lemma}\label{lem:meta-chordal:self-reduction}
There is an algorithm with running time~$2^{\Oh(k \log k)} \cdot n^{\Oh(1)}$ that solves \textsc{Chordal deletion} parameterized by the solution size $k$.
\end{lemma}
\begin{proof}
By trying all values~$k'$ from~$0$ to~$k$ with Theorem~\ref{thm:meta-chordal:caomarx:ChD}, we can determine whether there is a chordal deletion set in~$G'$ of size at most~$k$, and if so determine the minimum size~$k'$ of such a set. If such a set exists, then using Theorem~\ref{thm:meta-chordal:caomarx:ChD} as a subroutine it is easy to find one by self reduction. In particular, if the optimum value is~$k'$ then a vertex~$v$ belongs to an optimal solution precisely when the instance obtained by removing~$v$ has a solution of size~$k' - 1$. Hence by calling the algorithm~$n^{\Oh(1)}$ times for parameter values~$k' \leq k$, we find a set~$S$ as desired or conclude that no such set exists.
\end{proof}

}

In order to enforce the requirement that some vertices are not allowed to be part of a solution, we use the following consequence of the fact that holes do not contain vertices sharing the same closed neighborhood.

\begin{observation}\label{obs:meta-chordal:truetwin:staychordal}
Let~$G$ be a chordal graph and let~$v \in V(G)$. If~$G'$ is obtained from~$G$ by making a true-twin copy of~$v$, that is, by inserting a new vertex~$v'$ which becomes adjacent to~$N_G[v]$, then~$G'$ is chordal.
\end{observation}

\begin{thm}\label{thm:meta-chordal:main}
The \textsc{Chordal deletion} problem can be solved in time~$2^{\Oh(k^2)} \cdot n^{\Oh(1)}$ when given a~tree $\mathsf{chordal}$-decomposition of width~$k-1$ consisting of~$n^{\Oh(1)}$ nodes.
\end{thm}
\begin{proof}
We check the conditions of Theorem~\ref{thm:meta-uniform:main}.
The class of chordal graphs is clearly closed under vertex deletion and disjoint union of graphs. 
\jjh{Because of Lemma~\ref{lem:meta-chordal:self-reduction} and Observation~\ref{obs:meta-chordal:truetwin:staychordal}}, \mic{we can apply Lemma~\ref{lem:meta-uniform:undeletable} to solve \textsc{Disjoint chordal deletion} in time $2^{\Oh(s \log s)}\cdot n^{\Oh(1)}$. }
Next, by Corollary~\ref{cor:meta-chordal:representatives}
\mic{and \cref{lem:representative-generation-general},
an $(\mathsf{chordal},\le k)$-representative family can be computed in time $v(k) = 2^{\Oh(k^{2})}$.}
\end{proof}

\begin{corollary}\label{thm:final-chordal-tw}
The \textsc{Chordal deletion} problem can be solved in time~$2^{\Oh(k^{10})} \cdot n^{\Oh(1)}$ when parameterized by~$k = \mathbf{tw}_\mathsf{chordal}(G)$,
\mic{and in time~$2^{\Oh(k^{6})} \cdot n^{\Oh(1)}$ when parameterized by~$k = \mathbf{ed}_\mathsf{chordal}(G)$.}
\end{corollary}
\begin{proof}
\mic{
We use \cref{thm:decomposition:full} to find a tree $\mathsf{chordal}$-decomposition of width $\Oh(k^5)$, {which takes time $2^{\Oh(k^2 \log k)} \cdot n^{\Oh(1)}$,} and plug it into Theorem~\ref{thm:meta-chordal:main}.
The second claim follows again from \cref{thm:decomposition:full} by computing
a~$\mathsf{chordal}$-elimination forest of width $\Oh(k^3)$ in time $2^{\Oh(k^2)} \cdot n^{\Oh(1)}$.
This gives us a tree $\mathsf{chordal}$-decomposition of the same width (see Lemma~\ref{lem:treedepth-treewidth}), which can {again be} supplied to Theorem~\ref{thm:meta-chordal:main}.}
\end{proof}

\subsubsection{Interval deletion}
\label{sec:interval-deletion}
An interval graph is the intersection graph of intervals of the real line. In an interval model $\mathcal{I}_G = \{I(v) \mid v \in V(G)\}$ of a graph $G$, each vertex $v \in V(G)$ corresponds to a closed interval $I(v) = [\lp(v),\rp(v)]$, with left and right endpoints $\lp(v)$ and $\rp(v)$ such that $\lp(v) < \rp(v)$\bmp{; there is an edge between vertices~$u$ and~$v$ if and only if~$I(v) \cap I(u) \neq \emptyset$}. 
The goal of this section is to show that \textsc{Interval deletion} is FPT parameterized by $k = \hhtw[interval](G)$. In order to apply Theorem~\ref{thm:meta-uniform:main} (with 
\cref{lem:representative-generation-general}), we aim to bound the size of a minimal representative for the $(\mathsf{interval},k)$-equivalence classes.
We introduce some notation and definitions.
Since an edge contraction can be seen as \bmp{merging} two overlapping intervals, we have the following observation.
\begin{observation}\label{obs:interval:contractions}
Interval graphs are closed under edge contractions.
\end{observation}

For $u,v \in V(G)$, we say that $I(u)$ is strictly right of $I(v)$ (equivalently $I(v)$ is strictly left of $I(u)$) if $\lp(u) > \rp(v)$. We \bmp{denote} this by $I(u) > I(v)$ (equivalently $I(v) < I(u)$) for short. An interval model is called \emph{normalized} if no pair of distinct intervals shares an endpoint. Every interval graph has a normalized interval model that can be produced in linear time (cf.~\cite{Cao16}).
We use the following well known characterization of interval graphs.
Three distinct vertices $u,v,w \in V(G)$ form an \emph{asteroidal triple} (AT) of $G$ if for any two of them there is a path between them avoiding the closed neighborhood of the third.

\begin{thm}[\cite{LekkerkerkerB1962} cf.~\cite{BrandstadtLS99}] \label{thm:interval:chordal:atfree}
A graph $G$ is an interval graph if and only if $G$ is chordal and contains no AT.
\end{thm}

A vertex set $M \subseteq V(G)$ is a module of $G$ if $N_G(u) \setminus M = N_G(v) \setminus M$ for all $u,v \in M$. A module $M$ is \emph{trivial} if $|M| \leq 1$ or $|M| = |V(G)|$, and \emph{non-trivial} otherwise. \bmp{Throughout the section, we use the following terminology. An \emph{obstruction} in a graph~$G$ is an inclusion-minimal vertex set~$X$ such that~$G[X]$ is not interval.}

\begin{lemma}[{\cite[Proposition 4.4]{CaoM2015}}]\label{lem:forbiddenmodule}
Let $G$ be a graph and $M \subseteq V(G)$ be a module. 
If $X \subseteq V(G)$ is an obstruction and~$|X| > 4$, then either \bmp{$X \subseteq M$} or $|M \cap X| \leq 1$.
\end{lemma}
 
\bmp{
Theorem~\ref{thm:interval:chordal:atfree} implies that all obstructions induce connected graphs, as the obstructions to chordality---chordless cycles and minimal subgraphs containing an AT---are easily seen to be connected.} 
The only obstruction of no more than four vertices induces a $C_4$ (cf.~\cite{CaoM2015}). We use the following consequence.
 
\begin{lemma}\label{lem:Bmodule_singlevertex} 
Let $(A,X,B)$ be a tri-separation of $G$, and let $M \subseteq B$ be a module \bmp{in~$G$}. If $G[X \cup B]$ is interval, then any \jjh{obstruction} \bmp{in~$G$ contains} at most one vertex of $M$. Furthermore, \bmp{for each \jjh{obstruction~$S$} intersecting~$M$, for each~$v \in M$, the set~$(\jjh{S} \setminus M) \cup \{v\}$ \jjh{is also an obstruction.}}
\end{lemma}
\begin{proof}
\bmp{We first derive the first part of the statement.} 
For any \jjh{obstruction $S$} larger than four vertices, the statement follows from Lemma~\ref{lem:forbiddenmodule} \bmp{since no \jjh{obstruction} can be fully contained in~$M \subseteq B$ as~$G[X \cup B]$ is interval}. For the case of $\jjh{G[S]}$ isomorphic to $C_4$, at least one of its vertices must be in $A$ \bmp{as~$G[X \cup B]$ is interval}. Since $X$ is a separator, it follows that $|M \cap S| \leq |B \cap S| \leq 1$. 

For the second part, consider some \jjh{obstruction $S$} intersecting $M$. By the arguments above, this intersection is a single vertex, say, $u$. The statement for $u=v$ is clear as then $(S \setminus M) \cup \{v\} = S$. In all other cases, since none of $S \setminus \{u\}$ is part of $M$, \bmp{by definition of a module} it follows that $v$ has the exact same neighborhood to $S \setminus \{u\}$ as $u$. Hence, the graph induced by $(S \setminus M) \cup \{v\}$ is isomorphic to $G[S]$.
\end{proof}

\mic{We arrive at the first useful observations about the structure of minimal representatives.}  

\begin{lemma}\label{lem:no-non-trivial-modules}
If the $k$-boundaried graph $(G,X,\lambda)$ is a minimal representative in the relation of $(\mathsf{interval},k)$-equivalence and $G$ is interval, then $G$ has no non-trivial module $M \subseteq V(G) \setminus X$.
\end{lemma}
\begin{proof}
\jjh{
For the sake of contradiction, suppose that $G$ has a non-trivial module $M \subseteq V(G) \setminus X$. Pick an arbitrary vertex $v \in M$. We argue that $(G' = G-(M \setminus \{v\}),X,\lambda)$ is $(\mathsf{interval},k)$-equivalent to $(G,X,\lambda)$.
Consider a $k$-boundaried graph $H$ compatible with $(G,X,\lambda)$. Note that $H$ is compatible with $(G',X,\lambda)$ too. 

First suppose that $H \oplus (G,X,\lambda)$ is not interval. 
Consider the tri-separation $(V(H) \setminus X,X,V(G) \setminus X)$ of $F = H \oplus (G,X,\lambda)$. Let $S \subseteq V(F)$ be an obstruction.
By Lemma~\ref{lem:Bmodule_singlevertex} we have that $|S \cap M| \leq 1$ and furthermore that $S' = (S \setminus M) \cup \{v\}$ is an obstruction. It follows that $H \oplus (G',X,\lambda)$ contains the obstruction $S'$ and hence is not interval. 
Now suppose that $H \oplus (G,X,\lambda)$ is interval. Since interval graphs are hereditary, it follows that $H \oplus (G',X,\lambda)$ is also an interval graph.
}
\end{proof}

\paragraph{Marking scheme}
We proceed by marking a set of~$|X|^{\Oh(1)}$ vertices $Q \subseteq V(G)$ such that for any compatible $k$-boundaried graph $H$, the following holds: if $H \oplus (G,X,\lambda)$ contains an asteroidal triple, then it contains an AT~$(v_1, v_2, v_3)$ such that $\{v_1, v_2, v_3\} \cap V(G) \subseteq X \cup Q$. 
Our bound on the size of the minimal representative is then obtained by analyzing the size of $G-(X \cup Q)$. Before getting to the marking scheme, we introduce some definitions and notation.

For a path $P$ and $x,y \in V(P)$ let $P[x,y]$ be the subpath of $P$ from $x$ to $y$.
For a set $U \subseteq V(G)$ let $\mathcal{P}(P,U)$  be the family of maximal subpaths of $P$ contained in $U$.

\jjh{
\begin{observation}\label{obs:subpath_boundary}
Consider a vertex set $U \subseteq V(G)$. Let $P$ be a path whose \bmp{endpoints} are contained in $N(V(G) \setminus U)$. Then for each $Q \in \mathcal{P}(P,U)$, the \bmp{endpoints} of $Q$ are contained in $N(V(G) \setminus U)$.
\end{observation}
}

For a $k$-boundaried graph $(G,X,\lambda)$ such that $G$ is interval, and \jjh{a} normalized interval model $\mathcal{I} =  \{I(v) \mid v \in V(G)\}$, we shortly say that $(G,X,\lambda,\mathcal{I})$ is a $k$-boundaried interval graph with a model.
\jjh{Given a $k$-boundaried graph with a model $(G,X,\lambda,\mathcal{I})$ and connected vertex set $A \subseteq V(G)$, let \bmp{$I(A) = \bigcup_{a \in A} I(a)$} denote the union of intervals of the vertices in $A$. \bmp{Since~$A$ is connected, $I(A)$ is itself an interval~$[\lp(A), \rp(A)]$ with } $\lp(A) = \min_{a \in A}\lp(a)$ and $\rp(A) = \max_{a \in A}\rp(a)$.}

\begin{lemma}\label{lem:interval:linear-forest}
Let $(G,X,\lambda,\mathcal{I})$ be a $k$-boundaried interval graph with a model and $H$ be a $k$-boundaried graph compatible with $(G,X,\lambda)$.
For a chordless path $P$ in $H \oplus (G,X,\lambda)$, let $\mathcal{I}(P) = \{I(V(Q)) \mid Q \in \mathcal{P}(P,V(G))\}$.
Then $\mathcal{I}(P)$ is a set of pairwise disjoint intervals. \jjh{Furthermore, if $P$ is disjoint from \mic{$N_G[u]$} for some $u \in V(G)$, then these intervals are disjoint from $I(u)$.}


\end{lemma}
\begin{proof}
\jjh{First observe that for each $Q \in \mathcal{P}(P,V(G))$ we have that $V(Q)$ is a connected vertex set, namely a chordless path, and therefore $I(V(Q))$ is well-defined. If for any two distinct $Q, Q' \in \mathcal{P}(P,V(G))$, the intervals $I(V(Q))$ and $I(V(Q'))$ would overlap, then either the paths were not maximal subpaths of $P$, or $P$ would not be chordless. 

To see the second part, note that any overlap between $I(V(Q))$ and $I(u)$ would imply that $u$ is adjacent \bmp{(or equal)} to some vertex of $V(Q) \subseteq V(P)$.}
\end{proof}

\mic{We would like to encode all the relevant information about a path that connects two vertices $(v_1,v_2)$ and avoids the neighborhood of a vertex $u$,
so later we could argue that some vertex in an AT can be replaced with another one.
Since the boundaried graph $H$ is unknown, we want to encode the subpaths that might appear within $G$, in particular their starting and ending points in $X$.
However there might be $\Omega(|X|)$ such subpaths and exponentially-many combinations of starting/ending points.
We shall show that only the two subpaths including the vertices $v_1,v_2$ and the (at most) two subpaths closest to $u$ in the interval model are relevant.
This means we only need to encode $\Oh(1)$ subpaths which gives only $|X|^{\Oh(1)}$ combinations.
We begin with formalizing the concept of encoding a path.
}

\begin{definition}\label{def:signature}
Let $(G,X,\lambda,\mathcal{I})$ be a $k$-boundaried interval graph with a model and $H$ be a $k$-boundaried graph compatible with $(G,X,\lambda)$.
Furthermore, let $F = H \oplus (G,X,\lambda)$, $v_1,v_2,u \in V(F)$, and $P$ be a chordless $(v_1, v_2)$-path in $F - N_F[u]$.
The signature $S$ of $(P,u)$
with respect to $(G,X,\lambda,\mathcal{I})$
is defined as follows.

If $V(P) \cap X = \emptyset$ then $S$ is \emph{trivial}. 
Otherwise $S$ is a triple $(x_1, x_2, \mathcal{X})$ where $x_1, x_2 \in X$
and $\mathcal{X}$ is a set of ordered pairs from $X$.
\jjh{Let $x_1$ be the first vertex of $P$ starting from $v_1$ with $x_1 \in X$. Similarly let $x_2$ be the first vertex of $P$ \bmp{in~$X$} starting from $v_2$.}
If $u \not\in V(G)$, then $\mathcal{X} = \emptyset$.
\bmp{Otherwise, if there exists $Q \in \mathcal{P}(P[x_1, x_2],V(G))$ such that $I(V(Q)) < I(u)$, choose such~$Q_\ell = (w_1,\dots,w_{|V(Q_\ell)|})$ with maximal $\rp(V(Q_\ell))$ and add the pair $(w_1,w_{|V(Q_\ell)|})$ to $\mathcal{X}$. Similarly add a pair for a path $Q_r$ with minimal $\lp(V(Q_r))$ such that $I(u) < I(V(Q_r))$ if such $Q_r$ exists.}
\end{definition}

\begin{figure}
    \centering
    \includegraphics[page=1]{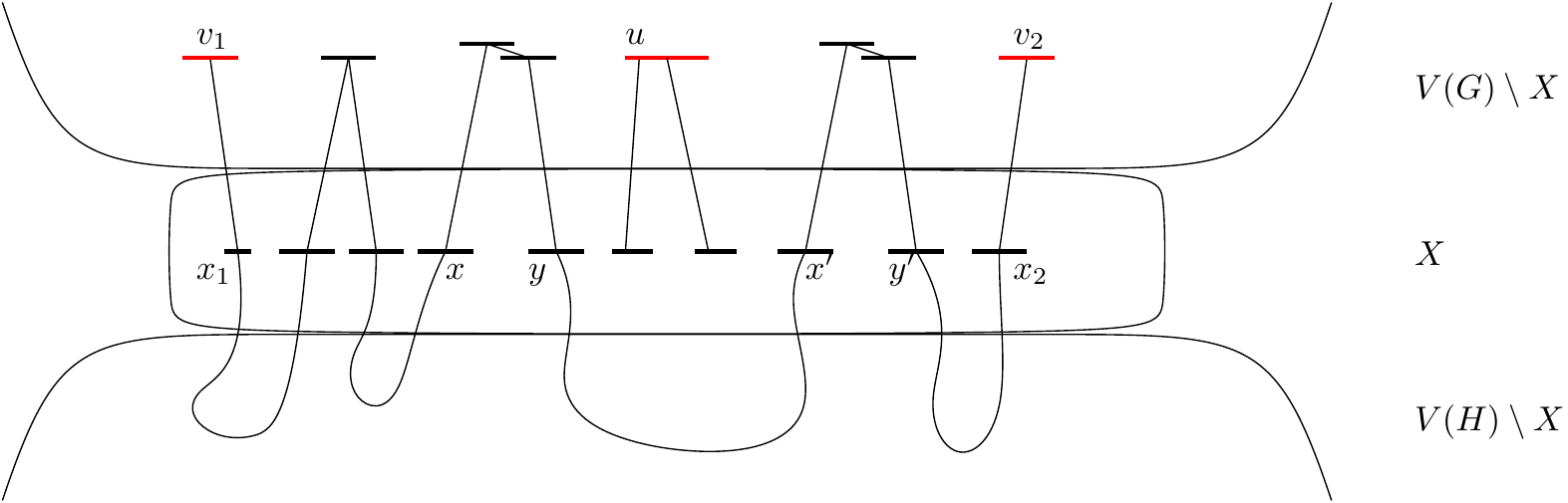}
    \caption{\bmp{Schematic illustration of a graph~$F = (G,X,\lambda) \oplus H$, where $(G,X,\lambda,\mathcal{I})$ is a $k$-boundaried interval graph with a model. A~$(v_1, v_2)$-path~$P$ in~$F - N_F[u]$ is shown.} The signature of $(P,u)$ is the triple $(x_1,x_2,\mathcal{X} = \{(x,y),(x',y')\})$. 
    }
    \label{fig:signature}
\end{figure}

\jjh{Note that the definition above} is well-defined due to \cref{lem:interval:linear-forest}. An example is shown in Figure~\ref{fig:signature}. In a signature we may have $x_1 = x_2$ and for any ordered pair $(x,y) \in \mathcal{X}$, possibly $x = y$.
\jjh{By Observation~\ref{obs:subpath_boundary} it follows that the pairs in $\mathcal{X}$ \bmp{consist of elements of} $X$. Since a non-trivial signature $S$ with respect to $(G,X,\lambda)$ can be represented as a sequence of at most six vertices of $X$, we observe the following.}

\begin{observation}
Let $\mathcal{S}(G,X,\lambda,\mathcal{I})$ be the family of all possible signatures with respect to $(G,X,\lambda,\mathcal{I})$.
Then  $|\mathcal{S}(G,X,\lambda,\mathcal{I})| = \Oh(|X|^6)$.
\end{observation}

\mic{We now define the obedience relation between a signature and a triple of vertices.
We want to ensure the following properties: (1) if $P$ is an $(v_1,v_2)$-path in $H \oplus (G,X,\lambda)$ avoiding the closed neighborhood of a vertex $u$, then $(v_1,v_2,u)$ obeys the signature of $(P,u)$, and (2) if two ``similar'' triples \bmp{obey} some signature, then for any choice of $H$ the desired path exists either for both triples or for none of them.

A technical issue occurs when we want to consider triples of vertices not only from $G$ but from $H \oplus (G,X,\lambda)$.
Since our framework should be oblivious to the choice of $H$, we introduce the symbol $\bot$ as a placeholder for a vertex from \mic{$H-X$.}
In the definition below we assume $N_G[\bot] = \emptyset$.}

\begin{definition}\label{def:obey_signature}
Let $(G,X,\lambda,\mathcal{I})$ be a $k$-boundaried interval graph with a model and $v_1, v_2, u \in V(G) \cup \{\bot\}$.
We say that $(v_1, v_2, u)$ obeys the trivial signature if  $v_1 = v_2 = \bot$ or $G-X-N_G[u]$ contains a $(v_1,v_2)$-path (the latter implies that $v_1, v_2 \in V(G) \setminus X$).
We say that $(v_1, v_2, u)$ obeys a non-trivial signature $S = (x_1, x_2, \mathcal{X}) \in \mathcal{S}(G,X,\lambda,\mathcal{I})$ if
all the following conditions hold. 
\begin{enumerate}
    \item $v_1 = \bot$ or there is a $(v_1, x_1)$-path \jjh{contained in $G-(X \setminus \{x_1\}) -N_G[u]$},
    \item $v_2 = \bot$ or there is a $(v_2, x_2)$-path \jjh{contained in $G-(X \setminus \{x_2\}) -N_G[u]$},
    \item $u = \bot$ or for each $(x, y) \in \mathcal{X}$ there is an $(x,y)$-path in $G-N_G[u]$.
\end{enumerate}
\end{definition}

\mic{We give some intuition behind the obedience definition above. Consider some $k$-boundaried graph $H$ compatible with $(G,X,\lambda)$.
Let $(v_1,v_2,u)$ be an AT in $F = H \oplus (G,X,\lambda)$, so there is a $(v_1,v_2)$-path $P$ in $F-N_F[u]$.
Suppose we want to replace some vertex from $(v_1,v_2,u)$ with another vertex from $F$, so that the new triple would still obey the signature of $(P,u)$, and certify that an analogous path exists.
If we replace $v_1$ or $v_2$ we will need to update the $(v_1, x_1)$-subpath (resp. $(v_2, x_2)$-subpath) of $P$.
The first (resp. second) condition certifies that such an update is possible: if $u \in V(G)$ then it directly states that the new subpath avoids $N_F[u] \cap V(G) = N_G[u]$ and if $u \not\in V(G)$ (this translates to $u = \bot$) then $N_F[u] \cap V(G) \subseteq X$ and we do not introduce any new vertices from $X$.
If we aim at replacing $u$, then the third condition states that we can update the two subpaths of $P$ which are closest to $u$ in the interval model $\mathcal{I}$---we will show that this is sufficient.}

Let $(G,X,\lambda,\mathcal{I})$ be a $k$-boundaried interval graph with a model.
We set $v^\bot = v$ if $v \in V(G)$ or $\bot$ otherwise, assuming that $G$ is clear from the context. 
We now show that the obedience relation satisfies the intuitive property that whenever a $(v_1, v_2)$-path $P$ avoids the closed neighborhood of $u$ then   $(v_1^\bot, v_2^\bot, u^\bot)$ obeys the signature of $(P,u)$.

\begin{lemma}\label{lem:interval:obedience}
Let $(G,X,\lambda,\mathcal{I})$ be a $k$-boundaried interval graph with a model and $H$ be a $k$-boundaried graph compatible with $(G,X,\lambda)$.
Furthermore, let $F = H \oplus (G,X,\lambda)$, $v_1,v_2,u \in V(F)$, $P$ be a chordless $(v_1, v_2)$-path in $F - N_F[u]$, and $S$ be the signature of $(P,u)$ with respect to $(G,X,\lambda,\mathcal{I})$. 
Then $(v_1^\bot, v_2^\bot, u^\bot)$ obeys $S$.
\end{lemma}
\begin{proof}
\jjh{
We do a case distinction on $V(P) \cap X$. First suppose that $V(P) \cap X = \emptyset$. By Definition~\ref{def:signature} we have that $S$ is trivial. In the case that $V(P) \subseteq V(H) \setminus X$, then $v_1^\bot = v_2^\bot = \bot$ and therefore $(v_1^\bot, v_2^\bot, u^\bot)$ obeys $S$. In the case that $V(P) \subseteq V(G) \setminus X$, then $P$ is a $(v_1,v_2)$-path in $G-X-N_G[u]$ and again $(v_1^\bot, v_2^\bot, u^\bot)$ obeys $S$. 

Next, suppose that $V(P) \cap X \neq \emptyset$. Then by Definition~\ref{def:signature} we have $S = (x_1,x_2,\mathcal{X})$. We check the obedience conditions of Definition~\ref{def:obey_signature} for $(v_1^\bot, v_2^\bot, u^\bot)$. 
\begin{enumerate}
    \item If $v_1 \notin V(G)$, then $v_1^\bot = \bot$ and the first condition clearly holds. Otherwise, $v_1 \in V(G)$, and $x_1$ is the first vertex of $P$ starting from $v_1$ with $x_1 \in X$. If $u \notin V(G)$, then $u^\bot = \bot$ and by construction there is a $(v_1,x_1)$-path contained in $G-(X \setminus \{x_1\})=G-(X \setminus \{x_1\})-N_G[u^\bot]$. Otherwise, $u^\bot = u \in V(G)$ and $N_G[u] \subseteq N_F[u]$. Since $P$ is disjoint from $N_F[u]$, it follows that there is a $(v_1,x_1)$-path contained in $G-(X \setminus \{x_1\})-N_G[u^\bot]$.
    \item The argument for the second condition is symmetric to the first condition.
    \item If $u \notin V(G)$, then $u^\bot = \bot$ and the third condition clearly holds. Otherwise, $u^\bot = u \in V(G)$. By \bmp{definition of~$\mathcal{X}$}, for each pair $(x,y) \in \mathcal{X}$, there is a subpath of $P$ in $G-N_G[u]$ starting at $x$ and ending at $y$. It follows that the third condition holds.\qedhere 
\end{enumerate}
}
\end{proof}

\mic{In order to introduce the concept of replacing vertices in a triple, we formalize what we mean by saying that two triples are ``similar''.
Simply speaking, we consider the endpoints of intervals of the boundary vertices and treat two vertices as equivalent if their intervals contain the same set of ``boundary endpoints''.
We give the definition for a non-necessarily boundaried graph as later we will use it in a more general context.
}

\begin{definition}\label{def:interval:regions-new}
Let $G$ be an interval graph with a normalized interval model $\mathcal{I} =  \{I(v) \mid v \in V(G)\}$
\mic{and $U \subseteq V(G)$}.
The endpoints of intervals of $U$ partition the real line into $z = 2\cdot |U| + 1$ regions. Let $(x_1,\dots,x_{z-1})$ be an increasing order of these endpoints and let $x_0 = -\infty$ and $x_z = \infty$. We define the set of subsets $\mathcal{J}_{U}^\mathcal{I} = \{J_{i,j} \subseteq {V(G) \setminus U} \mid i \leq j \in [z]\}$, where $u \in {V(G) \setminus U}$ is in $J_{i,j}$ if and only if $i$ is the smallest index such that the intervals $[\lp(u),\rp(u)]$ and $[x_{i-1},x_i]$ have a non-empty intersection, and $j$ is the largest index such that $[\lp(u),\rp(u)]$ and $[x_{j-1},x_j]$ have a non-empty intersection.
\end{definition}

The family $\mathcal{J}_{U}^\mathcal{I}$ clearly forms a partition of $V(G) \setminus U$. 
For a superset $A$ of $V(G)$ and $x,y \in A$ we say $x,y$ are $(G,X,\mathcal{I})$-equivalent if $x=y$ or $x,y \in V(G) \setminus X$ and $x,y$ belong to the same set in the partition $\mathcal{J}_{X}^\mathcal{I}$ \bmp{with respect to~$X$}. 
Two $\ell$-tuples over $A$ are $(G,X,\mathcal{I})$-equivalent if they are pointwise $(G,X,\mathcal{I})$-equivalent.

\mic{We now prove the main technical lemma which states that whenever two $(G,X,\mathcal{I})$-equivalent triples obey some signature, then
the desired path exists either for both triples of vertices 
or for none of them.
}

\begin{lemma}\label{lem:interval:path-exists}
Let $(G,X,\lambda,\mathcal{I})$ be a $k$-boundaried interval graph with a model and $H$ be a $k$-boundaried graph compatible with $(G,X,\lambda)$.
Furthermore, let $F = H \oplus (G,X,\lambda)$, $v_1,v_2,u,w_1,w_2,z \in V(F)$, and $P$ be a chordless $(v_1, v_2)$-path in $F - N_F[u]$.
Suppose that triples $(v_1, v_2, u)$ and $(w_1, w_2, z)$ are $(G,X, \mathcal{I})$-equivalent and $(w_1^\bot, w_2^\bot, z^\bot)$ obey the signature of $(P,u)$ \bmp{in $(G,X,\lambda,\mathcal{I})$}.
Then there exists a $(w_1,w_2)$-path in $F-N_F[z]$.
\end{lemma}
\begin{proof}
Let $S$ be the signature of $(P,u)$. We do a case distinction on $V(P) \cap X$. First suppose that $V(P) \cap X = \emptyset$. By Definition~\ref{def:signature} we have that $S$ is trivial. 
\begin{itemize}
    \item In the case that $V(P) \subseteq V(H) \setminus X$, we have $v_1^\bot = v_2^\bot = \bot$. By the $(G,X,\mathcal{I})$-equivalence we have that $v_1 = w_1$ and $v_2 = w_2$. If $u \notin V(G) \setminus X$, then by the $(G,X,\mathcal{I})$-equivalence we have $u = z$ and the path $P$ satisfies the lemma. Otherwise $u \in V(G) \setminus X$ and by the $(G,X,\mathcal{I})$-equivalence we have $z \in V(G) \setminus X$ and therefore $N_F[z] \subseteq V(G)$. Since $V(P) \subseteq V(H) \setminus X$, again the path $P$ satisfies the lemma.
    
    \item Next consider the case that $V(P) \subseteq V(G) \setminus X$. Since $v_1,v_2 \in V(G) \setminus X$, by the $(G,X,\mathcal{I})$-equivalence we have that $w_1,w_2 \in V(G) \setminus X$. Since $(w_1^\bot, w_2^\bot, z^\bot)$ obeys $S$, it follows that there is a $(w_1,w_2)$-path in $G-X-N_G[z^\bot]$. 
\end{itemize}

Next, suppose that $V(P) \cap X \neq \emptyset$. By Definition~\ref{def:signature} we have $S = (x_1,x_2,\mathcal{X})$. We transform $P$ into the required path. We do a case distinction on the location of $u \in V(F)$.
\begin{itemize}
    \item Suppose $u \in X$, then $u = z$ by the $(G,X,\mathcal{I})$-equivalence and therefore $P$ is a $(v_1,v_2)$-path in $F - N_F[z]$. We first argue that there is a $(w_1,x_1)$-path in $F - N_F[z]$. This is trivially true if $w_1^\bot = \bot$, since then $w_1 = v_1$ by the $(G,X,\mathcal{I})$ equivalence and~$P[v_1, x_1]$ is such a path. Otherwise, because of the first obedience condition it follows that there is a $(w_1,x_1)$-path $P'$ contained in $G-(X \setminus \{x_1\})-N_G[z^\bot = z]$. Analogously, there is a~$(x_2, w_2)$-path~$P''$ contained in~$F-N_F[z]$. Then the concatenation of~$P'$,~$P[x_1, x_2]$, and~$P''$ is a~$(w_1,w_2)$-path in~$F - N_F[z]$, as desired.
    
    
    \item 
    Suppose $u \in V(H) \setminus X$, then again $u = z$ by the $(G,X,\mathcal{I})$-equivalence and therefore $P$ is a $(v_1,v_2)$-path in $F - N_F[z]$. The construction of the required path is identical to the previous case, but the argument requires one more observation here to show that $V(P')$ avoids $N_F[z]$ as $z^\bot = \bot$. Consider the case that $w_1^\bot \neq \bot$, then by the first obedience condition it follows that there is a $(w_1,x_1)$-path $P'$ contained in $G-(X \setminus \{x_1\})-N_G[z^\bot = \bot]$ (recall that $N_G[\bot] = \emptyset$). Observe that $V(P') \cap X = \{x_1\}$. Since $N_F[z] \cap V(G) \subseteq X$ as $z \in V(H) \setminus X$, and $u=z$ is not adjacent to $x_1$ as $x_1 \in V(P)$ and $P$ is a path that avoids $N_F[u]$, it follows that $P'$ is a $(w_1,x_1)$-path in $F-N_F[z]$. A symmetric argument shows that $P''$ avoids $N_F[z]$. Concatenating these with~$P[x_1,x_2]$ yields a $(w_1,w_2)$-path in $F-N_F[z]$.
  
    \item 
    Finally suppose $u \in V(G) \setminus X$. By the $(G,X,\mathcal{I})$-equivalence, we have $z \in V(G) \setminus X$ and $u=u^\bot$ and $z=z^\bot$ belong to the same set in the partition $\mathcal{J}^\mathcal{I}_X$. Obtain a $(w_1,x_1)$-path $P'$ in $F$ as in the previous case, possibly identical to $P[v_1,x_1]$, and a $(x_2, w_2)$-path $P''$ in~$F$. By the obedience conditions of Definition~\ref{lem:interval:obedience}, these \bmp{can be taken} disjoint from $N_G[z] = N_F[z]$. 
    
    Now it suffices to argue that there is a $(x_1,x_2)$-path disjoint from $N_F[z] = N_G[z]$. We transform $P[x_1,x_2]$ to the desired path, only modifying $\mathcal{P}(P[x_1,x_2],V(G))$. By Observation~\ref{obs:subpath_boundary} it follows that $q_1,q_r \in X$  for each path $Q = (q_1,\dots,q_r) \in \mathcal{P}(P[x_1,x_2],V(G))$. To complete the argument, we show that there is a $(q_1,q_r)$-path in $G-N_G[z]$ for each $Q = (q_1,\dots,q_r) \in \mathcal{P}(P[x_1,x_2],V(G))$.
    Consider $\mathcal{I}(P[x_1,x_2]) = \{I(V(Q)) \mid Q \in \mathcal{P}(P[x_1,x_2],V(G))\}$. By Lemma~\ref{lem:interval:linear-forest} it follows that $\mathcal{I}(P[x_1,x_2])$ is a set of pairwise disjoint intervals. Consider the position of $I(u)$ with respect to the intervals in $\mathcal{I}(P[x_1,x_2])$. Since $P$ is disjoint from $N_F[u] = N_G[u]$, it follows that $I(u)$ does not intersect $I(V(Q))$ for any $Q \in \mathcal{P}(P[x_1,x_2],V(G))$. Suppose there is a path $Q \in \mathcal{P}(P[x_1,x_2],V(G))$ with $I(V(Q)) < I(u)$. Let $Q_\ell = (\ell_1,\dots,\ell_r)$ with $r = |V(Q_\ell)|$ be such that $I(V(Q_\ell)) < I(u)$ and $\rp(V(Q_\ell))$ is maximal. By Definition~\ref{def:signature} we have that $(\ell_1,\ell_r) \in \mathcal{X}$. By the third obedience condition, there is a $(\ell_1,\ell_r)$-path in $G-N_G[z]$. Note that $I(\ell_1) < I(z)$ and $I(\ell_r) < I(z)$ since $u$ and $z$ are in the same set in the partition $\mathcal{J}^\mathcal{I}_X$.
    We argue that for any $Q = (q_1,\dots,q_r) \in \mathcal{P}(P[x_1,x_2],V(G))$ with $I(V(Q)) < I(V(Q_\ell))$, the path $Q$ is disjoint from $N_G[z]$. Suppose not, then there is some $s \in V(Q)$ such that $s$ is adjacent to $z$. But since $I(s) < I(\ell_r) < I(z)$, this is not possible (refer to Figure~\ref{fig:no_long_interval} for an intuition).
    A symmetric argument shows the existence of a $(q_1,q_r)$-path for any $Q = (q_1,\dots,q_r) \in \mathcal{P}(P[x_1,x_2],V(G))$ with $I(u) < I(V(Q))$. \bmp{As each path of~$\mathcal{P}(P[x_1,x_2],V(G))$ can be replaced by a path with the same endpoints that avoids~$N_F[z]$, this yields the desired~$(x_1,x_2)$-path avoiding~$N_F[z]$ since the subpaths outside~$P[x_1,x_2]$ are trivially disjoint from~$N_F[z] = N_G[z] \subseteq V(G)$.} \qedhere
\end{itemize}
\end{proof}

\begin{figure}
    \centering
    \includegraphics[page=2]{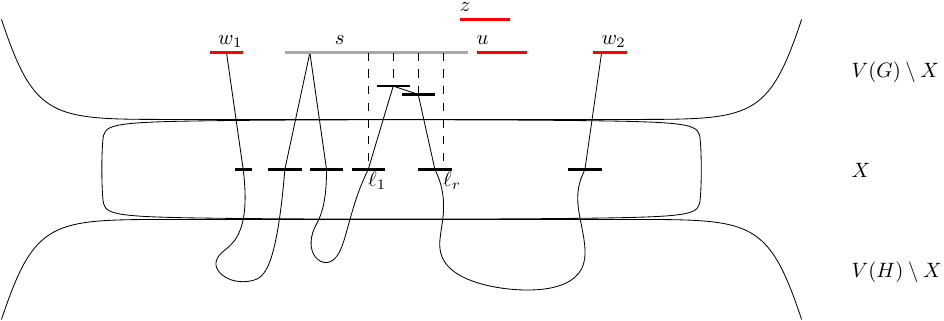}
    \caption{\mic{Illustration of the} last case in Lemma~\ref{lem:interval:path-exists}. By replacing the interval $u$ by $z$, we can only introduce adjacencies to the $(\ell_1,\ell_r)$-subpath of $P$; any \mic{potential \bmp{neighbor} $s \in N_F(z)$ in a path~$P'\in \mathcal{P}(P[x_1,x_2],V(G))$ with~$I(P') < I(P[\ell_1, \ell_r])$} would create a chord \bmp{such as~$s \ell_1$},
    \mic{which is impossible as the path~$P$ is assumed to be chordless.}}
    \label{fig:no_long_interval}
\end{figure}

\mic{We are ready to define the marking scheme and prove that we can always assume that a potential AT uses only the marked vertices.
Since there are only $k^{\Oh(1)}$ signatures in a $k$-boundaried graph $G$, and every AT can be represented by three signatures, the replacement property from \cref{lem:interval:path-exists} allows us to use the same vertices in $G$ for each of $k^{\Oh(1)}$ ``types'' of an AT.}

\begin{lemma}\label{lem:markedAT}
Let $(G,X,\lambda,\mathcal{I})$ be a $k$-boundaried interval graph with a model.
There exists a~set $Q \subseteq V(G) \setminus X$ of $\Oh(k^{24})$ vertices so that for any $k$-boundaried graph $H$ compatible with $(G,X,\lambda)$, if $F = H \oplus (G,X,\lambda)$ contains some AT, then $F$ contains an AT $(w_1,w_2,w_3)$ such that  $\{w_1,w_2,w_3\} \cap V(G) \subseteq X \cup Q$.
\end{lemma}
\begin{proof}

For each triple $(S_1,S_2,S_3)$ of signatures from $\mathcal{S}(G,X,\lambda,\mathcal{I})$, let $O(S_1, S_2, S_3)$ be the set of triples $(v_1,v_2,v_3)$ from $V(G) \cup \{\bot\}$ such that:
\begin{enumerate}
    \item $(v_2,v_3,v_1)$ obeys $S_1$,
    \item $(v_3,v_1,v_2)$ obeys $S_2$,
    \item $(v_1,v_2,v_3)$ obeys $S_3$.
\end{enumerate}
\bmp{Partition the triples~$O(S_1, S_2, S_3)$ into equivalence classes based on the relation of $(G,X,\mathcal{I})$-equivalence among triples.} 
Let $Q(S_1, S_2, S_3)$ contain an arbitrary triple from each class of $(G,X,\mathcal{I})$-equivalence within $O(S_1, S_2, S_3)$, \mic{as long as at least one such triple exists.}
We add to $Q$ every vertex from $V(G) \setminus X$ that appears in any triple in $Q(S_1, S_2, S_3)$ for any $(S_1,S_2,S_3)$.
Observe that $|Q| \in \Oh(k^{24})$ as there are $\Oh(k^{18})$ choices of $(S_1,S_2,S_3)$ and $\Oh(k^6)$ equivalency classes.

\bmp{We now argue that~$Q$ has the claimed property. So consider a $k$-boundaried graph~$H$ compatible with~$(G,X,\lambda)$ such that~$F = H \oplus (G,X,\lambda)$ contains an AT~$(v_1, v_2, v_3)$.} Let $P_1$ be a shortest $(v_2,v_3)$-path in $F - N_F[v_1]$. \jjh{Note that $P_1$ is chordless.}
By \cref{lem:interval:obedience} the triple $(v_2^\bot,v_3^\bot,v_1^\bot)$ obeys the signature $S_1$ of $(P_1, v_1)$.
The same holds for analogously defined $P_2,P_3$ and signatures $S_2,S_3$.
The marking scheme picks some triple $(b_1, b_2, b_3)$ with $b_i \in V(G) \cup \{\bot\}$ for each $i \in [3]$, which is $(G,X,\mathcal{I})$-equivalent to  $(v_1^\bot,v_2^\bot,v_3^\bot)$ and such that $(b_2,b_3,b_1)$ obeys $S_1$, $(b_3,b_1,b_2)$ obeys $S_2$ and $(b_1,b_2,b_3)$ obeys $S_3$.

Let $(w_1,w_2,w_3)$ be defined as follows: for each $i \in [3]$, if $b_i = \bot$ then $w_i = v_i$, otherwise $w_i = b_i$.
\mic{Note that $w_i \in V(F)$ for each $i \in [3]$.}
Then $(w_1,w_2,w_3)$ is $(G,X,\mathcal{I})$-equivalent to $(v_1,v_2,v_3)$ and $(w_1^\bot,w_2^\bot,w_3^\bot) = (b_1,b_2,b_3)$, so these triples behave the same for any signature.
\mic{It follows that $(w_2^\bot,w_3^\bot,w_1^\bot)$ obeys the signature $S_1$ so by
 \cref{lem:interval:path-exists} there exists a $(w_2,w_3)$-path in $F - N_F(w_1)$.
By applying \cref{lem:interval:path-exists} to each of the above orderings of $(w_1,w_2,w_3)$, it follows that $(w_1,w_2,w_3)$ is an AT in $F$.}
By the definition of \jjh{the marked set of vertices} $Q$, $\{w_1,w_2,w_3\} \cap V(G) \subseteq X \cup Q$.
\end{proof}


\paragraph{Wrapping up}
\mic{
\bmp{The last step of the proof is to give a bound on the size of a minimal representative~$(G,X,\lambda)$ in the relation of $(\mathsf{interval},k)$-equivalence, by arguing that in such a graph only few vertices exist outside the set~$Q$ of bounded size.} 
We achieve this by combining the module-free property (Lemma~\ref{lem:no-non-trivial-modules}) and the former contraction property for chordal graphs (Lemma~\ref{lem:meta-chordal:contraction-reverse}).
}

\begin{lemma}\label{lem:meta-interval:representatives}
If the $k$-boundaried graph $(G,X,\lambda)$ is a minimal representative in the relation of $(\mathsf{interval},k)$-equivalence, then $G$ has $\Oh(k^{48})$ vertices.
\end{lemma}

\begin{proof}
By the definition, $G$ is interval. 
Let $\mathcal{I}$ be a normalized interval model of $G$. Let $Q \subseteq V(G)$ be provided by Lemma~\ref{lem:markedAT}. Obtain the partition $\mathcal{J}_{\jjh{X \cup Q}}^\mathcal{I} = \{J_{i,j} \subseteq \jjh{V(G) \setminus (X \cup Q)} \mid i \leq j \in [z]\}$ of $V(G) \setminus (X \cup Q)$ via Definition~\ref{def:interval:regions-new}.

First consider the case that $J_{i,j}$ is an independent set for all $i \leq j \in [z]$. We show that $G$ satisfies the lemma. For each $i \neq j \in [z]$, $G[J_{i,j}]$ induces a clique as the intervals intersect at some region border between the $i$th and $j$th region. Therefore $G[J_{i,j}]$ consists of a single vertex for each $i \neq j \in [z]$. 
We argue that $J_{i,i}$ has size $\Oh(|X \cup Q|)$ for each $i \in [z]$. If any two vertices in $J_{i,i}$ have the same \bmp{open} neighborhood, then they form a non-trivial module in $V(G) \setminus X$, which by Lemma~\ref{lem:no-non-trivial-modules} would contradict that $(G,X,\lambda)$ is a minimal representative. It follows that each pair of vertices in $J_{i,i}$ must have different neighborhoods in $V(G) \setminus (X \cup Q)$ as they have the same neighborhoods in $X \cup Q$ by Definition~\ref{def:interval:regions-new}. Since every vertex in $J_{i,i}$ is adjacent to every vertex in $J_{p,q}$ for each $p < i$ and $i < q$, they can only have neighborhood differences by adjacencies to $J_{p,i}$ or $J_{i,q}$, each of which \bmp{consists} of a single vertex as argued above. Since there are only $\Oh(|X \cup Q|)$ such vertices as $z \in \Oh(|X \cup Q|)$, we have that $J_{i,i}$ has size $\Oh(|X \cup Q|)$. Since $|Q| \in \Oh(|X|^{24})$ by Lemma~\ref{lem:markedAT}, it follows that $G$ at most $|X| + |Q| + \sum_{i \in [z]} |J_{i,i}| + \sum_{i < j \in [z]} |J_{i,j}| \in \Oh(k^{48})$ \bmp{vertices} and the lemma holds.

In the remaining case suppose there is some edge $uv$ for $\{u,v\} \subseteq J_{i,j}$ for some $i \leq j \in [z]$. We argue that $(G,X,\lambda)$ is not a minimal representative. Let $H$ be a $k$-boundaried graph compatible with $(G,X,\lambda)$. Let $F = H \oplus (G,X,\lambda)$ with tri-separation $(A = V(H) \setminus X,X,B = V(G) \setminus X)$. We argue that $F$ is interval if and only if $F / uv $ is interval. Since interval graphs are closed under edge contractions by Observation~\ref{obs:interval:contractions}, we have that if $F$ is interval then so is $F / uv$. 
In the other direction suppose that $F$ is not interval. By Theorem~\ref{thm:interval:chordal:atfree} it follows that $F$ contains a chordless cycle of length at least four or an asteroidal triple. 
Suppose $F$ contains a chordless cycle of length at least four and $F' = F / uv$ does not. Since $F'[A \cup X] = F[A \cup X] = H$ we have that $F[A \cup X]$ is chordal. Since $F[B \cup X] = G$ and $G$ is an interval graph we have that $F[B \cup X]$ is chordal. But then by Lemma~\ref{lem:meta-chordal:contraction-reverse} we have that $F$ is chordal, contradicting that it contains a chordless cycle. It follows that $F / uv$ also contains a chordless cycle and therefore $F / uv$ is not interval.
Finally suppose that $F$ \bmp{is chordal but} contains an asteroidal triple $(a,b,c)$. By Lemma~\ref{lem:markedAT} we can assume that $\{a,b,c\} \cap V(G) \subseteq X \cup Q$. Since $u,v \in J_{i,j}$ have the same neighborhood in $X \cup Q$, it follows that they are adjacent to the same vertices in $\{a,b,c\}$. Therefore any path in $F$ that avoids the closed neighborhood of any vertex of the triple translates to a path in $F / uv$ that avoids the closed neighborhood. It follows that $(a,b,c)$ is an asteroidal triple in $F / uv$ and therefore $F / uv$ is not interval. 
Since $F / uv = H \oplus (G / uv, X, \lambda)$, we have shown that $(G / uv, X, \lambda)$ is in the same $(\mathsf{interval},k)$-equivalence class, contradicting that $(G,X,\lambda)$ is a minimal representative.
\end{proof}

\jjh{
As with \textsc{Chordal deletion} in the previous section, we proceed by showing that we can solve \textsc{Interval deletion} where some vertices become undeletable. As a first step, we observe that interval graphs are closed under addition of true twins, which follows from the fact that the added vertex can get the same interval in an interval model.

\begin{observation}\label{obs:meta-interval:truetwins:stayinterval}
Let $G$ be an interval graph an let $v \in V(G)$. If $G'$ is obtained from $G$ by making a true-twin copy $v'$ of $v$, then $G'$ is an interval graph.
\end{observation}

\textsc{Interval deletion} was shown to be FPT parameterized by \bmp{solution size~$\jjh{k}$} by Cao and Marx. Their algorithm either returns a minimum cardinality solution, or decides that no solution of size at most~$k$ exists.

\begin{thm}[\cite{CaoM2015}]\label{thm:meta-interval:interval_deletion_fpt}
\textsc{Interval deletion} parameterized by the solution size $k$ can be solved in time~\bmp{$10^\jjh{k} \cdot n^{\Oh(1)}$}. 
\end{thm}

\begin{thm}\label{thm:meta-interval:main}
The \textsc{Interval deletion} problem can be solved in time $2^{\Oh(k^{96})}\cdot n^{\Oh(1)}$ when given a tree $\mathsf{interval}$-\bmp{decomposition} of width $k-1$ consisting of $n^{\Oh(1)}$ nodes.
\end{thm}
\begin{proof}
We check the conditions of Theorem~\ref{thm:meta-uniform:main}.
The class of interval graphs is clearly closed under vertex deletion and disjoint union of graphs.
Because of Theorem~\ref{thm:meta-interval:interval_deletion_fpt} and Observation~\ref{obs:meta-interval:truetwins:stayinterval},
\mic{we can apply Lemma~\ref{lem:meta-uniform:undeletable} to solve \textsc{Disjoint interval deletion}
in time $10^s\cdot n^{\Oh(1)}$.}
Next, by Lemma~\ref{lem:meta-interval:representatives}
\mic{and \cref{lem:representative-generation-general},
an $(\mathsf{interval},\le k)$-representative family can be computed in time $v(k) = 2^{\Oh(k^{\jjh{96}})}$.}
\end{proof}
}

\mic{

\begin{corollary}\label{thm:meta-interval:final}
The \textsc{Interval deletion} problem can be solved in time $2^{\Oh(k^{\jjh{480}})}\cdot n^{\Oh(1)}$ when parameterized by $k = \hhtw[interval](G)$.
\end{corollary}
\begin{proof}
We use \cref{thm:decomposition:full} to find a tree $\mathsf{interval}$-decomposition of width $\Oh(k^5)$, {which takes time $2^{\Oh(k \log k)} \cdot n^{\Oh(1)}$,} and plug it into Theorem~\ref{thm:meta-interval:main}.
\end{proof}
}

\section{Hardness proofs}\label{sec:hardness}
In this section we provide two hardness proofs that mark boundaries of applicability of our approach. In Section~\ref{subsec:fpt:inapprox} we show that for~$\hh$ the class of perfect graphs, no FPT-approximation algorithms exist for computing~$\hhtw$ or~$\hhdepth$ unless W[1]~$=$~FPT. In Section~\ref{subsec:not:closed} we show why the parameterization by elimination distance is only fruitful for vertex-deletion problems to graph classes~$\hh$ which are closed under disjoint union, by showing that for~$\hh$ not closed under disjoint union, \textsc{$\hh$-deletion} can be NP-complete even for graphs of elimination distance~$0$ to a member of~$\hh$.

\subsection{No FPT-approximation for perfect depth and width} \label{subsec:fpt:inapprox}
Let $\mathsf{perfect}$ be shorthand for the class of perfect graphs.
Gr{\"o}tschel et al.~\cite{GrotschelLS81} show that computing a minimum vertex cover in a perfect graph can be done in polynomial time. {Given Theorems~\ref{thm:vcelim} and~\ref{thm:vc:width}}, computing a $\mathsf{perfect}$-elimination forest (or a tree $\mathsf{perfect}$-decomposition) would be of interest. However, we show that we cannot hope to approximate such decompositions {in FPT time}.


{For a function~$g \colon \mathbb{N} \to \mathbb{N}$, a fixed-parameter tractable $g$-approximation algorithm for a parameterized minimization problem is an algorithm that, given an instance of size~$n$ with parameter value~$k$, runs in time~$f(k)\cdot n^{\Oh(1)}$ for some computable function~$f$ and either determines that there is no solution of size~$k$, or outputs a solution of size at most~$g(k) \cdot k$.}

In the \textsc{$k$-hitting set} problem, we are given a universe $U$, a set system $\mathcal{F}$ {over~$U$}, and an integer $k$, and ask for a smallest cardinality set $S \subseteq U$ such that {$F \cap S \neq \emptyset$ for all $F \in \mathcal{F}$}. As shown by Karthik et al.~\cite{KarthikLM19}, we have the following inapproximability result for \textsc{$k$-hitting set}.

\begin{thm}[{\cite{KarthikLM19}}]\label{thm:hitsetinapprox}
Assuming W[1] $\neq$ FPT, {there is no FPT {$g(k)$}-approximation for  \textsc{$k$-hitting set} for any computable function $g$.}
\end{thm}

The decision variant of the \textsc{$k$-hitting set} problem, which asks if a solution of size at most $k$ exists, is known to be W[2]-hard. In the \textsc{$k$-perfect deletion} problem, we are given a graph $G$ and an integer $k$, and ask for a smallest cardinality set $S \subseteq V(G)$ such that $G-S$ is a perfect graph. 
By the Strong Perfect Graph Theorem~\cite{ChudnovskyRST06}, this amounts to deleting odd induced cycles of length at least 5 (odd {holes}) and their edge complements (odd {anti-holes}) from the graph. 
Heggernes et al.~\cite{HeggernesHJKV13} show that the decision variant of \textsc{$k$-perfect deletion} is W[2]-hard, reducing from \textsc{$k$-hitting set}. We show that their reduction also rules out good approximations for $\hhdepth[\mathsf{perfect}]$ and $\hhtw[\mathsf{perfect}]$.
Let \textsc{$k$-perfect depth} be the problem of computing a minimum depth $\mathsf{perfect}$-elimination forest. Similarly, let and \textsc{$k$-perfect width} be the problem of computing a minimum width tree $\mathsf{perfect}$-decomposition. We show that we cannot $g(k)$-approximate these problems in FPT time.

\begin{lemma} \label{lemma:inapprox:perfect}
Let~$g \colon \mathbb{N} \to \mathbb{N}$ be a computable function. Assuming W[1] $\neq$ FPT, there is no algorithm that, given a graph~$G$ and integer~$k$, runs in time~$f(k) \cdot n^{\Oh(1)}$ for some computable function~$f$ and either determines that the minimum size of a perfect deletion set in~$G$ is larger than~$k$, or outputs a tree $\mathsf{perfect}$-decomposition of~$G$ of width~$g(k) \cdot k$. 
\end{lemma}
\begin{proof}
{Suppose that an algorithm~$\mathcal{A}$ as described in the lemma statement does exist. We will use it to build an FPT-time $2 \cdot g(k)$-approximation for \textsc{$k$-hitting set}, thereby showing W[1]~$=$~FPT by Theorem~\ref{thm:hitsetinapprox}. By contraposition, this will prove the lemma.}

We use the construction of~\cite{HeggernesHJKV13} to reduce an instance of \textsc{$k$-hitting set} to \textsc{$k$-perfect deletion}. Let $(U,\mathcal{F},k)$ be an instance of \textsc{Hitting set}.
Assume $|F| \geq 2$ for $F \in \mathcal{F}$. Build an instance $(G,k)$ for \textsc{$k$-perfect deletion} as follows. 
\begin{itemize}
    \item Create an independent set $X$ on $|U|$ vertices, let $X = \{v_u \mid u \in U\}$.
    \item For each set $F = \{u_1,\ldots,u_t\} \in \mathcal{F}$, add $|F|+1$ new vertices $h_1,\ldots,h_{t+1}$. The set $\mathcal{G}_F = \{h_1,\ldots,h_{t+1}\}$ is the \emph{set gadget} for $F$. Let $U_F \subseteq U = \{v_{u_1},\ldots,v_{u_t}\}$ be the universe vertices corresponding to $F$. Add edges $h_1v_{u_1}$, $v_{u_1}h_2$, $h_2v_{u_2}$,\ldots,$v_{u_t}h_{t+1}$, $h_{t+1}h_1$. Note that $G[\mathcal{G}_F \cup U_F]$ induces an odd hole of length at least 5.
    \item Take the pairwise join of the set gadgets. {That is, make all vertices of~$\mathcal{G}_F$ adjacent to all vertices of~$\mathcal{G}_{F'}$ for distinct~$F, F' \in \mathcal{F}$.}
\end{itemize}

The following series of claims are proven {by Heggernes et al.~\cite[{Claim 1--4, Thm.~1}]{HeggernesHJKV13}.}

\begin{enumerate}
    \item The graph $G-X$ is a cograph and therefore perfect.
    \item Any hole in $G$ intersects $X$ and exactly one set gadget $\mathcal{G}_F$.
    \item Any anti-hole in $G$ has length 5 and is therefore a hole of length 5.
    \item If $S \subseteq V(G)$ such that $G-S$ is perfect, then there is $S' \subseteq X$ with $|S'| \leq |S|$ such that $G-S'$ is perfect.
    \item {For each~$\ell \geq 0$}, the instance $(U,\mathcal{F})$ has a solution of size at most $\ell$ if and only if $G$ has {a~$\mathsf{perfect}$-deletion set} of size at most $\ell$. \label{prop:perfect:preserve:approx}
\end{enumerate}

We note that the construction above is independent of $k$ and approximation preserving due to {Property~\ref{prop:perfect:preserve:approx}}. The reduction can easily be computed in polynomial time.

{The algorithm for \textsc{$k$-hitting set} constructs the instance~$(G,k)$ from~$(U,\mathcal{F},k)$, and then applies~$\mathcal{A}$ to~$(G,k)$.} If~$\mathcal{A}$ concludes that the minimum size of a perfect deletion set in $G$ is larger than $k$, then by Property 5 we can conclude that a minimum hitting set is larger than $k${, and reject the instance.}

Otherwise, let $(T,\chi,L)$ be the resulting tree $\mathsf{perfect}$-decomposition of width $d \leq g(k) \cdot k$ output by the algorithm. {To transform the decomposition of~$G$ into a hitting set, we distinguish two cases.}

First, suppose not a single gadget vertex $h$ is contained in a base component, that is, $\mathcal{G}_F \cap L = \emptyset$ for each set gadget $\mathcal{G}_F$. Let $S \subseteq V(G)$ be a set that contains a single arbitrary vertex from each {set gadget~$\mathcal{G}_F$}. {Then~$G[S]$ induces a clique by construction, and therefore there is a single bag~$\chi(t)$ for some~$t \in V(T)$ which contains all vertices of~$S$. {Since no gadget vertex is in a base component and so $S \cap L = \emptyset$,} we have~$|\chi(t) \setminus L| \geq |S|$.} It follows that $d \geq |S| - 1 \geq |\mathcal{F}| - 1$. Hence we can simply take an arbitrary universe element from {every set}, and get a hitting set of size at most {$d + 1 \leq 2 \cdot g(k) \cdot k$}.

If the previous case does not apply, then there exists {$F \in \mathcal{F}$} such that some $h \in \mathcal{G}_F$ is contained in $\chi(t) \cap L$ for a leaf $t \in V(T)$. We construct a hitting set $S$ for $(U,\mathcal{F})$ as follows. For each $v \in \chi(t) \setminus L$, if {$v=v_u \in X$} then add $u$ to $S$. Otherwise $v \in \mathcal{G}_F$ is a gadget vertex and we add $u$ to $S$ for an arbitrary $u \in F$. 
To see that $S$ is a valid hitting set, consider an arbitrary $F' \in \mathcal{F}$. Then either $h \in \mathcal{G}_{F'}$ or $h$ is adjacent to every vertex of $\mathcal{G}_{F'}$ by the join step {in the construction of~$G$}. Since $G[\mathcal{G}_{F'} \cup U_{F'}]$ induces an odd hole, at least one of its vertices {$u$} is outside the base component as the base component is perfect. Consider a shortest path $P$ from $h$ to $u$ in $G[\mathcal{G}_{F'} \cup U_{F'} \cup \{h\}]$. Note that such a path must exist by the previous observation. Let $w$ be the first vertex of $P$ such that $w \notin L$. Note that $w \in \chi(t) \setminus L$ by {Observation~\ref{obs:basecomponent:neighborhoods}}. Now since $w$ is either a universe vertex or a gadget vertex, it follows that $S$ contains an element that hits $F'$.

In both cases we have constructed a hitting set of size at most {$d + 1 \leq 2 \cdot g(k) \cdot k$}. But now we have an algorithm that either concludes that a minimum hitting set has size larger than $k$, or gives a hitting set of size at most $2 \cdot g(k) \cdot k$ in FPT time. Using Theorem~\ref{thm:hitsetinapprox}, this implies W[1]~$=$~FPT.
\end{proof}

Lemma~\ref{lemma:inapprox:perfect} leads to FPT-inapproximability results for \textsc{$k$-perfect depth} and \textsc{$k$-perfect width}, using the fact a small perfect deletion set trivially gives a decomposition of small depth or width.

\begin{thm}\label{thm:inapprox:perfectdepthwidth}
Assuming W[1] $\neq$ FPT, there is no FPT time $g$-approximation for \textsc{$k$-perfect depth} {or} \textsc{$k$-perfect width} for any computable function $g$.
\end{thm}
\begin{proof}
{We show that an FPT-approximation for either problem would give an algorithm satisfying the conditions of Lemma~\ref{lemma:inapprox:perfect} and therefore imply W[1]~$=$~FPT.}

Suppose there is an FPT time $g$-approximation for \textsc{$k$-perfect width} for some computable function $g$. If{, on input~$(G,k)$,} the algorithm decides that there is no tree $\mathsf{perfect}$-decomposition of width $k$, then we can also conclude that the minimum size of a perfect deletion set {of~$G$} is larger than $k$. {Otherwise, it provides a decomposition of width at most~$g(k) \cdot k$.}

Now suppose there is an FPT time $g$-approximation for \textsc{$k$-perfect depth} for some computable function $g$. If{, on input~$(G,k)$,} the algorithm decides that there is no $\mathsf{perfect}$-elimination forest of depth $k$, then we can also conclude that the minimum size of a perfect deletion set of the input graph is larger than $k$. If the algorithm returns a $\mathsf{perfect}$-elimination forest of depth at most $g(k) \cdot k$, then using Lemma~\ref{lem:treedepth-treewidth} we can construct a tree $\mathsf{perfect}$-decomposition of width at most $g(k) \cdot k - 1$ in polynomial time.
\end{proof}

{In the above FPT-inapproximability results for \textsc{$k$-perfect width} and \textsc{$k$-perfect depth}, we strongly rely on the fact that instances constructed in the W[2]-hardness reduction from \textsc{$k$-hitting set} to \textsc{$k$-perfect deletion} are very dense. While \textsc{$k$-wheel-free deletion} is also W[2]-hard~\cite{Lokshtanov08}, which is shown via an approximation-preserving reduction, this does not directly lead to FPT-inapproximability of \textsc{$k$-wheel-free width} and \textsc{$k$-wheel-free depth}, as the instances produced by that reduction are much less dense. In particular, in such instances there is no direct way to translate a graph decomposition into a wheel-free deletion set. We therefore do not know whether \textsc{$k$-wheel-free width} and \textsc{$k$-wheel-free depth} admit FPT-approximations.}

\subsection{Classes which are not closed under disjoint union} \label{subsec:not:closed}

In this section we show why the restriction of Theorem~\ref{thm:solving:general} to graph classes which are closed under disjoint unions is necessary. Recall that~$K_5$ is a clique on five vertices and~$K_{1,3}$ is the claw. Let~$H_{5+1,3} := K_5 + K_{1,3}$ denote the disjoint union of these two graphs. Let~$\hh_{5+1,3}$ denote the (hereditary) family of graphs which do not contain~$H_{5+1,3}$ as an induced subgraph.

For our hardness proof, we will use the following result of Lewis and Yannakakis. Here a graph property is nontrivial in the class of planar graphs if there are infinitely many planar graphs that have the property, and infinitely many planar graphs which do not.

\begin{thm}[{\cite[Cor.~5]{LewisY80}}]
The node-deletion problem restricted to planar graphs  for graph-properties that are hereditary on induced  subgraphs and nontrivial on planar graphs is NP-complete.
\end{thm}

In particular, their result shows that \textsc{Induced-$K_{1,3}$-free deletion} is NP-complete when restricted to planar graphs.

\begin{thm}
\textsc{$\hh_{5+1,3}$ deletion} is NP-complete when restricted to graphs whose elimination distance to~$\hh_{5+1,3}$ is~$0$.
\end{thm}
\begin{proof}
We give a reduction from an instance~$(G,k)$ of \textsc{Induced-$K_{1,3}$-free deletion} on planar graphs, which asks whether the planar graph~$G$ can be made (induced) claw-free by removing at most~$k$ vertices.

Let~$G'$ be the disjoint union of~$G$ with~$k+1$ copies of~$K_5$. Observe that the elimination distance of~$G'$ to~$\hh_{5+1,3}$ is~$0$: each connected component of~$G'$ is either isomorphic to~$K_5$ (and does not contain~$K_{1,3}$) or consists of the planar graph~$G$ (and therefore does not contain~$K_5$). Hence each connected component of~$G'$ belongs to~$\hh_{5+1,3}$.

We claim that the instance~$(G',k)$ of \textsc{$\hh_{5+1,3}$-deletion} is equivalent to the instance~$(G,k)$ of claw-free deletion. In one direction, any vertex set~$S \subseteq V(G)$ for which~$G - S$ is claw-free also ensures that~$G' - S$ is claw-free (as the components isomorphic to~$K_5$ do not contain any claws) and therefore implies~$G' - S \in \hh_{5+1,3}$. For the reverse direction, suppose~$S' \subseteq V(G')$ is a set of size at most~$k$ such that~$G' - S' \in \hh_{5+1,3}$. Then~$G' - S'$ must be claw-free, because if~$G' - S'$ contains a claw, then this claw forms an induced copy of~$H_{5+1,3}$ in~$G' - S'$ together with one of the~$k+1$ copies of~$K_5$ that contains no vertex of~$S'$. So~$G' - S'$ is claw-free, and therefore the induced subgraph~$G - S'$ is claw-free as well, showing that~$(G,k)$ is a yes-instance.

As the transformation can easily be performed in polynomial time, this completes the proof.
\end{proof}

We remark that similar hardness proofs can be obtained for graph classes defined by a disconnected forbidden minor rather than a disconnected forbidden induced subgraph. For example, when~$H'$ is the disjoint union of the cycle~$C_9$ on nine vertices and the graph~$K_5$, one can show that \textsc{Induced $H'$-minor-free deletion} is NP-complete on graphs whose elimination distance to an~$H'$-minor-free graph is~$0$. This can be seen from the fact that \textsc{Feedback vertex set} remains NP-complete on planar graphs of girth at least nine since subdividing edges does not change the answer, that planar graphs do not contain~$K_5$ as a minor, while~$K_5$ does not contain~$C_9$ as a minor. As the purpose of these lower bounds is merely to justify our restriction to graph classes closed under disjoint union, we omit further details.

\section{Conclusion}
\label{sec:conclusions}
We introduced a new algorithmic framework for developing algorithms that solve \textsc{$\hh$-deletion} problems. Our algorithms simultaneously exploit small separators as well as structural properties of~$\hh$, leading to fixed-parameter tractable algorithms for parameterizations by~$\hhdepth$ and~$\hhtw$ which can be arbitrarily much smaller than either the treewidth or the solution size. To obtain these algorithms, we showed that~$\hhdepth$ and~$\hhtw$ can be FPT-approximated using subroutines to compute~$(\hh,k)$-separations and used a number of different tools from algorithmic graph theory to compute separations. Our work opens up a multitude of directions for future work.

\subparagraph{Beyond undirected graphs} On a conceptual level, the idea of solving a deletion problem parameterized by the elimination distance (or even less) is not restricted to undirected graphs. By developing  notions of elimination distance for directed graphs, hypergraphs, or other discrete structures, similar questions could be pursued in those contexts. One could also consider undirected graphs with a distinguished set of terminal vertices, for example in an attempt to develop (uniform, single-exponential) FPT algorithms for \textsc{Multiway cut} parameterized by the elimination distance to a graph where each component has at most one terminal.

\subparagraph{Cross-parameterizations} In this paper, our main focus was on solving \textsc{$\hh$-deletion} parameterized by~$\hhdepth$ or~$\hhtw$. However, the elimination distance can also be used as a parameterization away from triviality for solving other parameterized problems~$\Pi$, when using classes~$\hh$ in which~$\Pi$ is polynomial-time solvable. This can lead to interesting challenges of exploiting graph structure. For problems which are FPT parameterized by deletion distance to~$\hh$, does the tractability extend to elimination distance to~$\hh$? For example, is \textsc{Undirected feedback vertex set} FPT when parameterized by the elimination distance to a subcubic graph or to a chordal graph? The problem is known to be FPT parameterized by the deletion distance to a chordal graph~\cite{JansenRV14} or the edge-deletion distance to a subcubic graph~\cite{MajumdarR18}.

As a step in this direction, Eiben et al.~\cite[{Thm.~4}]{EibenGHK19} present a meta-theorem that yields non-uniform FPT algorithms when~$\Pi$ satisfies several conditions, which require a technical generalization of an FPT algorithm for $\Pi$ parameterized by deletion distance to $\hh$. 

\subparagraph{Improving approximation guarantees} The FPT algorithms we developed to approximate~$\hhdepth$ and~$\hhtw$ output decompositions whose depth (or width) is polynomially bounded in the optimal value. An obvious question is whether better approximation guarantees can be obtained with similar running time bounds. Exact FPT algorithms are known for several classes~$\hh$~\cite{AgrawalKP0021,jansen2021fpt}, including all minor-closed classes~\cite{BulianD17}, but the current-best running times are not of the form~$2^{k^{\Oh(1)}} \cdot n^{\Oh(1)}$. \bmpr{Update references with recent work.} 

\subparagraph{Improving running times} \bmp{The running times we obtained for solving \textsc{$\hh$-deletion} parameterized by~$k = \hhdepth$ or~$k = \hhtw$ are typically of the form~$2^{k^{\Oh(1)}} \cdot n^{\Oh(1)}$, where the polynomial in the exponent has a larger degree than in the best-possible algorithms for the parameterizations by solution size or treewidth. One way to speed up these algorithms would be to improve the approximation ratios of the decomposition algorithms; one could also investigate whether the structural properties of bounded $\hhtw$ or $\hhdepth$ can be exploited without precomputing a decomposition. Furthermore, for several \textsc{$\hh$-deletion} problems we incur additional overhead when using a decomposition to find a deletion set; can this be avoided?

Concretely, it would be interesting to determine whether the running times for these hybrid parameterizations can match the best-known running times for the parameterizations by solution size and treewidth. For example, can \textsc{Odd Cycle Transversal} be solved in time~$3^k \cdot n^{\Oh(1)}$ when parameterized by the elimination distance~$k$ to a bipartite graph, or can such algorithms be ruled out by lower bounds based on the (Strong) Exponential Time Hypothesis?} 

\subparagraph{Capturing more vertex-deletion problems} There are some graph classes for which  Theorem~\ref{thm:decomposition:general} provides decomposition algorithms, but for which we currently have no follow-up algorithm to solve \textsc{$\hh$-deletion} on a given decomposition. A natural target for future work is to see whether the decompositions can be turned into vertex-deletion algorithms. For example, if~$\hh$ is characterized by a finite set of connected forbidden \bmp{topological} minors, is~\textsc{$\hh$-deletion} FPT parameterized by~$\hhtw$? 

\bibliographystyle{plainurl}
\bibliography{main}

\clearpage

\appendix

\end{document}

%% file: roadmap.tex
\tikzset{every picture/.style={line width=0.75pt}} 
\centering
\resizebox{.9\textwidth}{!}{
\begin{tikzpicture}
    \begin{scope}[every node/.style={dashed,draw,fill=gray!10}]
        \node (class-other) at (9, 0.5) {remaining graph classes};
        \node (class-bip) at (-1, 1) {bipartite};
        \node (class-subgraphs) at (6, 2) {classes with forbidden (induced) subgraphs};
    \end{scope}
    
    \begin{scope}[every node/.style={rectangle,thick,draw}]
        \node (separation) at (1,-1) {$(\mathcal{H} ,h)$-\textsc{separation finding}};
        \node (restricted-separation) at (10,-1) {\textsc{Restricted} $(\mathcal{H} ,h)$-\textsc{separation finding}};
        \node (separation-decomposition) at (0.5,-2.5) {{separation decomposition}};
        \node (restricted-separation-decomposition) at (10.5,-2.5) {{restricted separation decomposition}};
        \node (td-quotient) at (0,-4) {{elimination forest of the separation quotient graph}};
        \node (tw-quotient) at (11,-4) {{tree decomposition of the separation quotient graph}};
        \node (td-final) at (-0.5,-5.5) {$\hh$-elimination forest};
        \node[text width=3.5cm] (td-final-space) at (5,-5.75) {$\hh$-elimination forest in~polynomial space};
        \node (tw-final) at (11.5,-5.5) {tree $\hh$-decomposition};
    \end{scope}

    \begin{scope}[every node/.style={},
                  every edge/.style={draw=black, thick}]
        \path [->] (class-other) edge[bend right=5] node[xshift=0.2cm, yshift=0.5cm] {Lemma \ref{lem:branching:general}} (separation);
        \path [->] (class-bip) edge node[xshift=-1.3cm, yshift=0.0cm] {Lemma \ref{lem:separation-bipartite}} (separation);
        \path [->] (class-subgraphs) edge[bend right=10] node[xshift=-1.1cm, yshift=0.2cm] {Lemma \ref{lem:separation-subgraph}} (separation);
        
        \path [->] (separation) edge node[below] {Lemma \ref{lem:restricted-to-unrestricted}} (restricted-separation);
        \path [->] (separation) edge node[left] {Lemma \ref{lem:crown-decomposition}} (separation-decomposition);
        \path [->] (restricted-separation) edge node[left] {Lemma \ref{lem:restricted-decomposition}} (restricted-separation-decomposition);
        \path [->] (separation-decomposition) edge node[left] {Lemma \ref{lem:quotient-depth}} (td-quotient);
        \path [->] (restricted-separation-decomposition) edge node[left] {Lemma \ref{lem:quotient-width}} (tw-quotient);
        \path [->] (td-quotient) edge node[left] {Lemma \ref{lem:decomp-ed-exact}} (td-final);
        \path [->] (td-quotient) edge node[xshift=1.7cm] {Lemma \ref{lem:decomp-ed-polyspace}} (td-final-space);
        \path [->] (tw-quotient) edge node[left] {Lemma \ref{lem:decomp-tw}} (tw-final);
    \end{scope}
\end{tikzpicture}}